%% file: main.tex
\documentclass[11pt]{amsart}
\usepackage[T1]{fontenc}
\usepackage{amsmath,amssymb,amsthm}
\usepackage{enumitem}
\usepackage[margin=1.15in]{geometry}
\usepackage[colorlinks=true,linkcolor=blue,citecolor=blue,urlcolor=blue]{hyperref}

\numberwithin{equation}{section}
\theoremstyle{plain}
\newtheorem{theorem}{Theorem}[section]
\newtheorem{proposition}[theorem]{Proposition}
\newtheorem{lemma}[theorem]{Lemma}
\newtheorem{corollary}[theorem]{Corollary}
\theoremstyle{definition}

\newtheorem{assumption}[theorem]{Assumption}
\theoremstyle{remark}
\newtheorem{remark}[theorem]{Remark}

\newcommand{\R}{\mathbb{R}}
\newcommand{\C}{\mathbb{C}}
\newcommand{\N}{\mathbb{N}}
\newcommand{\Z}{\mathbb{Z}}
\newcommand{\T}{\mathbb{T}}
\newcommand{\Tr}{\operatorname{Tr}}
\newcommand{\tr}{\operatorname{tr}}
\newcommand{\supp}{\operatorname{supp}}
\newcommand{\dist}{\operatorname{dist}}
\newcommand{\curl}{\operatorname{curl}}
\newcommand{\Div}{\operatorname{div}}
\newcommand{\Ran}{\operatorname{Ran}}

\newcommand{\sinc}{\operatorname{sinc}}
\newcommand{\eps}{\varepsilon}
\newcommand{\Hf}{H_{\mathrm f}}
\newcommand{\ELL}{E_{\mathrm{LL}}}
\newcommand{\rhom}{\rho_{\mathrm m}}
\newcommand{\Efield}{E_{\mathrm{field}}}

\newcommand{\ip}[2]{\langle #1,#2\rangle}
\newcommand{\norm}[1]{\lVert #1\rVert}
\newcommand{\abs}[1]{\lvert #1\rvert}
\newcommand{\doi}[1]{\url{https://doi.org/#1}}
\newcommand{\arxiv}[1]{\href{https://arxiv.org/abs/#1}{arXiv:#1}}

\begin{document}

\title[Self-energy of the Pauli--Fierz model]{The sharp leading asymptotics of the self-energy\\ of the spinless Pauli--Fierz model}

\author{Shinichiro Futakuchi}
\address{Mugen Giken LLC, Tokyo, Japan}
\email{futakuchi0117@gmail.com}
\urladdr{https://orcid.org/0009-0000-5027-0517}
\thanks{Corresponding author: Shinichiro Futakuchi, Mugen Giken LLC, Tokyo, Japan; futakuchi0117@gmail.com; ORCID 0009-0000-5027-0517.}


\begin{abstract}
We study the self-energy $E_0(\alpha,\Lambda)$ of the spinless Pauli--Fierz model with a
sharp ultraviolet cutoff $\Lambda$, retaining the vacuum contribution of the $A^2$ term (no
normal ordering); $\alpha=e^2/4\pi$ is the fine structure constant. Lieb and Loss proved lower
and upper bounds of orders $\alpha^{1/2}\Lambda^{3/2}$ and $\alpha^{2/7}\Lambda^{12/7}$ for large
$\Lambda$ and suggested that the upper bound, which comes from product states, gives the
correct order. We prove that, for every fixed $\alpha>0$,
$E_0(\alpha,\Lambda)=(1+o(1))\,c_\alpha\,\alpha^{2/7}\Lambda^{12/7}$, with an explicit
constant $c_\alpha=0.99338\ldots$, as $\Lambda\to\infty$. The same asymptotics hold uniformly in the joint limit
$\Lambda/\alpha\to\infty$, $\alpha^{5/2}\Lambda\to\infty$, provided
$\log(2+\alpha)=o((\Lambda/\alpha)^{1/21})$. The infimum of the energy over product states has the
same leading asymptotics. The upper bound follows from a coherent state bound for the field
energy in a fixed electron density. For the lower bound we combine a periodization of the
field, a blockwise Born--Oppenheimer bound, and a lower bound on the magnetic Neumann
Laplacian on a cube for band-limited fields. The magnetic estimate follows from a
non-concentration bound for the vorticity of band-limited fields satisfying a weighted
gradient condition.
\end{abstract}

\maketitle

\input{sec1_intro}
\input{sec2_strategy}
\input{sec3_model}

\input{sec4_upper}
\input{sec5_periodization}
\input{sec6_blockBO}
\input{sec7_neumann}
\input{sec8_proof}
\input{sec9_remarks}

\appendix
\input{appA_BH}
\input{appB_constants}
\input{appD_counterexample}

\input{backmatter}

\input{refs}
\end{document}

%% file: sec1_intro.tex
\section{Introduction}\label{sec:intro}

\subsection{The problem}
The Pauli--Fierz model \cite{PF38} describes a non-relativistic charged particle
coupled to the quantized radiation field; see \cite{Spo04} for a monograph account
and \cite{Ara18} for the Fock space framework we use. For a single spinless electron
of unit mass and charge $e$, in units $\hbar=c=1$, with a sharp ultraviolet cutoff
$\Lambda>0$ on the photon momenta, the Hamiltonian is
\begin{equation}\label{eq:H-intro}
  H=\tfrac12\bigl(p+e A_\Lambda(x)\bigr)^2+\Hf\qquad\text{on }L^2(\R^3)\otimes\mathcal F,
  \qquad \alpha:=\frac{e^2}{4\pi},
\end{equation}
where $A_\Lambda$ is the quantized vector potential in the Coulomb gauge with the
canonical normalization \eqref{eq:A} below, $\Hf$ is the field energy, and $\alpha$ is
the fine structure constant (physical value $1/137.036$). We are interested in the
self-energy $E_0(\alpha,\Lambda):=\inf\sigma(H)$ as $\Lambda\to\infty$. The term $\frac{e^2}{2}A_\Lambda^2$ is not normal ordered; normal ordering would subtract the constant $\frac{e^2}{2}\langle\Omega,A_\Lambda(0)^2\Omega\rangle=\alpha\Lambda^2/(2\pi)$, which for fixed $\alpha$ dominates the leading term of Theorem~\ref{thm:main}.

First order perturbation theory gives
$E_0\approx\frac{e^2}{2}\langle\Omega,A_\Lambda(0)^2\Omega\rangle=\alpha\Lambda^2/(2\pi)$.
Lieb and Loss \cite{LL00} observed that this is not the correct behaviour for large
$\Lambda$ at fixed $\alpha$, and proved
\begin{equation}\label{eq:LL-bounds}
  C_1\,\alpha^{1/2}\Lambda^{3/2}\ \le\ E_0(\alpha,\Lambda)\ \le\ C_2\,\alpha^{2/7}\Lambda^{12/7}.
\end{equation}
The upper bound comes from product states $\phi\otimes\Psi$. If the electron is
localized on a length scale $R$, each field mode is stiffened by the term
$\frac{e^2}{2}\int\abs{\phi}^2\abs{A}^2$; the resulting increase of the zero-point
energy of the field is of order $e\Lambda^3\int\abs{\phi}$, and balancing this against the
kinetic energy $R^{-2}$ gives $\alpha^{2/7}\Lambda^{12/7}$. In the language of
density functional theory this is a local density approximation. Lieb and Loss
suggested that $\Lambda^{12/7}$ is the correct order. In support they showed
\cite[Theorem~1.2]{LL00} that if the paramagnetic term $p\cdot A$ is dropped from
\eqref{eq:H-intro}, then both bounds are of order $\alpha^{2/7}\Lambda^{12/7}$. They also
noted that, by the upper bound in \eqref{eq:LL-bounds}, perturbation theory can be
valid at most for a range of $\alpha$ shrinking at least like $\Lambda^{-2/5}$.

The remaining question is whether electron--field correlations lower the energy at
leading order. The difficulty is the paramagnetic term, which
vanishes in product states with real $\phi$ but not in general.

\subsection{Main result}
Let $\ELL(\alpha,\Lambda)$ denote the infimum of $\ip{\phi\otimes\Psi}{H\,\phi\otimes\Psi}$
over normalized product states, so that $E_0\le\ELL$. Define
\begin{equation}\label{eq:F1}
  F[1]:=\inf\Bigl\{\tfrac12\norm{\nabla\phi}_2^2+\norm{\phi}_1:\ \phi\in H^1\cap L^1(\R^3),\ \norm{\phi}_2=1\Bigr\}.
\end{equation}
This functional was introduced by Bach and Hach \cite{BH22}; its minimizer is explicit
(Proposition~\ref{prop:F1} in Appendix~\ref{app:constants}) and $F[1]=4.96461575\ldots$. Set
\begin{equation}\label{eq:calpha}
  c_\alpha:=(6\pi^2)^{-4/7}(4\pi)^{2/7}F[1]=0.99338240\ldots,
  \qquad \ell:=\Bigl(\frac{\Lambda}{\pi e^2}\Bigr)^{1/7},\qquad \Xi:=e^2\ell^2 .
\end{equation}
Note that $\Xi=\pi^{-2/7}(e^5\Lambda)^{2/7}$, so $\Xi\to\infty$ if and only if
$\alpha^{5/2}\Lambda\to\infty$, and $\ell\to\infty$ if and only if $\Lambda/\alpha\to\infty$.

\begin{theorem}\label{thm:main}
\begin{enumerate}[label=(\alph*)]
\item (Upper bound.) There is a function $\eps_+$ on $[1,\infty)$ with $\eps_+(\ell)\to0$ as $\ell\to\infty$ such
that for all $\alpha>0$ and $\Lambda>0$ with $\ell\ge1$,
\[
  \frac{E_0(\alpha,\Lambda)}{c_\alpha\alpha^{2/7}\Lambda^{12/7}}\le
  \frac{\ELL(\alpha,\Lambda)}{c_\alpha\alpha^{2/7}\Lambda^{12/7}}\le 1+\eps_+(\ell).
\]
\item (Lower bound.) For every $\eps\in(0,1)$ there are $\ell_0,\Xi_0<\infty$ and $\kappa_0>0$ such that for all
$\alpha>0$ and $\Lambda>0$ with $\ell\ge\ell_0$, $\Xi\ge\Xi_0$ and $\log(2+e^2)\le\kappa_0\ell^{1/3}$,
\[
  \frac{\ELL(\alpha,\Lambda)}{c_\alpha\alpha^{2/7}\Lambda^{12/7}}\ge
  \frac{E_0(\alpha,\Lambda)}{c_\alpha\alpha^{2/7}\Lambda^{12/7}}\ge1-\eps .
\]
\end{enumerate}
In other words, $E_0=(1+o(1))\,c_\alpha\alpha^{2/7}\Lambda^{12/7}$ as $\ell\to\infty$, $\Xi\to\infty$ and
$\log(2+e^2)=o(\ell^{1/3})$.
\end{theorem}
For fixed $\alpha$ this is the statement $E_0=(1+o(1))\,c_\alpha\alpha^{2/7}\Lambda^{12/7}$ as $\Lambda\to\infty$ (Corollary~\ref{cor:limit}(i)); the two parameters $\ell$ and $\Xi$ give the uniform version.

If the parameters of the proof are chosen as powers of $\ell$ and $\Xi$ (for instance $s$, $\nu^2$, $\eta^{-1}$ and $M$
equal to small powers of $\ell$ in Proposition~\ref{prop:relerr}), it gives power rates
(for the lower bound under a condition $\log(2+e^2)\le\ell^{c}$ with a smaller exponent $c$); we do not
optimize them. The upper bound holds with $\eps_+(\ell)\le C\ell^{-1/2}$ (Remark~\ref{rem:upper-rate}).

\begin{corollary}\label{cor:limit}
Let $\Lambda/\alpha\to\infty$ and $\alpha^{5/2}\Lambda\to\infty$, in such a way that
$\log(2+\alpha)=o((\Lambda/\alpha)^{1/21})$. Then
\[
  E_0(\alpha,\Lambda)=\bigl(1+o(1)\bigr)\,c_\alpha\,\alpha^{2/7}\Lambda^{12/7},
  \qquad \frac{E_0(\alpha,\Lambda)}{\ELL(\alpha,\Lambda)}\to1 .
\]
In particular this holds (i) for fixed $\alpha>0$ as $\Lambda\to\infty$; (ii) uniformly in
$\Lambda^{-2/5+\delta}\le\alpha\le\Lambda^{1-\delta}$ for each fixed $\delta\in(0,\frac25)$;
(iii) uniformly in $\alpha\in(0,\alpha_2]$ for each fixed $\alpha_2>0$, provided
$\alpha^{5/2}\Lambda\to\infty$.
\end{corollary}

\begin{proof}
Theorem~\ref{thm:main} with $\ell\to\infty$ and $\Xi\to\infty$; note that
$\log(2+e^2)\le C+\log(2+\alpha)$ and $\ell^{1/3}=(4\pi^2)^{-1/21}(\Lambda/\alpha)^{1/21}$. In case (ii),
$\ell^7=\Lambda/(4\pi^2\alpha)\ge\Lambda^\delta/(4\pi^2)$,
$\Xi\ge c\,\Lambda^{5\delta/7}$ and $\log(2+e^2)\le C\log\Lambda=o(\ell^{1/3})$; in case (iii), $\log(2+e^2)$ is bounded.
\end{proof}

By the two relations in the proof above, the logarithmic condition $\log(2+e^2)=o(\ell^{1/3})$
follows from $\ell\to\infty$ when $\alpha$ is bounded, and, when $\alpha\to\infty$, it holds if
$\Lambda/[\alpha(\log\alpha)^{21}]\to\infty$. Within $\Lambda/\alpha\to\infty$ and $\alpha^{5/2}\Lambda\to\infty$, a sequence is therefore
not covered by the last sentence of Theorem~\ref{thm:main} if and only if it has a subsequence along which $\alpha\to\infty$ and
$\Lambda/[\alpha(\log\alpha)^{21}]$ stays bounded.

Thus, in the regime of Theorem~\ref{thm:main}, product states (the local density
approximation of Lieb and Loss) give the ground state energy correctly to leading
order, including the constant. The ratio $E_0/\ELL\to1$ concerns energies; no norm approximation
of ground states by product states is asserted. Since $E_0\le\ELL$, Theorem~\ref{thm:main} also determines
the leading asymptotics of $\ELL$ itself; its constant differs from the one stated in
\cite[Theorem~1.1]{BH22} (Appendix~\ref{app:BH}). The condition $\alpha^{5/2}\Lambda\to\infty$ cannot be
dropped: if $\alpha^{5/2}\Lambda\to0$, the vacuum trial state gives
$E_0\le\alpha\Lambda^2/(2\pi)=o(\alpha^{2/7}\Lambda^{12/7})$ (Remark~\ref{rem:window}).
This is the crossover scale identified by Lieb and Loss. We do not describe the
behaviour near the crossover $\alpha^{5/2}\Lambda\asymp1$. As a by-product we obtain the
virial-type partition $\langle\frac12(p+eA)^2\rangle/E_0\to\frac57$ and
$\langle\Hf\rangle/E_0\to\frac27$ in all approximate ground states
(Corollary~\ref{cor:virial}).

For the physical value $\alpha=1/137.036$ the local density approximation is below the
perturbative value only for $\Lambda\gtrsim1.3\times10^{8}$ in units of the electron
mass, far beyond any physically meaningful cutoff. Theorem~\ref{thm:main} is a
mathematical statement about the ultraviolet limit of the model.

\subsection{Background}
The mathematical analysis of the Pauli--Fierz model at fixed ultraviolet cutoff is by
now well developed: self-adjointness for all values of the coupling \cite{Hir02},
existence of ground states \cite{BFS98,BFS99,GLL01}, and, for the translation
invariant model, the fiber decomposition \cite{LLP53,LMS07} and the existence of a ground state at zero
total momentum \cite{HaSi23}. At fixed $\Lambda$ and small $\alpha$ the
self-energy and the renormalized mass have been analysed perturbatively, with rigorous
control of the low order terms \cite{HSe02,CH04,HS05,BCVV08,Che08}. These results
concern the regime below the crossover $\alpha^{5/2}\Lambda\asymp1$ mentioned above. See
also \cite{LL02} for non-perturbative bounds on binding energies and mass
renormalization, which are not restricted to this regime.

The regime $\Lambda\to\infty$ at fixed $\alpha$ is non-perturbative. After \cite{LL00},
Bach and Hach \cite{BH22} studied the product state energy $\ELL$ and reduced its
leading asymptotics to the minimum of $F[\beta]$, where
$F[\beta]=\inf\{\frac12\norm{\nabla\phi}^2+\beta\norm{\phi}_1\}$; this identifies the
leading term of $\ELL$, and hence the candidate for the leading term of $E_0$, with the
minimum of an explicit functional; the leading constant stated in \cite{BH22} differs from
ours by the factor $(2\pi)^{6/7}$ (Appendix~\ref{app:BH}). More recently Bach, Ballesteros and Mlinarzik \cite{BBM26} and Bach
and Herdzik \cite{BHe26} studied the Bogoliubov--Hartree--Fock (BHF) approximation of
the fiber Hamiltonians (at zero total momentum in \cite{BBM26}, at fixed total momentum in \cite{BHe26}); see also \cite{Her25}, and \cite{BBT13} for
an earlier minimization over quasi-free states, as well as the quasi-classical
approach of \cite{BFP23,BFP25}. In \cite{BBM26} it
is shown that the BHF functional with a non-convex term removed, which is convex, has minimum of order $\Lambda^{3/2}$;
the introduction of \cite{BBM26} describes the conjecture of Lieb and Loss as confirmed for $\ELL$ in \cite{BH22}, and
the estimate for the true ground state energy as a longstanding challenge.

Theorem~\ref{thm:main} also gives a lower bound of order $\Lambda^{12/7}$ for the BHF energy, since
the latter is an upper bound for $E_0$; see Remark~\ref{rem:BHF}.

As in the strong-coupling polaron problem, the leading energy is determined by a
variational functional: for the polaron, the ground state energy is asymptotically given by the classical Pekar functional
\cite{Pek54,DV83,LT97}; see \cite{LY58} for early bounds, \cite{Lie77} for the Pekar
minimizer, and \cite{FS21} for the quantum correction. The functional $F[\beta]$ plays
the role of the Pekar functional, with the non-local self-interaction replaced by
$\beta\norm{\phi}_1$. Lieb and Thomas \cite{LT97} gave a simple proof of the polaron
asymptotics by grouping field modes into cells in momentum space; our periodization
principle (Theorem~\ref{thm:periodization}) is an operator inequality version of that
step. The essential difference is that the Pauli--Fierz electron couples to the field
both through the diamagnetic term $e^2A^2$ and through the paramagnetic term
$e\,p\cdot A$, and the latter does not disappear when the field is treated
classically. Here the main additional estimate controls the paramagnetic term: the magnetic
Neumann lower bound of Theorem~\ref{thm:neumann}, which rests on the non-concentration estimate of
Theorem~\ref{thm:sup} (see Section~\ref{sec:strategy}).

\subsection{Scope}
The analysis is restricted to the spinless model with a sharp cutoff and without external
potentials; electrons with spin and smooth cutoffs are discussed in Section~\ref{sec:remarks}. For the
fibers with non-zero total momentum $P$, Remark~\ref{rem:fiberP} shows that, in the regime of Theorem~\ref{thm:main}, the fiber
energy satisfies $E(P)/E_0\to1$ uniformly for $\abs P=o(\alpha^{1/7}\Lambda^{6/7})$; larger momenta are not treated.
The main theorem is stated without rates.

In Appendix~\ref{app:BH} we prove two statements comparing our results with
\cite{BH22} (Propositions~\ref{prop:A1} and~\ref{prop:A2}); the discussion there of
where the factor $(2\pi)^{3/2}$ enters the arguments of \cite{BH22} is our reading and
is not part of our claims.

This paper is the first of a series on the ultraviolet limit of the Pauli--Fierz
model. The configuration-wise block Born--Oppenheimer bound (Proposition~\ref{prop:configBO}),
the periodization principle and the Neumann lower bound (the latter through Proposition~\ref{prop:configBO}) are stated in a form that is used
in the companion paper \cite{SerBC}, which proves a cutoff-independent lower bound on the
effective mass and the breakdown of mass renormalization at fixed physical mass. The present
paper is self-contained.

\subsection{Organization}
After the strategy (Section~\ref{sec:strategy}) and the definitions (Section~\ref{sec:model}), Section~\ref{sec:upper}
proves the upper bound. The lower bound combines the periodization principle (Section~\ref{sec:periodization}),
the configuration-wise block Born--Oppenheimer bound (Section~\ref{sec:blockBO}) and the Neumann lower bound
for band-limited magnetic fields (Section~\ref{sec:neumann}); Section~\ref{sec:proof} assembles them, and Section~\ref{sec:remarks}
contains remarks on extensions. The appendices compare with \cite{BH22} (Appendix~\ref{app:BH}), compute the
leading constant (Appendix~\ref{app:constants}), and show that the gradient condition in the Neumann lower bound
cannot be dropped (Appendix~\ref{app:counterexample}).

%% file: sec2_strategy.tex
\section{Strategy of the proof and heuristics}\label{sec:strategy}

This section is informal; precise statements follow in Sections~\ref{sec:model}--\ref{sec:proof}.
By a unitary scaling (Lemma~\ref{lem:scaling}) we may fix the photon band to
$\abs{k}<\pi$ and put the large parameter into the kinetic energy: the Hamiltonian
becomes $\frac\lambda2(p+eA)^2+\Hf$ with $\lambda=\Lambda/\pi$, and we write
$g:=\lambda e^2/2$ for the coefficient of $\abs{A}^2$.

\subsection{The product state computation}
For a real normalized $\varphi$ and a field state $\Psi$, the paramagnetic cross term
vanishes, and
$\ip{\varphi\otimes\Psi}{H\varphi\otimes\Psi}
=\frac\lambda2\norm{\nabla\varphi}^2+\ip{\Psi}{(\Hf+g\int\varphi^2\abs{A}^2)\Psi}$.
For a fixed density $\rho=\varphi^2$ the field part is a quadratic problem whose ground
state energy is
\begin{equation}\label{eq:Efield-heur}
  \Efield(\rho)=\tfrac12\Tr\bigl[(\omega^2+C_\rho)^{1/2}-\omega\bigr],\qquad C_\rho=2g\,\Pi'\rho\,\Pi',
\end{equation}
where $\Pi'$ is the projection onto transverse fields with momenta in the band. If
the induced stiffness dominates, $2g\rho\gg\omega^2$ where $\rho$ lives, then $C_\rho\gg\omega^2$ on the modes localized there and $\Efield(\rho)\approx\frac12\Tr C_\rho^{1/2}$. (The density itself is small: for the minimizers below $\rho\approx\ell^{-3}$, while $2g\rho\approx e^4\ell^4=\Xi^2$ is large compared with $\omega^2\le\pi^2$.) A
semiclassical count of the eigenvalues of the localized band projection
gives $\Tr(\Pi'\rho\Pi')^{1/2}\approx\rhom\int\sqrt\rho$, with
$\rhom=\pi/3$ the number of transverse modes per unit volume (the local Weyl law). Hence
\[
  \Efield(\rho)\approx\beta\norm{\varphi}_1,\qquad \beta=\frac{\rhom}{2}\sqrt{2g},
\]
and minimizing over $\varphi$ leads to $F_\lambda[\beta]=\inf\{\frac\lambda2\norm{\nabla\varphi}^2+\beta\norm{\varphi}_1\}$,
whose minimizer lives on the length scale $\ell=(\lambda/e^2)^{1/7}$ and whose value is
of order $e^2\ell^5$. Undoing the scaling gives $c_\alpha\alpha^{2/7}\Lambda^{12/7}$. The
upper bound (Section~\ref{sec:upper}) makes this rigorous: $\Efield\le\frac12\Tr C_\rho^{1/2}$ holds without error, and
since the square root is concave, $\Tr(\Pi'\rho\Pi')^{1/2}$ is bounded from above by the integral of the
square root of a lower symbol in coherent states (Proposition~\ref{prop:coherent}). The coefficient $\rhom$
appears as the volume of the transverse part of phase space. In the lower bound the same
local Weyl number appears as the trace $\rhom\abs Q$ of a Gram matrix on each block $Q$
(Corollary~\ref{cor:hardy}). This coincidence is what makes a sharp constant possible.

\subsection{The obstacle in the lower bound}
In a general state the paramagnetic term does not vanish, and the electron can lower its
energy by correlating its phase with the field. For the polaron, Lieb and Thomas
\cite{LT97} handle the corresponding problem by replacing the quantized field with a
classical one at a controlled cost. For the Pauli--Fierz model a classical field $a$
still acts on the electron through the magnetic kinetic energy
$\frac\lambda2\int\abs{(\nabla+iea)\psi}^2$, and the electron can screen a smooth field by
a gauge transformation. The lower bound therefore has to show that, after the quantum
fluctuations of the field are taken into account, the field energy cost of
localizing the electron is still $\beta\norm{\varphi}_1$ to leading order.

\subsection{The four steps}
\emph{Step 1 (localization).} By the IMS localization formula the electron is confined to
a box of side $n_e\approx M\ell$ at a relative cost $O(M^{-2})$ (Lemma~\ref{lem:IMS}).

\emph{Step 2 (periodization).} The continuum of field modes is replaced by the finitely
many modes of a large periodic box $\T_L$, with an error that vanishes as
$L\to\infty$ (Theorem~\ref{thm:periodization}). In the finite mode model the field is a
finite system of harmonic oscillators, and in the Schr\"odinger representation of
the field coordinates $q$ the electron Hamiltonian is a direct integral over $q$. Hence
$h\ge\Hf+W(q)$, where $W(q)$ is the electron ground state energy in the classical field
$A_q$ (a Born--Oppenheimer lower bound).

\emph{Step 3 (configuration-wise block Born--Oppenheimer bound).} Two observations turn this into
the expression of the upper bound. First, for $\nu\in\N$ and $u=\psi^\nu/\abs{\psi}^{\nu-1}$,
\[
  \abs{(\nabla+ia)\psi}^2=\Bigl(1-\frac1{\nu^2}\Bigr)\abs{\nabla\abs{\psi}}^2+\frac1{\nu^2}\abs{(\nabla+i\nu a)u}^2
\]
pointwise, vortices included (Lemma~\ref{lem:phase}). At the price of a fraction
$\nu^{-2}$ of the kinetic energy, the magnetic coupling is transferred to the field
$\nu a$, and on each cube $Q$ of a partition the electron energy is bounded below by
its mass in $Q$ times the Neumann ground state energy of $(\nabla+i\nu a)$ on $Q$.
Second, the kinetic energy of the field coordinates belonging to a block $Q$ is bounded
below by an inverse of the local field mass $u_Q(q)$ (a Hardy inequality in field space,
Lemma~\ref{lem:hardy}): $-\frac12\Delta_{S_Q}\ge h_Q^2/(8u_Q)$, where $h_Q\approx\rhom\abs Q$ is the trace of a Gram
matrix. If the Neumann energy is at least $\gamma\,u_Q/\abs{Q}$, the arithmetic--geometric mean
inequality gives
\[
  \frac{h_Q^2}{8u_Q}+\gamma\,\bar\rho_Q\,u_Q\ \ge\ h_Q\sqrt{\gamma\bar\rho_Q/2}\ \ge\ (1-o(1))\,\beta\int_Q\abs\psi ,
\]
the last step by the Cauchy--Schwarz inequality, and summing over $Q$ gives $\beta\norm{\psi}_1$ for every
configuration $q$ of the field. This is Proposition~\ref{prop:configBO}, deduced from an abstract version
(Theorem~\ref{thm:abstractBO}). The bound loses the zero-point energy of the field modes in the box
(Lemma~\ref{lem:vacuum}), of order $n_e^3\approx M^3\ell^3$; relative to the leading term $e^2\ell^5$ this is
$M^3/(e^2\ell^2)=M^3/\Xi$. This is the origin of the condition $\Xi\to\infty$, that is $\alpha^{5/2}\Lambda\to\infty$,
which cannot be dropped (Remark~\ref{rem:window}).

\emph{Step 4 (Neumann lower bound).} The input needed in Step~3 is a lower bound on the
magnetic Neumann Laplacian on a cube: $E_0^N(Q;\nu e B)$ should be at least
$(1-\eta)$ times the averaged energy of the gauge-projected field, or a uniform
positive floor if the field is large (Theorem~\ref{thm:neumann}). This is the main difficulty.
Such a bound is false in general: a band-limited
field may concentrate its vorticity in a small ball inside $Q$ while being large
outside $Q$ (superoscillation), and then the electron simply avoids the ball
(Appendix~\ref{app:counterexample}). We prove the bound for fields satisfying a weighted
gradient condition, which controls the amplitude outside $Q$. The core is a non-concentration
estimate: under the gradient condition, the vorticity of a band-limited field in $Q$ is controlled by
its gauge-projected $L^2$ norm, with a factor polynomial in the degree $D$ of a Taylor approximation
(Theorem~\ref{thm:sup}); given this, the Neumann bound follows from elementary inequalities. Field
configurations violating the gradient condition are paid for by a small weighted part of the field
energy, whose cost is a local vacuum energy mismatch (Lemma~\ref{lem:weight}).

%% file: sec3_model.tex
\section{The model and the comparison functional}\label{sec:model}

\subsection{The Hamiltonian}
Let $\mathfrak h=L^2(\R^3\times\{1,2\})$ and let $\mathcal F=\mathcal F_{\mathrm b}(\mathfrak h)$ be the
bosonic Fock space with vacuum $\Omega$, annihilation and creation operators
$a(k,r),a^*(k,r)$, and $\Phi(h)=a(h)+a^*(h)$ for $h\in\mathfrak h$ \cite{Ara18} (without the factor $2^{-1/2}$ of the
Segal field in \cite{Ara18}). Inner products are linear in the second variable, the Fourier transform is
$\hat f(k)=\int\mathrm e^{-ik\cdot x}f(x)\,dx$, so that $\norm f^2=(2\pi)^{-3}\norm{\hat f}^2$, and inequalities between unbounded
operators are meant in the sense of quadratic forms. The
polarization vectors $\eps_r(k)$, $r=1,2$, depend only on $\hat k=k/\abs{k}$, and
$\{\hat k,\eps_1(k),\eps_2(k)\}$ is orthonormal (Coulomb gauge). For $\Lambda>0$ the
quantized vector potential is
\begin{equation}\label{eq:A}
  A_\Lambda(x)=\sum_{r=1,2}\int_{\abs{k}\le\Lambda}\frac{\eps_r(k)}{(2\pi)^{3/2}\sqrt{2\abs{k}}}
  \bigl(\mathrm e^{ik\cdot x}a(k,r)+\mathrm e^{-ik\cdot x}a^*(k,r)\bigr)\,dk ,
\end{equation}
which is the continuum version of the normalization \cite[(1.1)]{LL00}, and
$\Hf=d\Gamma(\abs{k})$. The Hamiltonian of a spinless electron of unit mass is
\begin{equation}\label{eq:H}
  H=H_{e,\Lambda}=\tfrac12\bigl(p+eA_\Lambda(x)\bigr)^2+\Hf\qquad\text{on }L^2(\R^3)\otimes\mathcal F,
  \qquad e^2=4\pi\alpha,
\end{equation}
with $p=-i\nabla_x$; $A^2$ is not normal ordered (the sign of $e$ plays no role). We set
\[
  E_0(\alpha,\Lambda):=\inf\sigma(H),\qquad
  \ELL(\alpha,\Lambda):=\inf\bigl\{\ip{\phi\otimes\Psi}{H\,\phi\otimes\Psi}:\ \norm\phi=\norm\Psi=1\bigr\}.
\]
Clearly $E_0\le\ELL$. Spreading out $\phi$ shows
$E_0\le\frac{e^2}{2}\langle\Omega,A_\Lambda(0)^2\Omega\rangle=\frac{e^2\Lambda^2}{8\pi^2}=\frac{\alpha\Lambda^2}{2\pi}$.

\begin{remark}[Definition of $H$]\label{rem:def}
Let $\mathcal D_0:=C_c^\infty(\R^3)\otimes_{\mathrm{alg}}\mathcal F_{\mathrm{fin}}(C_c^\infty)$. The
operators $p_j+eA_{\Lambda,j}$ and $\Hf^{1/2}$ are symmetric on $\mathcal D_0$, hence closable,
and therefore the form
$\mathfrak q(\Psi)=\frac12\norm{(p+eA_\Lambda)\Psi}^2+\norm{\Hf^{1/2}\Psi}^2$ on $\mathcal D_0$ is
closable. We let $H$ be the self-adjoint operator associated with its closure (the
Friedrichs extension). Then $E_0$ is the infimum of $\mathfrak q(\Psi)/\norm\Psi^2$ over
$\mathcal D_0$. All our statements concern this number and do not depend on the
operator domain. Other self-adjoint realizations, for instance the one on
$D(\Delta)\cap D(\Hf)$ constructed in \cite{Hir02}, have $\mathcal D_0$ as a form core and
give the same infimum: by (T2) of Section~\ref{sec:periodization} the form norm of $\mathfrak q$ on $\mathcal D_0$ is
equivalent to that of $\frac12p^2+\Hf+1$, for which $\mathcal D_0$ is a core, and the form of such a realization is
continuous for this norm on its domain.
\end{remark}

\begin{remark}[Schr\"odinger representation of finitely many modes]\label{rem:Qrep}
For the finite mode models of Sections~\ref{sec:periodization}--\ref{sec:proof} we use the
joint spectral representation of the commuting real field operators, in which
$\mathcal F\cong L^2(\R^d,dq)$, $\Hf=\frac12(\abs{p_q}^2+q\cdot\omega^2q)-\frac12\Tr\omega$ and
$A(x)$ acts as multiplication by a classical field $A_q(x)$. The electron operator is
then decomposable, $\int^\oplus h_{\mathrm{el}}[A_q]\,dq$, and consequently
$h_{\mathrm{el}}[A_q]\ge W(q):=\inf\sigma(h_{\mathrm{el}}[A_q])$ as an operator inequality.
In particular the Born--Oppenheimer comparison $h\ge\Hf+W(q)$ holds.
\end{remark}

\begin{remark}[Fibers]\label{rem:fibers}
Theorem~\ref{thm:main} concerns $\inf\sigma(H)$ on the full space. By
\cite[Theorem~4.1(iv) and Proposition~7.1]{LMS07}, $\inf\sigma(H)=E(0)$, the ground
state energy of the fiber Hamiltonian $H(0)$ at total momentum zero; we use this only
to read Theorem~\ref{thm:main} as a statement about $H(0)$.
\end{remark}

\begin{remark}[Smooth versions of band-limited Gaussian fields]\label{rem:smooth}
Let $R(x)=\Phi(f_x)$ with $f_x\in\mathfrak h$ band-limited and $x\mapsto f_x$ smooth. In
the representation of Remark~\ref{rem:Qrep}, expanding $f_x$ in a real orthonormal basis
$\{\xi_m\}$ of $\mathfrak h$ gives $R_q(x)=\sum_mq_mc_m(x)$ with independent standard
Gaussian $q_m$ and $c_m(x):=\overline{\ip{\xi_m}{f_x}}$. On a bounded cube $B$,
$\sum_m\norm{c_m}^2_{H^{j+2}(B)}=\sum_{\abs\gamma\le j+2}\int_B\norm{\partial^\gamma_xf_x}^2dx<\infty$
since $\abs{k^\gamma}\le\pi^{\abs\gamma}$ on the band. Hence, with $\mu$ the Gaussian
measure of the $q_m$, the partial sums form a
Cauchy sequence in $L^2(d\mu;H^{j+2}(B))$, and by the Sobolev embedding
$H^{j+2}(B)\subset C^j(\bar B)$ also in $L^2(d\mu;C^j(\bar B))$; a subsequence converges
almost surely in $C^j(\bar B)$. The limit is a $C^j$ version of $R$, and
$\partial^\gamma R(x)=\Phi(\partial^\gamma f_x)$ almost surely for each $x$.
\end{remark}

\subsection{Scaling to a fixed band}
\begin{lemma}[Scaling]\label{lem:scaling}
For $c,\Lambda,a>0$ let $H_{c,\Lambda}:=\frac c2(p+eA_\Lambda)^2+\Hf$. There is a unitary
$U_a$, mapping product states to product states and $\mathcal D_0$ onto itself, such that
$U_aH_{c,\Lambda}U_a^*=a^{-1}H_{c/a,a\Lambda}$. In particular $E_0$ and $\ELL$ of the two
operators are related by the factor $a^{-1}$.
\end{lemma}

\begin{proof}
The map $(uf)(k,r)=a^{-3/2}f(k/a,r)$ is unitary on $\mathfrak h$; we have
$u\abs{k}u^*=a^{-1}\abs{k}$ and $\Gamma(u)\Hf\Gamma(u)^*=a^{-1}\Hf$. For the form factor
\[
  f_{\Lambda,x}(k)=\chi_{\abs k\le\Lambda}\,\eps(k)\,\mathrm e^{-ik\cdot x}(2\pi)^{-3/2}(2\abs k)^{-1/2}
\]
one has $uf_{\Lambda,x}=a^{-1}f_{a\Lambda,x/a}$, so
$\Gamma(u)A_\Lambda(x)\Gamma(u)^*=a^{-1}A_{a\Lambda}(x/a)$. With the dilation
$(V_a\psi)(y)=a^{3/2}\psi(ay)$, for which $V_axV_a^*=ay$ and $V_apV_a^*=a^{-1}p$, the
unitary $U_a:=V_a\otimes\Gamma(u)$ gives, on $\mathcal D_0$,
\[
  U_aH_{c,\Lambda}U_a^*=\frac c{2a^2}\bigl(p+eA_{a\Lambda}(y)\bigr)^2+a^{-1}\Hf=a^{-1}H_{c/a,a\Lambda}.
\]
\end{proof}

\emph{Convention.} Taking $a=\pi/\Lambda$ we have $H_{1,\Lambda}\cong\frac\Lambda\pi H_{\lambda,\pi}$
with $\lambda:=\Lambda/\pi$ (in \cite{SerBC} the letter $\lambda$ denotes the ratio of the cutoff to the bare mass; here it is the coefficient of the kinetic energy). In Sections~\ref{sec:upper}--\ref{sec:proof} we work with
the photon band $\abs k<\pi$ (the sphere has measure zero) and the operator
$H^{(\lambda)}:=\frac\lambda2(p+eA_\pi)^2+\Hf$, and we use the parameters
\begin{equation}\label{eq:params}
  \ell:=\Bigl(\frac{\lambda}{e^2}\Bigr)^{1/7}=\Bigl(\frac\Lambda{\pi e^2}\Bigr)^{1/7},\qquad
  g:=\frac{\lambda e^2}{2},\qquad \Xi:=e^2\ell^2,\qquad \lambda=e^2\ell^7 .
\end{equation}
Then $E_0(\alpha,\Lambda)=\frac\Lambda\pi E_0(H^{(\lambda)})$, and similarly for $\ELL$.

\subsection{The comparison functional}
Following \cite[(1.6)--(1.9)]{BH22} (only this definition is taken from \cite{BH22}; none of its results is used in the proof of Theorem~\ref{thm:main}, see Appendix~\ref{app:BH}), let
$\mathcal F_\beta(\phi):=\frac12\norm{\nabla\phi}^2+\beta\norm{\phi}_1$ and
\[
  F[\beta]:=\inf\bigl\{\mathcal F_\beta(\phi):\ \phi=\abs\phi\in H^1\cap L^1(\R^3),\ \norm\phi_2=1\bigr\}.
\]
The scaling $\phi_R=R^{-3/2}\phi(\cdot/R)$ with $R=\beta^{-2/7}$ gives
$F[\beta]=\beta^{4/7}F[1]$, and
\[
  F_\lambda[\beta]:=\inf\Bigl\{\frac\lambda2\norm{\nabla\phi}^2+\beta\norm{\phi}_1\Bigr\}=\lambda F[\beta/\lambda]=\lambda^{3/7}\beta^{4/7}F[1].
\]

Thus $F[1]$ is an infimum, and only this is used in the proofs. Its minimizer is explicit and
$F[1]=4.96461575\ldots$ (Proposition~\ref{prop:F1} in Appendix~\ref{app:constants}); this identification
is used only for the numerical value of $c_\alpha$ and in Remark~\ref{rem:upper-rate}.

\emph{The leading coefficient.} In physical units set $\beta_*:=e\Lambda^3/(6\pi^2)$; in
the working convention set $\beta:=\frac\pi6\,e\sqrt{\lambda}=\frac\rhom2\sqrt{2g}$ with
$\rhom=\pi/3$ (Section~\ref{sec:upper}). Then
\begin{equation}\label{eq:coeff}
  F[1]\beta_*^{4/7}=c_\alpha\,\alpha^{2/7}\Lambda^{12/7}=\frac\Lambda\pi F_\lambda[\beta],\qquad
  F_\lambda[\beta]=c_F\,e^2\ell^5,\quad c_F:=(\pi/6)^{4/7}F[1].
\end{equation}

\begin{remark}[The condition $\alpha^{5/2}\Lambda\to\infty$]\label{rem:window}
If $\alpha^{5/2}\Lambda\to0$, the vacuum trial state gives
\[
  \frac{E_0}{c_\alpha\alpha^{2/7}\Lambda^{12/7}}\le\frac{\alpha\Lambda^2/(2\pi)}{c_\alpha\alpha^{2/7}\Lambda^{12/7}}
  =\frac{(\alpha^{5/2}\Lambda)^{2/7}}{2\pi c_\alpha}\to0,
\]
with $2\pi c_\alpha=6.2416\ldots$, so the non-perturbative leading term does not appear.
The exponent $-\frac25$ in Corollary~\ref{cor:limit}(ii) is this crossover scale, which
was pointed out in \cite[Section~1]{LL00}: perturbation theory can work at most for a
range of $\alpha$ shrinking at least like $\Lambda^{-2/5}$. In our lower bound the condition
$\Xi\to\infty$ enters through the zero-point energy of the field modes in the electron box
(Lemma~\ref{lem:vacuum}), which is of relative size $M^3/\Xi$.
For $\alpha=1/137.036$ the crossover $(\alpha^{5/2}\Lambda)^{2/7}=2\pi c_\alpha$ occurs at
$\Lambda\approx1.3\times10^8$ in units of the electron mass (a numerical value, not a bound).
In \cite{LL00} the coupling is written $\sqrt\alpha\,A$ with $\alpha\approx1/137$ in Gaussian
units. With the normalization \eqref{eq:A}, which is that of \cite[(1.1)]{LL00}, the physical value is $e^2=4\pi/137.036$ (reading the coefficient in \cite{LL00} literally would give $e^2=1/137.036$ and change only numerical values); this is the value
used above.
\end{remark}

%% file: sec4_upper.tex
\section{The field energy in a fixed density and the upper bound}\label{sec:upper}

\subsection{Band projections}
The field geometry $\sharp$ is either $\R^3$ or the torus $\T_L=(\R/L\Z)^3$ with an integer
$L\ge4$. Let $\Pi'$ be the orthogonal projection in $L^2(\sharp;\C^3)$ onto transverse
fields ($P_k=1-\hat k\hat k^{\mathsf T}$) with momenta in the open ball $\abs k<\pi$ (on
$\T_L$, $k\in\frac{2\pi}L\Z^3\setminus\{0\}$), with integral kernel $K(x,y)$, and let
$\rhom:=\tr K(x,x)$. On $\R^3$, $\rhom=2\abs{B_\pi}/(2\pi)^3=\pi/3$ (two transverse modes per
unit volume of phase space); on $\T_L$, $\rhom=2N_L/L^3$ with
$N_L=\#\{m\in\Z^3\setminus\{0\}:\abs m<L/2\}$, and
\begin{equation}\label{eq:rhom-torus}
  \frac\pi3\Bigl(1-\frac{\sqrt3}{L}\Bigr)^3-\frac2{L^3}\ \le\ \rhom^{\T_L}\ \le\ \frac\pi3\Bigl(1+\frac{\sqrt3}{L}\Bigr)^3\qquad(L\ge4),
\end{equation}
since the union of the unit cubes centred at the lattice points contains, and is
contained in, balls of radius $L/2\mp\sqrt3/2$. We also use the points
$\Z^3+(\frac12,\frac12,\frac12)$ and $\sinc_a:=\Pi_{\mathrm{cube}}\delta_a$, where
$\Pi_{\mathrm{cube}}$ is the projection onto the cubic band $[-\pi,\pi)^3$; the functions
$\sinc_a\mathbf e_i$ form an orthonormal basis of the cubic band. On $\T_L$ we take $L$ odd whenever these functions serve as real coordinates (in \eqref{eq:Efield} and in Sections~\ref{sec:blockBO}--\ref{sec:proof}): then $\frac{2\pi}L\Z^3\cap[-\pi,\pi)^3$ is invariant under $\kappa\mapsto-\kappa$ and the $\sinc_a$ are real. Since the torus is used only for all sufficiently large $L$ or in the limit $L\to\infty$, this is no restriction.

\begin{lemma}\label{lem:hansen}
Let $f:[0,\infty)\to[0,\infty)$ be concave with $f(0)=0$, for example $f(t)=\sqrt t$.
(a) For a contraction $K$ of finite rank and bounded $X\ge0$, $\Tr f(K^*XK)\ge\Tr K^*f(X)K$.
(b) In finite dimensions, for fixed $C\ge0$ the map $A\mapsto\Tr[(A+C)^{1/2}-A^{1/2}]$ is
decreasing on $A\ge0$.
\end{lemma}

\begin{proof}
(a) Since $f$ is concave and $f(0)=0$, we have $f(tx)\ge tf(x)$ for
$t\in[0,1]$. Let $\{v_n\}$ be an orthonormal basis of eigenvectors of the finite rank operator
$K^*XK$ and $t_n:=\norm{Kv_n}^2\le1$. If $t_n>0$, put $w_n:=Kv_n/\sqrt{t_n}$; Jensen's inequality for
the spectral measure of $X$ in the state $w_n$ gives
$f(\ip{v_n}{K^*XKv_n})=f(t_n\ip{w_n}{Xw_n})\ge t_nf(\ip{w_n}{Xw_n})\ge t_n\ip{w_n}{f(X)w_n}=\ip{v_n}{K^*f(X)Kv_n}$,
and both sides vanish if $t_n=0$. Sum over $n$. (For operator monotone $f$ the operator inequality $f(K^*XK)\ge K^*f(X)K$
is Hansen's inequality \cite{Han80}, see also \cite[Corollary~2.3]{HP03}; only its trace is used.)
(b) For $A>0$,
$\frac{d}{dt}\Tr(A+tC)^{1/2}=\frac12\Tr[(A+tC)^{-1/2}C]$, and $x\mapsto x^{-1/2}$ is operator
decreasing, so the integrand decreases as $A$ increases; integrate over $t\in[0,1]$. The
general case follows by replacing $A$ with $A+\eps$.
\end{proof}

\subsection{The field energy}
For $\rho\ge0$ we consider the ground state energy of the quadratic field problem
\[
  \Efield^\sharp(\rho):=\inf\sigma\Bigl(\Hf^\sharp+g\int\rho(x)\abs{A^\sharp(x)}^2dx\Bigr),\qquad g=\frac{\lambda e^2}{2},
\]
where $A^\sharp$ is the field with band $\abs k<\pi$; on $\T_L$ it has finitely many modes,
and auxiliary uncoupled oscillators may be added (Section~\ref{sec:blockBO}). On $\T_L$ the
normal form of the oscillators gives exactly
\begin{equation}\label{eq:Efield}
  \Efield^{\T_L}(\rho)=\tfrac12\Tr\bigl[(\omega^2+C_\rho)^{1/2}-\omega\bigr],\qquad
  C_\rho=2g\,X^*\Pi'\rho\,\Pi'X,
\end{equation}
where (with the auxiliary oscillators of Section~\ref{sec:finite-model}, on which $C_\rho$ vanishes) $X$ is the unitary from mode coordinates to the cubic band in
$L^2(\T_L;\R^3)$. In the coordinates $q$ of Remark~\ref{rem:Qrep} the field is $A_q=\Pi'Xq$, because the form factor
$(2\abs k)^{-1/2}$ of the field cancels the factor $\sqrt{2\omega}$ in $a+a^*=\sqrt{2\omega}\,q$; hence
$g\int\rho\abs{A_q}^2=\frac12q\cdot C_\rho q$, and together with $\Hf=\frac12(\abs p^2+q\cdot\omega^2q)-\frac12\Tr\omega$ this gives
\eqref{eq:Efield}. Comparing with $A=0$ in Lemma~\ref{lem:hansen}(b),
\begin{equation}\label{eq:Efield-upper}
  \Efield^{\T_L}(\rho)\ \le\ \tfrac12\Tr C_\rho^{1/2}=\tfrac12\sqrt{2g}\,\Tr(\Pi'\rho\Pi')^{1/2}.
\end{equation}
On $\R^3$, Theorem~\ref{thm:periodization}(b) gives
$\Efield^{\R^3}(\rho)=\lim_L\Efield^{\T_L}(\rho)$ for $\rho\ge0$ of bounded support with $\int\rho=1$. The upper bound therefore
reduces to an upper bound on the trace of the square root of the localized band projection
$\Pi'\rho\Pi'$, with the sharp coefficient $\rhom$ of $\norm{\sqrt\rho}_1$. When $2g\rho\gg\omega^2$ on the support of $\rho$ (large induced stiffness, as for the densities of scale $\ell$, where $2g\rho\approx\Xi^2$) the right
side of \eqref{eq:Efield-upper} is also the correct size of $\Efield$, so no constant is lost in
\eqref{eq:Efield-upper}.

\subsection{A coherent state bound for the trace of the square root}
Since $\sqrt\cdot$ is concave, the trace of $B^{1/2}$ is bounded from above by the integral of
the square root of a lower symbol (the concave side of the Berezin--Lieb inequality \cite{Ber72,Lie73,Sim80}). To
keep the symbol localized we first pass to the symmetric operator $\sqrt\rho\,\Pi'\sqrt\rho$, and
we use coherent states whose polarization is adapted to their momentum.

\begin{lemma}[Jensen's inequality for a frame]\label{lem:frame}
Let $(\mathcal A,\mu)$ be a $\sigma$-finite measure space and $\Phi_a\in\mathcal H$ a measurable family with
$\int\abs{\ip{\Phi_a}f}^2d\mu(a)=\norm f^2$ for all $f\in\mathcal H$. Then for every bounded $B\ge0$, in
$[0,\infty]$,
\[
  \Tr B^{1/2}\ \le\ \int\norm{\Phi_a}\,\ip{\Phi_a}{B\Phi_a}^{1/2}\,d\mu(a).
\]
\end{lemma}

\begin{proof}
For bounded $T\ge0$ and an orthonormal basis $(e_n)$, Tonelli's theorem gives
$\Tr T=\sum_n\norm{T^{1/2}e_n}^2=\int\sum_n\abs{\ip{T^{1/2}\Phi_a}{e_n}}^2d\mu=\int\ip{\Phi_a}{T\Phi_a}d\mu$. Take $T=B^{1/2}$
and apply the Cauchy--Schwarz inequality to the spectral measure $\nu_a$ of $B$ in the vector
$\Phi_a$: $\int t^{1/2}d\nu_a\le\norm{\Phi_a}(\int t\,d\nu_a)^{1/2}$.
\end{proof}

For $t\ge1$ let $\phi_0(y):=(2\pi t^2)^{-3/4}\mathrm e^{-\abs y^2/(4t^2)}$, $\psi_{x,\xi}(y):=\mathrm e^{i\xi\cdot y}\phi_0(y-x)$, and for
$\xi\ne0$ choose measurably a real orthonormal basis $e_1(\xi),e_2(\xi)\perp\xi$, $e_3(\xi):=\hat\xi$. The
family $\Phi_{x,\xi,j}:=\psi_{x,\xi}e_j(\xi)$ with $d\mu=dx\,d\xi/(2\pi)^3\otimes(\text{counting in }j)$ satisfies the
hypothesis of Lemma~\ref{lem:frame} in $L^2(\R^3;\C^3)$, and $\norm{\Phi_{x,\xi,j}}=1$: indeed, by
Plancherel in $x$ (we use $\norm f^2=(2\pi)^{-3}\int\abs{\hat f}^2$),
$\int\abs{\ip{\Phi_{x,\xi,j}}f}^2dx=(2\pi)^{-3}\int\abs{e_j(\xi)\cdot\hat f(k)}^2\abs{\hat\phi_0(k-\xi)}^2dk$, and
$\sum_j\abs{e_j(\xi)\cdot v}^2=\abs v^2$, $(2\pi)^{-3}\int\abs{\hat\phi_0}^2=1$.

\begin{proposition}[Coherent state bound]\label{prop:coherent}
There is an absolute constant $C$ such that the following holds. Let $0\le\varphi\in W^{2,\infty}(\R^3)$
with $\supp\varphi\subset\bar B(0,R)$, $L_1:=\norm{\nabla\varphi}_\infty$, $L_2:=\norm{\Delta\varphi}_\infty$, $\rho:=\varphi^2$, and $t\ge1$.
\begin{enumerate}[label=(\alph*)]
\item On $\R^3$,
\[
  \Tr(\Pi'\rho\Pi')^{1/2}\ \le\ \frac\pi3\Bigl(1+\frac Ct\Bigr)\norm\varphi_1+C(R+t)^3\Bigl(L_1t+L_2+\frac{L_1}t\Bigr).
\]
\item For $L>2R$ view $\rho$ as a function on $\T_L$. Then the limit superior as $L\to\infty$ of
$\Tr(\Pi'_L\rho\Pi'_L)^{1/2}$ satisfies the same bound.
\end{enumerate}
\end{proposition}

\begin{proof}
(a) With $\mathcal T_\varphi:=\Pi'\varphi$ we have $\mathcal T_\varphi\mathcal T_\varphi^*=\Pi'\rho\Pi'$ and $\mathcal T_\varphi^*\mathcal T_\varphi=\varphi\Pi'\varphi$, whose non-zero eigenvalues
coincide with multiplicities, so $\Tr(\Pi'\rho\Pi')^{1/2}=\Tr(\varphi\Pi'\varphi)^{1/2}$. Apply Lemma~\ref{lem:frame} to
$B=\varphi\Pi'\varphi$; the lower symbol is $b_j(x,\xi)=\norm{\Pi'(\varphi\Phi_{x,\xi,j})}^2$. Write
$\varphi=\varphi(x)+\chi_x$; the triangle inequality gives
\[
  b_j^{1/2}\le\varphi(x)\,m_j(\xi)^{1/2}+\norm{\Pi'(\chi_x\Phi_{x,\xi,j})},\qquad
  m_j(\xi):=\norm{\Pi'\Phi_{x,\xi,j}}^2=\int_{\abs k<\pi}\abs{P_ke_j(\xi)}^2\,\hat\gamma(k-\xi)\,dk,
\]
where $\hat\gamma:=(2\pi)^{-3}\abs{\hat\phi_0}^2$ is the centred Gaussian density with covariance
$(4t^2)^{-1}\mathbf 1$; $m_j$ does not depend on $x$. Let $Z$ have density $\hat\gamma$.

\emph{Main term.} For $j=1,2$ we have $m_j\le1$. For $j=3$, $\abs{P_k\hat\xi}=\sin\angle(k,\xi)\le\abs{k-\xi}/\abs\xi$
(the distance of $\xi$ from the line $\R k$), so $m_3\le\mathbb E\abs Z^2/\abs\xi^2=3/(4t^2\abs\xi^2)$. For
$\abs\xi>\pi$ and all $j$, $\abs{\xi+Z}<\pi$ forces $\hat\xi\cdot Z<-(\abs\xi-\pi)$, so
$m_j\le\mathrm e^{-2t^2(\abs\xi-\pi)^2}$. Hence, with $\abs{B_\pi}/(2\pi)^3=\pi/6$,
\[
  \sum_j\int\frac{d\xi}{(2\pi)^3}\,m_j^{1/2}\ \le\ \frac\pi3+\frac{4\pi}{(2\pi)^3}\int_0^\pi\frac{\sqrt3}{2tr}\,r^2dr
  +\frac{3\cdot4\pi}{(2\pi)^3}\int_0^\infty(\pi+\tau)^2\mathrm e^{-t^2\tau^2}d\tau\ \le\ \frac\pi3\Bigl(1+\frac Ct\Bigr),
\]
and the main term contributes at most $\frac\pi3(1+\frac Ct)\int\varphi(x)\,dx$. The two transverse polarizations
give exactly $\rhom=\pi/3$; the longitudinal one is suppressed because its polarization is
adapted to the momentum $\xi$ of the coherent state.

\emph{Error term.} Let $h_x:=\chi_x\phi_0(\cdot-x)$ (a scalar function). Since
$(\chi_x\psi_{x,\xi})^\wedge(k)=\hat h_x(k-\xi)$ and $\abs{P_ke_j}\le1$,
$\norm{\Pi'(\chi_x\Phi_{x,\xi,j})}^2\le n_x(\xi):=(2\pi)^{-3}\int_{\abs k<\pi}\abs{\hat h_x(k-\xi)}^2dk$. With $w(\xi):=(1+\abs\xi^2)^2$,
$\int w^{-1}d\xi=\pi^2$ and $w(k-\eta)\le(1+2\pi^2)^2(1+\abs\eta^2)^2$ for $\abs k<\pi$, the Cauchy--Schwarz
inequality gives
\[
  \int\frac{d\xi}{(2\pi)^3}\,n_x(\xi)^{1/2}\le\Bigl(\frac{\pi^2}{(2\pi)^3}\Bigr)^{1/2}\Bigl(\frac{\abs{B_\pi}}{(2\pi)^3}\Bigr)^{1/2}(1+2\pi^2)\norm{(1-\Delta)h_x}
  =\frac{1+2\pi^2}{\sqrt{48}}\norm{(1-\Delta)h_x}.
\]
Put $z=y-x$. Since $\abs{\chi_x(y)}\le L_1\abs z$, $\nabla\chi_x=\nabla\varphi$, $\Delta\chi_x=\Delta\varphi$,
$\nabla\phi_0=-\frac{z}{2t^2}\phi_0$, $\Delta\phi_0=(\frac{\abs z^2}{4t^4}-\frac3{2t^2})\phi_0$, and $\abs{\phi_0}^2$ is a Gaussian with
variance $t^2$ in each coordinate, the Gaussian moments give
$\norm{(1-\Delta)h_x}\le C(L_1t+L_2+L_1/t)$ for all $x$. If $\abs x\ge R+t$, then $\varphi(x)=0$, the main
term vanishes, $h_x=\varphi\,\phi_0(\cdot-x)$ is supported in $\abs z\ge\delta_x:=\abs x-R$, and
$\abs z^{2m}\mathrm e^{-\abs z^2/(2t^2)}\le\mathrm e^{-\delta_x^2/(4t^2)}\abs z^{2m}\mathrm e^{-\abs z^2/(4t^2)}$ there, so
$\norm{(1-\Delta)h_x}\le C(L_1t+L_2+L_1/t)\mathrm e^{-\delta_x^2/(8t^2)}$. Integrating over $\abs x<R+t$ and
$\abs x\ge R+t$ (where $\int_t^\infty(R+\tau)^2\mathrm e^{-\tau^2/(8t^2)}d\tau\le Ct(R+t)^2$) gives the error term.

(b) For $L>2R$ the support of $\varphi$ lies in a fundamental domain $\mathcal D$ of $\T_L$, and
$\varphi\Pi'_L\varphi$ is unitarily equivalent to the operator $f\mapsto\varphi\,\Pi'_L(\varphi f)$ on $L^2(\R^3;\C^3)$
(zero outside $\mathcal D$). For $h$ supported in $\mathcal D$ one has $\norm{\Pi'_Lh}^2=\mathcal N_L(h)^2$ with the seminorm
$\mathcal N_L(h)^2:=L^{-3}\sum_{\kappa\in\frac{2\pi}L\Z^3,\,0<\abs\kappa<\pi}\abs{P_\kappa\hat h(\kappa)}^2$ on $L^1(\R^3;\C^3)$. Apply
Lemma~\ref{lem:frame} with the same frame; the lower symbol is $\mathcal N_L(\varphi\Phi_{x,\xi,j})^2$, and the triangle
inequality for $\mathcal N_L$ gives the same splitting with $m_j^L(\xi):=\mathcal N_L(\Phi_{x,\xi,j})^2$ and with
$n^L_x(\xi):=L^{-3}\sum_\kappa\abs{\hat h_x(\kappa-\xi)}^2$. As Riemann sums, $m^L_j\to m_j$ pointwise; for $\abs\xi\le\pi+1$,
$m^L_j\le Ct^3$, and for $\abs\xi>\pi+1$, $m^L_j\le Ct^3\mathrm e^{-2t^2(\abs\xi-\pi)^2}$, uniformly in $L\ge4$, so by dominated
convergence $\int(m^L_j)^{1/2}d\xi\to\int m_j^{1/2}d\xi$. In the error term the factor $\abs{B_\pi}/(2\pi)^3$ is replaced
by $L^{-3}\#\{0<\abs\kappa<\pi\}=N_L/L^3\to\pi/6$. Hence the bound of (a) holds in the limit superior.
\end{proof}

\begin{remark}\label{rem:coherent}
For the rescaled functions $\varphi_R:=R^{-3/2}u(\cdot/R)$ with a fixed $0\le u\in W^{2,\infty}$ of bounded support
and $1\le t\le R$, Proposition~\ref{prop:coherent} (with width $t$ and its $R$ replaced by $R$ times the radius of a ball containing $\supp u$) gives
$\Tr(\Pi'\varphi_R^2\Pi')^{1/2}\le\frac\pi3\norm{\varphi_R}_1(1+C/t+C_ut/R)$ with $C_u$ depending only on $u$; with
$t=\lceil R^{1/2}\rceil$ the relative error is $C_uR^{-1/2}$. The loss $C/t$ comes from the smoothing of the
sharp band edge by the coherent states, the loss $C_ut/R$ from the variation of $\varphi$ on their scale $t$.
Since $\Tr\,\Pi'\sqrt\rho\,\Pi'=\rhom\norm{\sqrt\rho}_1$, Lemma~\ref{lem:hansen}(a) (on $\T_L$, with $K=\Pi'$) shows that the coefficient $\rhom$ cannot be improved; the
lower bound of Section~\ref{sec:blockBO} recovers the same coefficient as the trace of a Gram matrix.
\end{remark}

\subsection{Proof of Theorem~\ref{thm:main}(a)}
We work in the convention of Lemma~\ref{lem:scaling}. Let $T(\varphi):=\frac\lambda2\norm{\nabla\varphi}^2$.
For real $\varphi\in H^1$ and any $\Psi$, the cross term
$\ip{\varphi}{(p\cdot A+A\cdot p)\varphi}=\int A\cdot2\,\mathrm{Re}\bigl(\varphi\,\overline{(-i\nabla\varphi)}\bigr)$
vanishes, so
$\ip{\varphi\otimes\Psi}{H^{(\lambda)}\varphi\otimes\Psi}=T(\varphi)+\ip\Psi{(\Hf+g\int\varphi^2\abs A^2)\Psi}$, and taking
the infimum over $\Psi$,
\begin{equation}\label{eq:ELL-MF}
  E_0(H^{(\lambda)})\ \le\ \ELL(H^{(\lambda)})\ \le\ T(\varphi)+\Efield^{\R^3}(\varphi^2).
\end{equation}
Let $\eps>0$. By truncation and mollification there is $0\le u\in C_c^\infty(\R^3)$ with $\norm u_2=1$ and
$\mathcal F_1(u)\le F[1]+\eps$. Let $R:=(\lambda/\beta)^{2/7}=(6/\pi)^{2/7}\ell$ and $\varphi_R:=R^{-3/2}u(\cdot/R)$, so that
$\lambda R^{-2}=\beta R^{3/2}=\lambda^{3/7}\beta^{4/7}$ and $T(\varphi_R)+\beta\norm{\varphi_R}_1=\lambda^{3/7}\beta^{4/7}\mathcal F_1(u)$. By
Theorem~\ref{thm:periodization}(b) (with a box $\mathcal C_e$ containing $\supp\varphi_R$), \eqref{eq:Efield-upper},
$\frac12\sqrt{2g}\,\frac\pi3=\beta$ and Proposition~\ref{prop:coherent}(b) with Remark~\ref{rem:coherent},
\[
  \Efield^{\R^3}(\varphi_R^2)=\lim_L\Efield^{\T_L}(\varphi_R^2)\ \le\ \beta\norm{\varphi_R}_1\bigl(1+C_uR^{-1/2}\bigr)\qquad(R\ge1).
\]
With \eqref{eq:ELL-MF} and \eqref{eq:coeff} this gives
$\ELL/(c_\alpha\alpha^{2/7}\Lambda^{12/7})\le(1+C_uR^{-1/2})(F[1]+\eps)/F[1]$. The right side depends only on $\ell$ and $\eps$,
so
\[
  \eps_+(\ell):=\inf_{u}\,\Bigl(\bigl(1+C_u\,((6/\pi)^{2/7}\ell)^{-1/2}\bigr)\frac{\mathcal F_1(u)}{F[1]}\Bigr)-1
\]
(the infimum over $0\le u\in C_c^\infty$ with $\norm u_2=1$) satisfies $\limsup_{\ell\to\infty}\eps_+(\ell)\le\eps/F[1]$ for every
$\eps>0$, that is, $\eps_+(\ell)\to0$. Returning to physical units by \eqref{eq:coeff} proves
Theorem~\ref{thm:main}(a) for all $\alpha>0$. \qed

\begin{remark}[Rate of the upper bound]\label{rem:upper-rate}
The normalized minimizer $u_*$ of $F[1]$ (Proposition~\ref{prop:F1}) is $C^{1,1}=W^{2,\infty}$ and has bounded
support, so the same argument with $u=u_*$ and $\eps=0$ gives $\eps_+(\ell)\le C\ell^{-1/2}$ for $\ell\ge1$.
\end{remark}

%% file: sec5_periodization.tex
\section{Periodization of the field}\label{sec:periodization}

In this section the electron enters only through the condition (W$\pm$) below, so the
results apply beyond the Pauli--Fierz model.

\subsection{Setting}
\emph{Field on $\R^3$.} Let $\mathcal K=\{k:\abs k<\pi\}$, $\mathfrak h=L^2(\mathcal K;\C^3)$ and
$\omega=\abs k$; the longitudinal component is included as an uncoupled mode. Let
$A_i(x)=\Phi(f^i_x)$ with
$f^i_x(k)_j=(2\pi)^{-3/2}(2\abs k)^{-1/2}\mathrm e^{-ik\cdot x}P_\perp(k)_{ji}$. With respect to the real
structure $(Jh)(k)=\overline{h(-k)}$ the operators $\{\Phi(h):Jh=h\}$ commute, and we use their
joint spectral representation (Remark~\ref{rem:Qrep}), with Gaussian measure $\mu$.

\emph{Electron.} Let $\mathcal C_e$ be the cube of side $n_e$ centred at the origin and
$\mathfrak k$ the electron Hilbert space. To each real band-limited vector field $a$ there
corresponds a self-adjoint operator $h_{\mathrm{el}}[a]\ge0$ on $\mathfrak k$, whose form
domain does not depend on $a$, such that for some constant $C_W$ and all
$\delta\in(0,1]$
\begin{equation}\tag{W$\pm$}\label{eq:W}
  (1-\delta)h_{\mathrm{el}}[a]-\frac{C_W}{\delta}\norm b^2_{\infty,\mathcal C_e}\ \le\ h_{\mathrm{el}}[a+b]
  \ \le\ (1+\delta)h_{\mathrm{el}}[a]+C_W(1+\delta^{-1})\norm b^2_{\infty,\mathcal C_e}.
\end{equation}
Examples: (i) the magnetic operator $\frac\lambda2(p+ea)^2$ with Dirichlet conditions on
$\mathcal C_e$, $\mathfrak k=L^2(\mathcal C_e)$, $C_W=\lambda e^2/2$ (for self-adjoint $Z_1=p+ea$ and
$Z_2=eb$ one has $\pm(Z_1Z_2+Z_2Z_1)\le\delta Z_1^2+\delta^{-1}Z_2^2$, without commutativity);
(ii) $g\int\rho\abs a^2$ for fixed $\rho$ with $\int\rho=1$, $\mathfrak k=\C$, $C_W=g$;
(iii) the ground state energy $W[a]=\inf\sigma(h_{\mathrm{el}}[a])$, $\mathfrak k=\C$, with the
same $C_W$ as in (i) or (ii), by taking infima in the operator inequalities.

\emph{Operators.} Let $h^\infty:=d\Gamma(\omega)\otimes1+\int^\oplus h_{\mathrm{el}}[A_q]\,d\mu(q)$ on
$\mathcal F(\mathfrak h)\otimes\mathfrak k$, defined as a form sum (both terms are non-negative and closed; for the
Pauli--Fierz electron its relation to $H^{(\lambda)}$ is Lemma~\ref{lem:Qform}); we use only lower bounds on it,
valid on its form domain, and the upper bound of Step~4 below uses explicit trial states. For
$L\in\N$ with $L>\sqrt3n_e+2$, the model $h^L$ on $\T_L$ is given by the same expression
with the field replaced by the finitely many modes
$\kappa\in\frac{2\pi}L\Z^3\setminus\{0\}$, $\abs\kappa<\pi$:
$A^L_i(x)=\sum_{\kappa,j}L^{-3/2}(2\abs\kappa)^{-1/2}P_\perp(\kappa)_{ij}(\mathrm e^{i\kappa\cdot x}a_{\kappa j}+\mathrm{h.c.})$.
Uncoupled auxiliary oscillators do not change $\inf\sigma(h^L)$.

\subsection{Statement}
\begin{theorem}[Periodization principle]\label{thm:periodization}
(a) Under \eqref{eq:W},
\[
  \inf\sigma(h^\infty)\ \ge\ \limsup_{L\to\infty}\inf\sigma(h^L).
\]
More precisely, for every $\delta\in(0,\frac12]$ there is $L_0(\delta,n_e,C_W)$ such that
$\inf\sigma(h^\infty)\ge\frac{1-\delta}{1+\delta}\inf\sigma(h^L)-\delta$ for $L\ge L_0$.

(b) For $\rho\ge0$ with $\int\rho=1$ and $\supp\rho\subset\mathcal C_e$,
$\Efield^{\R^3}(\rho)=\lim_{L\to\infty}\Efield^{\T_L}(\rho)$, uniformly in such $\rho$: for
$L\ge L_0(\delta,n_e,g)$,
$\frac{1-\delta}{1+\delta}\Efield^{\T_L}(\rho)-\delta\le\Efield^{\R^3}(\rho)\le\frac{1+\delta}{1-\delta}\Efield^{\T_L}(\rho)+\delta$.
\end{theorem}

The proof uses only \eqref{eq:W}, the boundedness of the band, and the infrared
integrability $\int_{\abs k<r}\abs k^{-2}dk=4\pi r$ in three dimensions; the same argument works in all
dimensions $d\ge3$ (for $d=2$ the infrared term diverges, cf.\ Section~\ref{sec:remarks}).

\subsection{The error kernel}
Let $d_L:=\sqrt3\pi/L$, $B_\kappa:=\kappa+[-\frac\pi L,\frac\pi L)^3$ and $r:=L^{-1/2}$ (so $r\ge2d_L$ for
$L\ge119$). The inner modes are $\mathcal M:=\{\kappa:\abs\kappa\ge r,\ \abs\kappa+d_L<\pi\}$ and the
outer modes $\mathcal O$ are the remaining ones. The vectors
$v_{\kappa j}:=\abs{B_\kappa}^{-1/2}\mathbf 1_{B_\kappa}\mathbf e_j$, $\kappa\in\mathcal M$, are orthonormal; let
$\mathfrak h_c$ be their span, $P$ the projection onto it, and $b_{\kappa j}=a(v_{\kappa j})$. The
cell field $\tilde A_i(x):=\Phi(h^i_x)$ with
$h^i_x(k)_j:=(2\pi)^{-3/2}(2\abs\kappa)^{-1/2}\mathrm e^{-i\kappa\cdot x}P_\perp(\kappa)_{ji}$ for $k\in B_\kappa$ is
unitarily equivalent to the inner part $A^{L,\mathcal M}$ of the torus field
($b\leftrightarrow a^L$). The remainder is $R_i(x)=\Phi(\check f^i_x)$, $\check f^i_x:=f^i_x-h^i_x$.

\begin{lemma}[Error kernel]\label{lem:errorkernel}
Let $\rho_0\ge\max(1,\sqrt3n_e/2)$, $\abs x\le\rho_0$, $\abs\gamma\le2$ and
$\varsigma:=r+d_L+d_L^2((1+\rho_0)^2+r^{-1})$. There is an absolute constant $C_E$ such that:
(a) $\norm{\omega^{-1/2}\partial^\gamma_x\check f_x}^2\le C_E\varsigma$ and $\norm{\partial^\gamma_x\check f_x}^2\le\pi C_E\varsigma$;
(b) for the outer torus modes,
$\sum_{\kappa\in\mathcal O,i,j}\abs{\partial^\gamma_xc^i_{x,\kappa j}}^2/\abs\kappa\le C_E(r+d_L)$, where
$c^i_{x,\kappa j}$ are the coefficients of $A^L$.
With $r=L^{-1/2}$, $\varsigma\to0$ as $L\to\infty$.
\end{lemma}

\begin{proof}
(1) On the non-inner region $\mathcal N\subset\{\abs k<r+d_L\}\cup\{\pi-2d_L\le\abs k<\pi\}$ we have $\check f=f$
and $\omega^{-1}\abs{\partial^\gamma f}^2\le(2\pi)^{-3}\pi^4\abs k^{-2}$, whose integral is
$\le C(r+d_L)$ since $\int_{\abs k<r}\abs k^{-2}dk=4\pi r$. (2) On an inner cell
$\abs{k-\kappa}\le d_L$ and $\abs{k'}\ge\abs k/3$ on the segment joining them; the gradient of
$F(k)=(-ik)^\gamma \mathrm e^{-ik\cdot x}(2\abs k)^{-1/2}P_\perp(k)$ is
$\le C[(1+\rho_0)\abs k^{-1/2}+\abs k^{-3/2}]$, so
$\omega^{-1}\abs{\partial^\gamma\check f}^2\le Cd_L^2[(1+\rho_0)^2\abs k^{-2}+\abs k^{-4}]$; integrate over
$\{r/2\le\abs k<\pi\}$. (3) $L^{-3}\sum_{0<\abs\kappa<r}\abs\kappa^{-2}\le r$ (bound $\abs m^{-2}$ by its
cell average), and the number of lattice points in the boundary shell is $\le CL^3d_L$.
\end{proof}

We use the following standard facts.
(T1) $a^*(h)a(h)\le\norm{\omega^{-1/2}h}^2d\Gamma(\omega)$.
(T2) $\Phi(h)^2\le4\norm{\omega^{-1/2}h}^2d\Gamma(\omega)+2\norm h^2$.
(T3) Sobolev on the cube: $\sup_{\mathcal C_e}\abs u^2\le C_S\sum_{\abs\gamma\le2}\int_{\mathcal C_e}\abs{\partial^\gamma u}^2$ for $u:\mathcal C_e\to\R^3$. Here $C_S$ is
the constant for vector-valued functions, taken as three times the scalar constant, so that it absorbs the sum over
the three components when (T2) is applied componentwise; the factor $40=10\cdot4$ in Steps~2 and~4 below counts the
ten multi-indices $\abs\gamma\le2$ and the factor $4$ of (T2).
(T4) $R$ is a centred Gaussian field with smooth covariance, has a smooth version
(Remark~\ref{rem:smooth}), and for a.e.\ $q$ one has $A_q=\tilde A_q+R_q$ for all $x\in\mathcal C_e$.
(T5) $0\le\omega_1\le\omega_2$ implies $d\Gamma(\omega_1)\le d\Gamma(\omega_2)$.

\subsection{Proof of Theorem~\ref{thm:periodization}}
\emph{(a)} Step 1 (field). Let $\omega^-_\kappa:=\inf_{B_\kappa}\abs k\ge(1-d_L/r)\abs\kappa$. Then
$\omega_-:=\sum_{\mathcal M}\omega^-_\kappa\mathbf 1_{B_\kappa}\le\omega$ commutes with $P$, and by (T5)
$d\Gamma(\omega)\ge(1-d_L/r)\Hf^{\mathcal M}\otimes1$.

Step 2 (electron). By (T4), the left inequality in \eqref{eq:W}, (T3), (T2) and
Lemma~\ref{lem:errorkernel}(a) (ten multi-indices),
$\int^\oplus h_{\mathrm{el}}[A_q]\ge(1-\delta)\int^\oplus h_{\mathrm{el}}[\tilde A_q]-\eps_Ld\Gamma(\omega)-\frac\pi2\eps_L$
with $\eps_L:=40C_EC_WC_S\delta^{-1}n_e^3\varsigma\to0$.

Step 3. For all large $L$ we have $(1-\eps_L)(1-d_L/r)\ge1-\delta$, hence
$h^\infty\ge(1-\delta)h^{\mathcal M}\otimes1-\frac\pi2\eps_L$ with $h^{\mathcal M}:=\Hf^{\mathcal M}+h_{\mathrm{el}}[\tilde A]$.

Step 4 (the torus from above). Write $\mathcal F^L=\mathcal F(\mathcal M)\otimes\mathcal F(\mathcal O)$. By the right
inequality in \eqref{eq:W}, (T3), (T1), (T2) and Lemma~\ref{lem:errorkernel}(b), trial
states $\Psi\otimes\Omega_{\mathcal O}$ give
$\inf\sigma(h^L)\le(1+\delta)\inf\sigma(h^{\mathcal M})+\frac\pi2\eps'_L$ with
$\eps'_L:=40C_EC_WC_S(1+\delta^{-1})n_e^3(r+d_L)\to0$ (for all large $L$, $\eps'_L\le\delta$; note $\Hf^{\mathcal M}\ge0$).

Step 5. Combine Steps 3 and 4 and choose $L_0$ such that $\frac\pi2(\eps_L+\eps'_L)\le\delta$
for $L\ge L_0$. $L_0$ depends only on $(\delta,n_e,C_W)$.

\emph{(b)} Part (a) applied to example (ii) gives the lower inequality, uniformly in
$\rho$ since only $C_W=g$ enters. For the upper inequality let
$E^{\mathcal M}(\rho):=\inf\sigma(\Hf^{\mathcal M}+g\int\rho\abs{\tilde A}^2)$. On the torus,
$g\int\rho\abs{A^{\mathcal M}+A^{\mathcal O}}^2\ge(1-\delta)g\int\rho\abs{A^{\mathcal M}}^2-g\delta^{-1}\int\rho\abs{A^{\mathcal O}}^2$, and
(T1), (T2) and Lemma~\ref{lem:errorkernel}(b) with $\gamma=0$ give
$\Efield^{\T_L}\ge(1-\delta)E^{\mathcal M}-2\pi g\delta^{-1}C_E(r+d_L)$; no supremum over $x$ is needed
for fixed $\rho$. On $\R^3$ use trial states $\Psi\otimes\Omega_\perp$ in
$\mathcal F(\mathfrak h)=\mathcal F(\mathfrak h_c)\otimes\mathcal F(\mathfrak h_c^\perp)$: then
$\ip{\Psi\otimes\Omega_\perp}{d\Gamma(\omega)\Psi\otimes\Omega_\perp}=\ip\Psi{d\Gamma(P\omega P)\Psi}\le(1+d_L/r)\ip\Psi{\Hf^{\mathcal M}\Psi}$,
since $P\omega P$ is diagonal and $\le(1+d_L/r)\abs\kappa$. Moreover
$g\int\rho\abs{\tilde A+R}^2\le(1+\delta)g\int\rho\abs{\tilde A}^2+g(1+\delta^{-1})\int\rho\abs R^2$, and the component of
$\check f^i_x$ in $\mathfrak h_c$ acts on $\Psi$; by (T2) and Lemma~\ref{lem:errorkernel}(a) with $\gamma=0$ (recall $\int\rho=1$ and
$\supp\rho\subset\mathcal C_e$),
$\ip{\Psi\otimes\Omega_\perp}{\int\rho\abs R^2\,\Psi\otimes\Omega_\perp}\le4C_E\varsigma\ip{\Psi\otimes\Omega_\perp}{d\Gamma(\omega)\Psi\otimes\Omega_\perp}+2\pi C_E\varsigma$.
Hence the energy of $\Psi\otimes\Omega_\perp$ is at most
$(1+d_L/r)(1+4gC_E\varsigma(1+\delta^{-1}))\ip\Psi{\Hf^{\mathcal M}\Psi}+(1+\delta)\ip\Psi{g\int\rho\abs{\tilde A}^2\Psi}+2\pi g(1+\delta^{-1})C_E\varsigma$.
For $L$ so large that $(1+d_L/r)(1+4gC_E\varsigma(1+\delta^{-1}))\le1+\delta$, the infimum over $\Psi$ gives
$\Efield^{\R^3}\le(1+\delta)E^{\mathcal M}+2\pi g(1+\delta^{-1})C_E\varsigma$. Combine. \qed

\begin{remark}\label{rem:period-last}
In the proof of the lower bound (Section~\ref{sec:proof}) the Born--Oppenheimer step,
the block decomposition and the Neumann bound are all carried out in the finite mode
model on $\T_L$, and part (a) is used once, at the end. In particular, no
Born--Oppenheimer argument is needed for the continuum field. For the upper bound,
part (b) is used once, for a fixed density (Section~\ref{sec:upper}).
\end{remark}

%% file: sec6_blockBO.tex
\section{A block Born--Oppenheimer lower bound}\label{sec:blockBO}

This section turns the Born--Oppenheimer bound $h\ge\Hf+W(q)$ of Remark~\ref{rem:Qrep} into a lower
bound by $(1-\nu^{-2})T(\abs\psi)+\beta\norm\psi_1$ up to small errors, configuration by configuration of the
field. The input on the electron side is the Neumann lower bound of Section~\ref{sec:neumann}; the result
is Proposition~\ref{prop:configBO}.

\subsection{The finite mode model on the torus}\label{sec:finite-model}
On $\T_L$ ($L$ an odd integer, see Section~\ref{sec:upper}) the field consists of the physical modes (transverse, open
ball $\abs k<\pi$, $\omega=\abs k$) and auxiliary oscillators (longitudinal modes and the
modes in the cube minus the ball), which are uncoupled and have frequency $\pi$ (any frequencies in $(0,\pi]$ depending only on the Fourier mode would do; what is used is that $\omega$ commutes with translations). The coordinates $q\in\R^d$ are the coordinates with respect to the orthonormal
basis $\{\sinc_a\mathbf e_i\}$ of the cubic band (Section~\ref{sec:upper}); $X:\R^d\to L^2(\T_L;\R^3)$ is the
corresponding unitary and $B_q:=\Pi'Xq$ is the physical field. In the representation of
Remark~\ref{rem:Qrep}, $\Hf=\frac12(\abs p^2+q\cdot\omega^2q)-\frac12\Tr\omega$. Since the auxiliary
oscillators are uncoupled, neither $\inf\sigma$ nor $\Efield^{\T_L}$ depends on their
frequencies. They are added because the sinc basis provides coordinates that are
orthogonal block by block, which allows us to apply Lemma~\ref{lem:hardy} to all blocks
simultaneously.

Let $s\in\N$, let $\mathcal C_{\mathrm b}\subset\T_L$ be a box of side $n_{\mathrm b}\in s\N$ with corners in $\Z^3$ (the block box), and
$L\ge\max(n_{\mathrm b}+8s,s^2)$. Let $\mathcal P$ be the partition of $\mathcal C_{\mathrm b}$ into cubes of side $s$,
$S_Q:=\mathrm{span}\{\sinc_a\mathbf e_i:a\in Q\}$ (mutually orthogonal), $S_{\mathrm{in}}:=\bigoplus_QS_Q$ and
$P_{\mathrm{in}}$ the projection onto it. We use the weights
$\mathrm w_Q(z):=(1+\dist(z,Q)/s)^{-4}$, the constant $C_w:=\sup_x\sum_{Q\in\mathcal P}\mathrm w_Q(x)\le43$, and
the weight function $\theta:=\theta_0\sum_Q\mathrm w_Q/C_w$ with $\theta_0\in(0,\frac12]$ (the exponent $4$ is the
smallest integer for which $C_w$ is bounded independently of the block box, since about $j^2$ cubes lie at
distance $\approx js$ from a point; any fixed $\theta_0<1$ would do, and we take $\theta_0=\frac12$ below). The projected
field mass is $u'_Q(q):=\norm{P'_Q\mathbf 1_QB_q}^2_{L^2(Q)}$, where $P'_Q$ is the projection in
$L^2(Q;\C^3)$ onto the orthogonal complement of $\nabla H^1(Q)$. Finally $E_0^N(Q;b)$ denotes
the ground state energy of the operator associated with the form $\int_Q\abs{(\nabla+ib)v}^2$,
$v\in H^1(Q)$ (the magnetic Neumann Laplacian), and $T(\psi)=\frac\lambda2\norm{\nabla\psi}^2$.

\subsection{An abstract lower bound}
Let $d\in\N$, let $\omega^2$ be a positive definite matrix on $\R^d$, and let
$\Hf:=\frac12(\abs p^2+q\cdot\omega^2q)-\frac12\Tr\omega$ on $L^2(\R^d)$, $p=-i\nabla_q$.

\begin{assumption}\label{ass:H}
$\mathcal P$ is a finite index set and $\abs Q>0$ for $Q\in\mathcal P$.
\begin{enumerate}[label=(H\arabic*)]
\item (Block coordinates.) $S_Q\subset\R^d$, $Q\in\mathcal P$, are mutually orthogonal subspaces, and
$P_{\mathrm{in}}$ is the projection onto $S_{\mathrm{in}}:=\bigoplus_QS_Q$.
\item (Gram condition.) For each $Q$ there is a linear map $r_Q$ from $\R^d$ to a real
Hilbert space such that, with $u_Q(q):=\norm{r_Q(q)}^2$, $Y_Q:=r_Q\restriction_{S_Q}$ and
$G'_Q:=Y_Q^*Y_Q$, one has $\norm{G'_Q}\le1$. Set $h_Q:=(\tr G'_Q-2)_+$.
\item (Weight.) There are a matrix $0\le\Theta\le\theta_0\omega^2$ with $\theta_0\in(0,1)$, non-negative
quadratic forms $\mathcal E_Q(q)$ and a constant $c_\theta>0$ with
$\frac12q\cdot\Theta q\ge c_\theta\sum_Q\mathcal E_Q(q)$. Set $\omega_\theta:=(\omega^2-\Theta)^{1/2}$ and
$\mathcal D_\theta:=\frac12\Tr(\omega-\omega_\theta)$.
\item (Electron.) There are a set $\mathfrak S$ of electron states, functions $T\ge0$ and
$m_Q\ge0$ on $\mathfrak S$ with $\sum_Qm_Q\le1$, $\nu\ge1$, and measurable functions $f_Q(q)\ge0$
such that the electron energy $\mathcal W$ (a measurable function of $q$, bounded below)
satisfies
\[
  \mathcal W(q)\ \ge\ \inf_{\varphi\in\mathfrak S}\Bigl[\Bigl(1-\frac1{\nu^2}\Bigr)T(\varphi)+\sum_Qm_Q(\varphi)f_Q(q)\Bigr],
\]
and there are constants $\gamma,u_b,\Upsilon>0$ such that $f_Q(q)\ge\gamma\min\{u_Q(q),u_b\}/\abs Q$
whenever $\mathcal E_Q(q)\le\Upsilon u_Q(q)$, and $c_\theta\Upsilon\ge\gamma/\abs Q$ for all $Q$.
\end{enumerate}
\end{assumption}

\begin{theorem}[Abstract block Born--Oppenheimer bound]\label{thm:abstractBO}
Assume (H1)--(H4), let $\bar\rho_Q:=m_Q/\abs Q$ and
\[
  Y:=\inf_{\varphi\in\mathfrak S}\Bigl[\Bigl(1-\frac1{\nu^2}\Bigr)T(\varphi)+\sum_Qh_Q\sqrt{\tfrac{\gamma\bar\rho_Q(\varphi)}{2}}\Bigr]-\sum_Q\frac{h_Q^2}{2u_b}.
\]
Then (i) $\frac12\abs{P_{\mathrm{in}}p}^2+\mathcal W(q)+\frac12q\cdot\Theta q\ge Y$ as forms on $L^2(\R^d)$;
(ii) $\inf\sigma(\Hf+\mathcal W)\ge Y-\frac12\Tr(P_{\mathrm{in}}\omega P_{\mathrm{in}})-\mathcal D_\theta$.
\end{theorem}

The proof is given in Section~\ref{sec:proof-abstractBO}; its Steps~1--3 give the
configuration-wise form used in Proposition~\ref{prop:configBO}. The mechanism is the
arithmetic--geometric mean inequality $\frac{h^2}{8u}+\gamma\bar\rho\,u\ge h\sqrt{\gamma\bar\rho/2}$: the kinetic energy of the
field coordinates of a block (a Hardy inequality, Lemma~\ref{lem:hardy}) is inversely proportional to the
local field mass $u$, while the electron energy (H4) is proportional to it. The minimum over $u$ does
not depend on the field, and summing over the blocks gives $\beta\norm\varphi_1$.

\subsection{Tools}
\begin{lemma}[Separation of the vacuum energy]\label{lem:vacuum}
Let $\R^d=S_{\mathrm{in}}\oplus S_{\mathrm{out}}$ (orthogonal), let $\mathcal V(q)$ be a multiplication operator
bounded below, and let $\omega^2>0$ be any frequency matrix. Then
$\frac12(\abs p^2+q\cdot\omega^2q)-\frac12\Tr\omega+\mathcal V\ge\frac12\abs{p_{\mathrm{in}}}^2+\mathcal V(q)-\frac12\Tr(P_{\mathrm{in}}\omega P_{\mathrm{in}})$.
\end{lemma}

\begin{proof}
Decompose $\omega^2$ into blocks; with the Schur complement $\Sigma\ge0$,
$q\cdot\omega^2q=(q_o+Gq_i)\cdot\omega^2_{oo}(q_o+Gq_i)+q_i\cdot\Sigma q_i$. In each fiber of fixed $q_i$
the shifted oscillator is $\ge\frac12\Tr(\omega^2_{oo})^{1/2}$; since it contains no $p_i$, the
fiberwise inequality is an inequality on the whole space. By Lemma~\ref{lem:hansen}(a)
with $f=\sqrt\cdot$ and $K=P_{\mathrm{out}}$, $\Tr(\omega^2_{oo})^{1/2}\ge\Tr(P_{\mathrm{out}}\omega P_{\mathrm{out}})$.
\end{proof}

The loss $\frac12\Tr(P_{\mathrm{in}}\omega P_{\mathrm{in}})$ is the zero-point energy of the field modes in the block box.
It is of order $n_{\mathrm b}^3$ (Proposition~\ref{prop:configBO}(ii)); its role for the condition
$\Xi\to\infty$ is explained in Remark~\ref{rem:window}.

\begin{lemma}[Phase multiplication]\label{lem:phase}
(a) Let $U\subset\R^3$ be open, $a\in L^\infty_{\mathrm{loc}}(U;\R^3)$, $\psi\in H^1_{\mathrm{loc}}(U)$,
$\nu\in\N$, and $u:=F_\nu\circ\psi$ with $F_\nu(z)=z^\nu/\abs z^{\nu-1}$, $F_\nu(0)=0$. Then
$\abs u=\abs\psi$, $u\in H^1_{\mathrm{loc}}$, and almost everywhere
\[
  \abs{(\nabla+ia)\psi}^2=\Bigl(1-\frac1{\nu^2}\Bigr)\abs{\nabla\abs\psi}^2+\frac1{\nu^2}\abs{(\nabla+i\nu a)u}^2 .
\]
(b) Consequently, for $\psi\in H^1_0(\mathcal C_{\mathrm b})$ with $\norm\psi=1$,
\[
  \frac\lambda2\int\abs{(\nabla+ia)\psi}^2\ge\Bigl(1-\frac1{\nu^2}\Bigr)T(\abs\psi)+\frac{\lambda}{2\nu^2}\sum_QE_0^N(Q;\nu a)\,m_Q(\psi),\qquad
  m_Q(\psi):=\int_Q\abs\psi^2 .
\]
\end{lemma}

\begin{proof}
(a) $F_\nu$ is homogeneous of degree one and globally Lipschitz, and
$F_\nu(z)=z^{(\nu+1)/2}\bar z^{(1-\nu)/2}$ is smooth on $\C\setminus\{0\}$, so the chain rule holds
almost everywhere (approximate $F_\nu$ by $C^1$ Lipschitz functions equal to $F_\nu$ outside $\{\abs z<\eps\}$); on
$\{\psi=0\}$ we have $\nabla\psi=0$ and $\nabla u=0$ a.e.\ (gradients of Sobolev functions vanish a.e.\ on level sets,
\cite[Chapter~6]{LLAna}). At a point where
$\psi(x)=r\mathrm e^{i\vartheta}\ne0$ write $\nabla\psi(x)=\mathrm e^{i\vartheta}(X_\psi+iY_\psi)$ with $X_\psi=\nabla\abs\psi$ and
$Y_\psi=\mathrm{Im}(\bar\psi\nabla\psi)/\abs\psi$. Since $\partial_zF_\nu=\frac{\nu+1}{2}F_\nu/z$ and
$\partial_{\bar z}F_\nu=\frac{1-\nu}{2}F_\nu/\bar z$, we get $\nabla u(x)=\mathrm e^{i\nu\vartheta}(X_\psi+i\nu Y_\psi)$. Hence
$\abs{(\nabla+ia)\psi}^2=X_\psi^2+(Y_\psi+ra)^2$ and $\abs{(\nabla+i\nu a)u}^2=X_\psi^2+\nu^2(Y_\psi+ra)^2$, which gives
the identity pointwise. No continuous branch of the phase is used.
(b) Integrate (a) and use, for the disjoint cubes $Q$,
$\int_Q\abs{(\nabla+i\nu a)u}^2\ge E_0^N(Q;\nu a)\int_Q\abs u^2$, as $u\restriction_Q\in H^1(Q)$.
\end{proof}

The identity holds including the zeros (vortices) of $\psi$, since $\mathrm e^{i\nu\vartheta}$ is single
valued. The magnetic coupling is transferred to the field $\nu a$ at the cost of a fraction
$\nu^{-2}$ of the kinetic energy. Letting $\nu\to\infty$ recovers the diamagnetic inequality
\cite[Theorem~7.21]{LLAna}.

\begin{lemma}[Hardy inequality in field space]\label{lem:hardy}
Let $S\subset\R^d$ be a subspace, $r$ a linear map from $\R^d$ to a real Hilbert space,
$Y:=r\restriction_S$, $G_S:=Y^*Y$ and $u(q):=\norm{r(q)}^2$. If $g_S:=\tr G_S-2\norm{G_S}>0$, then
$-\Delta_S\ge g_S^2/(4\norm{G_S}u(q))$ as forms on $\{\psi\in L^2(\R^d):\nabla_S\psi\in L^2\}$.
\end{lemma}

\begin{proof}
For $\eps>0$ let $V:=tY^*r/(\abs r^2+\eps)$. Expanding $0\le\int\abs{\nabla_S\psi+V\psi}^2$ gives, on
$C_c^\infty(\R^d)$,
$-\Delta_S\ge\Div_SV-\abs V^2\ge(t\tr G_S-(2t+t^2)\norm{G_S})/(u+\eps)$; take $t=g_S/(2\norm{G_S})$. The
right side is bounded, hence the inequality extends to the form domain because
$C_c^\infty$ is dense in $\{\psi\in L^2:\nabla_S\psi\in L^2\}$ for the graph norm. Finally let
$\eps\downarrow0$ by monotone convergence (the right side may take the value $+\infty$).
\end{proof}

If $S=\R^n$ and $r$ is the identity, then $\tr G_S-2\norm{G_S}=n-2$ and the lemma is the
classical Hardy inequality $-\Delta\ge(n-2)^2/(4\abs x^2)$, proved here by the usual
ground state substitution; the new feature is the weight given by the Gram matrix.

\emph{Gram matrices.} Let $\sharp=\R^3$ or $\T_L$ and let $Q$ be a cube of side $s\in\N$ with corners in
$\Z^3$. Define $(G_Q)_{(a,i),(b,j)}:=\ip{\mathbf 1_Q\Pi'\sinc_a\mathbf e_i}{\mathbf 1_Q\Pi'\sinc_b\mathbf e_j}$ for
$a,b\in Q\cap(\Z^3+(\frac12,\frac12,\frac12))$, $B_Q:=\Pi'\mathbf 1_Q\Pi'$ and $L_Q:=\Tr(B_Q-B_Q^2)$. On $\R^3$ the kernel of
$\Pi'$ is $K(x-y)$ with $K(z)=(2\pi)^{-3}\int_{\abs k<\pi}\mathrm e^{ik\cdot z}P_\perp(k)\,dk$, and
\begin{equation}\label{eq:CK}
  \norm{K(z)}\ \le\ C_K(1+\abs z)^{-2}\qquad(z\in\R^3),
\end{equation}
with the operator norm of $3\times3$ matrices and an absolute constant $C_K$: $K$ is a combination of the
Fourier transforms of $\mathbf 1_{\abs k<\pi}$ and $\mathbf 1_{\abs k<\pi}\hat k_i\hat k_j$, which decay like $\abs z^{-2}$ (from the sphere)
and $\abs z^{-3}$ (from the singularity of $\hat k\hat k^{\mathsf T}$ at the origin).

\begin{lemma}[Gram trace]\label{lem:gram}
(a) $\Tr B_Q=\rhom\abs Q$ and $L_Q=\norm{\mathbf 1_{Q^c}\Pi'\mathbf 1_Q}_{\mathfrak S_2}^2$.
(b) $0\le\rhom\abs Q-\tr G_Q\le L_Q$.
(c) On $\R^3$, $L_Q\le72\pi C_K^2\,s^2\log(1+s/2)$.
\end{lemma}

\begin{proof}
(a) The first identity is the diagonal of the kernel. Since $\Pi'=\Pi'(\mathbf 1_Q+\mathbf 1_{Q^c})\Pi'$ on the range
of $\Pi'$, $B_Q-B_Q^2=\Pi'\mathbf 1_Q\Pi'\mathbf 1_{Q^c}\Pi'$, whose trace is $\norm{\mathbf 1_{Q^c}\Pi'\mathbf 1_Q}_{\mathfrak S_2}^2$.
(b) $\tr G_Q=\Tr(\mathbf 1_Q\Pi'P_{S_Q}\Pi'\mathbf 1_Q)\le\Tr(\mathbf 1_Q\Pi'\mathbf 1_Q)=\rhom\abs Q$, where $P_{S_Q}$ is the projection onto
the span of the $\sinc_a\mathbf e_i$, $a\in Q$. For the other inequality note that $\tr G_Q$ is a sum over the
sinc points $a\in Q$ of $\norm{\mathbf 1_QK(\cdot-a)}^2_{\mathrm{HS}}$ and $\norm{\mathbf 1_Q\Pi'\mathbf 1_Q}^2_{\mathfrak S_2}$ is the corresponding
integral over $a\in Q$; the following Fourier computation shows that the sum dominates. Write
$F_Q(\kappa)=\int_Q\mathrm e^{i\kappa\cdot x}dx=D_Q\hat\chi_1$ with $D_Q=\sum_{a\in Q}\mathrm e^{i\kappa\cdot a}$ and
$\hat\chi_1(\kappa)=\prod_i\frac{2\sin(\kappa_i/2)}{\kappa_i}\in[0,1]$ for $\abs{\kappa_i}\le2\pi$ (the Fourier transform of the
indicator of a unit cube centred at $0$). Since $\Pi'\sinc_a\mathbf e_i=K(\cdot-a)\mathbf e_i$ (the ball
lies in the cubic band), Plancherel gives
$\tr G_Q-(\rhom\abs Q-L_Q)=\tr G_Q-\norm{\mathbf 1_Q\Pi'\mathbf 1_Q}_{\mathfrak S_2}^2=\int\!\!\int\tr(P_kP_{k'})\abs{D_Q}^2(\hat\chi_1-\hat\chi_1^2)\ge0$, the integrals
over $\abs k,\abs{k'}<\pi$ with $\kappa=k-k'$ and the measure $(2\pi)^{-6}dk\,dk'$ (on $\T_L$, the corresponding sums).
(c) This is a bound of the type of Landau and Widom \cite{LW80} on the number of eigenvalues of $B_Q$
away from $0$ and $1$. By (a), $L_Q=\int_Qdy\int_{Q^c}dx\,\norm{K(x-y)}^2_{\mathrm{HS}}$ and $\norm M_{\mathrm{HS}}^2\le3\norm M^2$. With $\delta_Q(y):=\dist(y,\partial Q)$,
\eqref{eq:CK} gives $\int_{Q^c}\norm{K(x-y)}^2_{\mathrm{HS}}dx\le3C_K^2\int_{\abs z\ge\delta_Q(y)}(1+\abs z)^{-4}dz\le12\pi C_K^2/(1+\delta_Q(y))$. The set
$\{y\in Q:\delta_Q(y)>t\}$ is a cube of side $s-2t$, so $\int_Q(1+\delta_Q)^{-1}dy\le6s^2\int_0^{s/2}(1+t)^{-1}dt=6s^2\log(1+s/2)$.
\end{proof}

\begin{corollary}[Hardy inequality for the projected field]\label{cor:hardy}
(a) (Trace of the local gauge projection.) Let $s\ge1$, $Q$ a cube of side $s$ in $\R^3$, and
$P_{\nabla,Q}:=1-P'_Q$. Then $\Tr(P_{\nabla,Q}\mathbf 1_Q\Pi'\mathbf 1_QP_{\nabla,Q})\le\min\{\rhom s^3,\,C_D\abs{\partial Q}\}$, with
$\abs{\partial Q}=6s^2$ and an absolute constant $C_D$.
(b) On $\T_L$, let $Y'_Q:=P'_Q\mathbf 1_Q\Pi'X\restriction_{S_Q}$, $G'_Q:=Y'^*_QY'_Q$ and $h_Q:=(\tr G'_Q-2)_+$. Then
$-\Delta_{S_Q}\ge h_Q^2/(4u'_Q(q))$ as forms, and $h_Q\le\tr G_Q\le\rhom^{\T_L}\abs Q$.
(c) There are absolute constants $C_0<C$ and a function $L_{\mathrm G}(s)$ such that for every integer $s\ge1$ and
every cube $Q$ of side $s$ with corners in $\Z^3$: on $\R^3$, $h_Q\ge(1-C_0\log(2+s)/s)_+\rhom\abs Q$ (with $h_Q$
defined as in (b) on $\R^3$); and on $\T_L$ with $L\ge L_{\mathrm G}(s)$ odd, $h_Q\ge(1-C\log(2+s)/s)_+\frac\pi3\abs Q$.
\end{corollary}

\begin{proof}
(a) Let $Q=(0,s)^3$ and let $\tau_Q$ be the trace. Since $\Tr(\mathbf 1_Q\Pi'\mathbf 1_Q)=\rhom s^3$, $\tau_Q\le\rhom s^3\le2\rhom s^2$
for $s\le2$; let $s\ge2$. The Neumann eigenfunctions
$\varphi_n=\prod_ic_{n_i}(x_i)$, $c_0=s^{-1/2}$, $c_m=\sqrt{2/s}\cos(\pi mx/s)$, with eigenvalues $\mu_n=\pi^2\abs n^2/s^2$, give
the orthonormal basis $\nabla h_n$, $h_n:=\mu_n^{-1/2}\varphi_n$ ($n\in\N_0^3\setminus\{0\}$), of $\nabla H^1(Q)$, so
$\tau_Q=\sum_n\norm{\Pi'(\mathbf 1_Q\nabla h_n)}^2$. Let $\varrho(t):=\min\{1,t,s-t\}$ and $\chi(x):=\prod_i\varrho(x_i)$. Then $\chi h_n\in H^1_0(Q)$,
its extension by zero lies in $H^1(\R^3)$, and transverse projections annihilate gradients, so
$\Pi'(\mathbf 1_Q\nabla h_n)=\Pi'((1-\chi)\nabla h_n)-\Pi'(h_n\nabla\chi)$ and $\tau_Q\le2T_1+2T_2$ with the sums $T_1,T_2$ of the
squared norms of the two terms. First,
$T_1=\norm{\Pi'M_{(1-\chi)\mathbf 1_Q}P_{\nabla,Q}}^2_{\mathfrak S_2}\le\norm{\Pi'M_{(1-\chi)\mathbf 1_Q}}^2_{\mathfrak S_2}=\rhom\int_Q(1-\chi)^2\le\rhom\cdot6s^2=2\pi s^2$,
since $\int\norm{K(z)}^2_{\mathrm{HS}}dz=\tr K(0)=\rhom$ ($\Pi'$ is a projection) and $1-\chi$ vanishes at distance $\ge1$ from
$\partial Q$. Second, let $P$ be the band projection $\mathbf 1_{\abs k<\pi}$ acting on each component, so that $\Pi'\le P$,
and $D_\chi:=\int_Q\abs{\nabla\chi}^2\le6s^2$. Split the sum $T_2\le\sum_n\mu_n^{-1}\norm{P(\varphi_n\nabla\chi)}^2$ at $\abs n=s$. For
$\abs n\le s$ use $\norm{\varphi_n}_\infty^2\le8/s^3$; for $\abs n>s$ use $\mu_n^{-1}<\pi^{-2}$ and
$\sum_{n\in\N_0^3}\norm{P(\varphi_ng)}^2=\norm{PM_g\mathbf 1_Q}^2_{\mathfrak S_2}=\frac\pi6\int\abs g^2$. This gives
\[
  T_2\ \le\ \frac{8D_\chi}{\pi^2s}\,S(s)+\frac{D_\chi}{6\pi},\qquad S(N):=\sum_{0<\abs n\le N}\abs n^{-2}\le C_{\mathrm{lat}}N
\]
for $N\ge1$, with an absolute constant $C_{\mathrm{lat}}$ ($n\in\N_0^3$; at most $Ck^2$ of these $n$ satisfy $k-1<\abs n\le k$, and $\abs n\ge1$).
Hence $\tau_Q\le4\pi s^2+12s^2\bigl[\frac{8C_{\mathrm{lat}}}{\pi^2}+\frac1{6\pi}\bigr]$; with the case $s\le2$ this gives $\tau_Q\le C_D\abs{\partial Q}$
with an absolute constant $C_D$.
(b) Apply Lemma~\ref{lem:hardy} to $r'(q)=P'_Q\mathbf 1_Q\Pi'Xq$. Since $\norm{G'_Q}\le1$ and
$x\mapsto(T-2x)^2/x$ is decreasing on $(0,T/2]$, we have $g'^2/\norm{G'}\ge(T-2)^2$ with $T=\tr G'_Q$.
Moreover $h_Q\le\tr G'_Q\le\tr G_Q\le\rhom\abs Q$ by Lemma~\ref{lem:gram}(b).
(c) On $\R^3$: $\tr G_Q-\tr G'_Q=\Tr((1-P'_Q)Y_QY_Q^*(1-P'_Q))$ with $Y_QY_Q^*\le\mathbf 1_Q\Pi'\mathbf 1_Q$, so by (a),
Lemma~\ref{lem:gram}(b),(c),
$h_Q\ge\rhom s^3-L_Q-6C_Ds^2-2\ge\rhom s^3-C's^2\log(2+s)$, which gives the claim with $C_0:=C'/\rhom$ (for small $s$ the
right side is negative and the claim holds since $h_Q\ge0$).
On $\T_L$ we let $L\to\infty$ with $s$ fixed. Since the ball lies in the cubic band,
$\Pi'_L\sinc^{(L)}_a\mathbf e_i=K_L(\cdot-a)\mathbf e_i$ with $K_L(z)=L^{-3}\sum_{0<\abs\kappa<\pi}\mathrm e^{i\kappa\cdot z}P_\perp(\kappa)$, a Riemann sum
for $K(z)$; the midpoint rule on the cells $\kappa+[-\frac\pi L,\frac\pi L)^3$ gives
$\norm{K_L(z)-K(z)}\le C(L^{-1}+(1+\abs z)^2L^{-2})$ (the spherical shell has volume $O(1/L)$, and the cells near the
origin contribute $O(L^{-3})$), uniformly for $\abs z\le\sqrt3s$. For $L>s$ the cube $Q$ embeds isometrically
into $\T_L$, and $L^2(Q)$, $H^1(Q)$ and $P'_Q$ do not depend on $L$. Hence $\mathbf 1_Q\Pi'_L\sinc^{(L)}_a\mathbf e_i\to\mathbf 1_Q\Pi'\sinc_a\mathbf e_i$
uniformly on $Q$, so the entries of the $3s^3\times3s^3$ matrix $G'^{(L)}_Q$ converge to those of $G'_Q$ on $\R^3$, and
$h^{(L)}_Q\to h_Q$. Integer translations commute with $\Pi'_L$, $\Pi_{\mathrm{cube}}$ and $P'_Q$ (up to relabelling), so
$h^{(L)}_Q$ and $h_Q$ do not depend on the integer cube $Q$. If $(1-C\log(2+s)/s)_+=0$ there is nothing to prove;
otherwise, since $C>C_0$, $(1-C\log(2+s)/s)\frac\pi3s^3<h_Q$, and the convergence gives the bound for all odd
$L\ge L_{\mathrm G}(s)$.
\end{proof}

\begin{lemma}[Vacuum mismatch of the weight]\label{lem:weight}
On $\T_L$ let $0\le\theta\le\theta_0<1$, $\Theta:=X^*\Pi'(-\nabla\cdot\theta\nabla)\Pi'X$ on the physical
modes (so $0\le\Theta\le\theta_0\omega^2$) and $\omega_\theta^2:=\omega^2-\Theta$. Then
\[
  \Hf=\tfrac12\bigl(\abs p^2+q\cdot\omega_\theta^2q\bigr)-\tfrac12\Tr\omega_\theta+\tfrac12\int\theta\abs{\nabla B_q}^2-\mathcal D_\theta,\qquad
  0\le\mathcal D_\theta:=\tfrac12\Tr(\omega-\omega_\theta)\le\frac{\rho'^{L}_{\mathrm m}}{4\sqrt{1-\theta_0}}\int\theta,
\]
where $\rho'^L_{\mathrm m}:=2L^{-3}\sum_{0<\abs\kappa<\pi}\abs\kappa\le\pi\rhom^{\T_L}\le2\pi^2$ for $L\ge4$.
\end{lemma}

\begin{proof}
The identity is $q\cdot\Theta q=\int\theta\abs{\nabla B_q}^2$. By concavity of the square root,
$\frac12\Tr(\omega-\omega_\theta)\le\frac14\Tr[\Theta\omega_\theta^{-1}]$, and
$\omega_\theta^{-1}\le(1-\theta_0)^{-1/2}\omega^{-1}$ by operator monotonicity. Finally
$\Tr[\Theta\omega^{-1}]=\int\theta(y)\sum_j\abs{\nabla(\Pi'f_j)(y)}^2/\omega_j\,dy=\rho'^L_{\mathrm m}\int\theta$, the sum
running over the normalized physical modes $f_j$; the bound on $\rho'^L_{\mathrm m}$ follows
from \eqref{eq:rhom-torus}.
\end{proof}

\subsection{The Pauli--Fierz case}
For $\nu\in\N$, $\eta\in(0,1)$, $u_b>0$ and $\Upsilon>0$ consider the lower bound
\begin{equation}\tag{F}\label{eq:F}
  E_0^N\bigl(Q;\nu eB_q\bigr)\ \ge\ \nu^2(1-\eta)\,\frac{e^2\min\{u'_Q(q),u_b\}}{\abs Q}\qquad\text{whenever}\quad\int_{\T_L}\mathrm w_Q\abs{\nabla B_q}^2\le\Upsilon\,u'_Q(q).
\end{equation}
Theorem~\ref{thm:neumann} shows that \eqref{eq:F} holds for all $Q$ and $q$, with $u_b=N_b/(\nu^2e^2)$ and $N_b=c'\eta s/K_V$. Without the
gradient condition \eqref{eq:F} is false (Remark~\ref{rem:why-window}).

\begin{proposition}[Configuration-wise block Born--Oppenheimer bound]\label{prop:configBO}
There are absolute constants $s_0$, $C$ and $c'$ such that the following holds. Let $s\ge s_0$ be an integer,
$\nu\in\N$, $\eta\in(0,1)$, $\theta_0=\frac12$, $\gamma:=(1-\eta)\lambda e^2/2$ and $\Upsilon:=2C_w(1-\eta)\lambda e^2/s^3$, and let
$K_V$, $N_b=c'\eta s/K_V$ and $u_b:=N_b/(\nu^2e^2)$ be as in Theorem~\ref{thm:neumann} for this $\Upsilon$, so that
$K_V\le C(1+s+\log(1+\lambda e^2))^6$. Let $\mathcal C_{\mathrm b}$, $\mathcal P$ and the coordinates $q$ be as in Section~\ref{sec:finite-model}, with $L$ odd and
$L\ge\max(n_{\mathrm b}+8s,s^2,L_{\mathrm G}(s))$, and let $U\subset\mathcal C_{\mathrm b}$ be open. Let
$W_U(q):=\inf\{\frac\lambda2\int\abs{(\nabla+ieB_q)\psi}^2:\psi\in H^1_0(U),\ \norm\psi=1\}$, and for such $\psi$ let
\[
  J_q(\psi):=\Bigl(1-\frac1{\nu^2}\Bigr)T(\abs\psi)+\sum_{Q\in\mathcal P}\Bigl[\frac{h_Q^2}{8u'_Q(q)}+\gamma\,\bar\rho_Q(\psi)\min\{u'_Q(q),u_b\}\Bigr]-\mathcal V_0-\mathcal D_\theta ,
\]
where $\bar\rho_Q(\psi):=m_Q(\psi)/\abs Q$, $h_Q$ is as in Corollary~\ref{cor:hardy}(b) (with $h_Q^2/(8u'_Q):=+\infty$ if $u'_Q(q)=0<h_Q$ and $:=0$ if $h_Q=0$), $\mathcal V_0:=\frac12\Tr(P_{\mathrm{in}}\omega P_{\mathrm{in}})$, and
$\mathcal D_\theta$ is as in Lemma~\ref{lem:weight} with the weight $\theta$ of Section~\ref{sec:finite-model}. Then:
\begin{enumerate}[label=(\roman*)]
\item $\Hf+W_U\ge\mathcal J$ as forms on $L^2(\R^d,dq)$, where $\mathcal J$ is a measurable function of $q$ with
$\mathcal J(q)\ge\inf_\psi J_q(\psi)$ for every $q$;
\item $0\le\mathcal V_0\le\frac{3\sqrt3\pi}2n_{\mathrm b}^3$, $0\le\mathcal D_\theta\le Cn_{\mathrm b}^3$, and
$(1-C\log(2+s)/s)_+\frac\pi3\abs Q\le h_Q\le\pi\abs Q$ for all $Q\in\mathcal P$;
\item for all $q$, all $Q$ and all $x\ge0$,
$\frac{h_Q^2}{8u'_Q}+x\min\{u'_Q,u_b\}\ge h_Q\sqrt{x/2}-\mathbf 1_{\{x>0\}}\frac{h_Q^2}{2u_b}$; moreover $\sum_Qh_Q\sqrt{\gamma\bar\rho_Q(\psi)/2}\ge(1-\eta)^{1/2}(1-C\log(2+s)/s)_+\beta\norm\psi_1$ for every $\psi$, and consequently, for all $q$ and $\psi$,
\[
  J_q(\psi)\ \ge\ \Bigl(1-\frac1{\nu^2}\Bigr)T(\abs\psi)+(1-\eta)^{1/2}\Bigl(1-\frac{C\log(2+s)}s\Bigr)_+\beta\,\norm\psi_1-\mathcal V_0-\mathcal D_\theta-\sum_{Q\cap U\ne\emptyset}\frac{h_Q^2}{2u_b};
\]
\item $u'_Q$, $h_Q$, $\mathcal V_0$ and $\mathcal D_\theta$ do not change if the block box, the blocks, the sinc lattice and the
field are translated jointly by any $\tau\in\R^3$; in particular (i)--(iii) hold for block boxes with
corners in $\tau+\Z^3$ and the sinc points $\tau+\Z^3+(\frac12,\frac12,\frac12)$;
\item (average over translations of the blocks) let $\mu$ be a Borel probability measure on $\R^3$, let
$U\subset\bigcap_{\tau\in\supp\mu}(\mathcal C_{\mathrm b}+\tau)$ be open, and let $J^\tau_q(\psi)$ be defined as $J_q(\psi)$ with the block
box, the blocks and the sinc points translated by $\tau$ (the field and its coordinates $q$ are not
changed). Then $\Hf+W_U\ge\mathcal J_\mu$ as forms, where $\mathcal J_\mu$ is a measurable function with
$\mathcal J_\mu(q)\ge\inf_\psi\int J^\tau_q(\psi)\,d\mu(\tau)$ for every $q$. By (iv), $h_Q$, $\mathcal V_0$ and $\mathcal D_\theta$ in $J^\tau_q$ do not depend
on $\tau$, and (iii) applies to $J^\tau_q$ for each $\tau$.
\end{enumerate}
\end{proposition}

The bound is stated before the minimization over $\psi$, for each configuration $q$ of the field and
each $U$, because this is the form needed when the electron energy is studied with a
$U$-dependent localization \cite{SerBC}; parts (iv) and (v) are used only there. For the self-energy only the
following consequence is used.

\begin{corollary}[Block Born--Oppenheimer bound]\label{thm:blockBO}
In the setting of Proposition~\ref{prop:configBO} let $U=\mathcal C_{\mathrm b}=:\mathcal C_e$ (of side $n_e$),
$H_{\mathrm{BO}}:=\Hf+W_{\mathcal C_e}$ and $\kappa_s:=(1-\eta)^{1/2}(1-C\log(2+s)/s)_+$. Then
\[
  \inf\sigma(H_{\mathrm{BO}})\ \ge\ \min\Bigl(1-\frac1{\nu^2},\kappa_s\Bigr)F_\lambda[\beta]-\Bigl(\frac{3\sqrt3\pi}2+C\Bigr)n_e^3-\frac{\pi^2}{2}\,\frac{\nu^2e^2n_e^3s^3}{N_b}.
\]
\end{corollary}

\begin{proof}
By Proposition~\ref{prop:configBO}(i),(iii), $\inf\sigma(H_{\mathrm{BO}})$ is at least
$\inf_\psi[(1-\nu^{-2})T(\abs\psi)+\kappa_s\beta\norm\psi_1]-\mathcal V_0-\mathcal D_\theta-\sum_Qh_Q^2/(2u_b)$. The extension of $\abs\psi$ by zero lies in
$H^1\cap L^1(\R^3)$, and for $a,b\in[0,1]$ the scaling of Section~\ref{sec:model} gives
$F_{a\lambda}[b\beta]=a^{3/7}b^{4/7}F_\lambda[\beta]\ge\min(a,b)F_\lambda[\beta]$. For the errors use (ii), $\#\mathcal P=(n_e/s)^3$ and
$h_Q^2/(2u_b)\le\pi^2s^6\nu^2e^2/(2N_b)$.
\end{proof}

\begin{remark}[Why the lower bound \eqref{eq:F} is assumed only under a gradient condition]\label{rem:why-window}
Without the gradient condition the bound \eqref{eq:F} is false
(Proposition~\ref{prop:counterexample}). Configurations violating the condition are paid for
by the weight $\theta$ in Step~3 below, at the cost $\mathcal D_\theta$ of Lemma~\ref{lem:weight}.
\end{remark}

\subsection{Proof of Theorem~\ref{thm:abstractBO}}\label{sec:proof-abstractBO}
\emph{Step 1 (splitting the oscillator).} Since $\omega_\theta^2=\omega^2-\Theta\ge(1-\theta_0)\omega^2>0$,
we have the identity
$\Hf=\frac12(\abs p^2+q\cdot\omega_\theta^2q)-\frac12\Tr\omega_\theta+\frac12q\cdot\Theta q-\mathcal D_\theta$. Apply
Lemma~\ref{lem:vacuum} with $\omega_\theta$ and $\mathcal V:=\mathcal W+\frac12q\cdot\Theta q$, and use
$\omega_\theta\le\omega$ (operator monotonicity of the square root), so that
$\Tr(P_{\mathrm{in}}\omega_\theta P_{\mathrm{in}})\le\Tr(P_{\mathrm{in}}\omega P_{\mathrm{in}})$:
\[
  \Hf+\mathcal W\ \ge\ \tfrac12\abs{P_{\mathrm{in}}p}^2+\mathcal W(q)+\tfrac12q\cdot\Theta q-\tfrac12\Tr(P_{\mathrm{in}}\omega P_{\mathrm{in}})-\mathcal D_\theta .
\]
Thus (ii) follows from (i). (For the size of $\mathcal D_\theta$ see Lemma~\ref{lem:weight}.)

\emph{Step 2 (Hardy inequality in each block).} By (H2) and Lemma~\ref{lem:hardy} (with
$\norm{G'_Q}\le1$, so $g'^2_Q/\norm{G'_Q}\ge h_Q^2$ as in the proof of
Corollary~\ref{cor:hardy}), applied on the mutually orthogonal $S_Q$,
$\frac12\abs{P_{\mathrm{in}}p}^2\ge\sum_Qh_Q^2/(8u_Q(q))$; the zero set of $u_Q$ is a subspace, outside
which the form inequality holds. Hence the left side of (i) is bounded below by
multiplication by
$\mathcal B(q):=\sum_Q\frac{h_Q^2}{8u_Q}+\mathcal W+c_\theta\sum_Q\mathcal E_Q$.

\emph{Step 3 (case distinction).} By (H4),
$\mathcal B(q)\ge\inf_\varphi[(1-\nu^{-2})T+\sum_Q(\frac{h_Q^2}{8u_Q}+m_Qf_Q+c_\theta\mathcal E_Q)]$. For each $Q$,
$m_Qf_Q+c_\theta\mathcal E_Q\ge\gamma\bar\rho_Q\min(u_Q,u_b)$: if $\mathcal E_Q\le\Upsilon u_Q$, the first term does this
by (H4); otherwise the second term exceeds $c_\theta\Upsilon u_Q\ge\gamma u_Q/\abs Q\ge\gamma\bar\rho_Qu_Q$
(as $m_Q\le1$). The case distinction depends only on $q$ and $Q$, not on $\varphi$, so it may
be used inside the infimum over $\varphi$.

\emph{Step 4 (completing the square).} For each $Q$ and $u=u_Q(q)$: if $u\le u_b$,
$\frac{h^2}{8u}+\gamma\bar\rho u\ge h\sqrt{\gamma\bar\rho/2}$ by the arithmetic--geometric mean inequality; if
$u\ge u_b$, then $\frac{h^2}{8u}+\gamma\bar\rho u_b\ge\gamma\bar\rho u_b\ge h\sqrt{\gamma\bar\rho/2}-\frac{h^2}{2u_b}$, because with
$\tau=\sqrt{\gamma\bar\rho}$ one has
$\tau^2u_b-\frac{h\tau}{\sqrt2}+\frac{h^2}{2u_b}=u_b(\tau-\frac{h}{2\sqrt2u_b})^2+\frac{3h^2}{8u_b}\ge0$. Hence
$\mathcal B(q)\ge Y$ for all $q$, which proves (i). \qed

\subsection{Proof of Proposition~\ref{prop:configBO}}
We verify (H1)--(H4) in the model of Section~\ref{sec:finite-model}. (H1): $S_Q$ are the sinc
coordinates in $Q$, mutually orthogonal. (H2): $r_Q(q):=P'_Q\mathbf 1_Q\Pi'Xq$, $u_Q=u'_Q$, and
$\norm{G'_Q}\le1$ as a composition of contractions. (H3): $\Theta:=X^*\Pi'(-\nabla\cdot\theta\nabla)\Pi'X$ satisfies
$0\le\Theta\le\theta_0\omega^2$ (Lemma~\ref{lem:weight}) and
$\frac12q\cdot\Theta q=\frac12\int\theta\abs{\nabla B_q}^2=\frac{\theta_0}{2C_w}\sum_Q\mathcal E_Q$ with
$\mathcal E_Q(q):=\int_{\T_L}\mathrm w_Q\abs{\nabla B_q}^2$ and $c_\theta=\frac{\theta_0}{2C_w}$. (H4): $\mathcal W=W_U$, which is upper
semicontinuous in $q$ (an infimum of continuous functions) and hence measurable;
$\mathfrak S$ the normalized $\psi\in H^1_0(U)$, extended by zero to $\mathcal C_{\mathrm b}$; $T=T(\abs\psi)$, $m_Q=m_Q(\psi)$ and
$f_Q(q)=\frac{\lambda}{2\nu^2}E_0^N(Q;\nu eB_q)$ (Lemma~\ref{lem:phase}(b)). Since $\frac{\lambda}{2\nu^2}\cdot\nu^2(1-\eta)e^2=\gamma$,
\eqref{eq:F}, which holds by Theorem~\ref{thm:neumann}, is exactly the required lower bound, and
$c_\theta\Upsilon=\gamma/s^3=\gamma/\abs Q$.

(i) Steps~1--3 of Section~\ref{sec:proof-abstractBO} give $\Hf+W_U\ge\mathcal B(q)-\mathcal V_0-\mathcal D_\theta=:\mathcal J(q)$ and
$\mathcal B(q)\ge\inf_\psi J_q(\psi)+\mathcal V_0+\mathcal D_\theta$.
(ii) $\Tr P_{\mathrm{in}}=3n_{\mathrm b}^3$ and $\norm\omega\le\sqrt3\pi$ give the bound on $\mathcal V_0$. By Lemma~\ref{lem:weight},
$\mathcal D_\theta\le C\int\theta\le C\sum_Q\int_{\T_L}\mathrm w_Q\le C(n_{\mathrm b}/s)^3s^3$. The bounds on $h_Q$ are
Corollary~\ref{cor:hardy}(b),(c), with $\rhom^{\T_L}\le\frac\pi3(1+\frac{\sqrt3}4)^3<\pi$ by \eqref{eq:rhom-torus}.
(iii) The inequality is Step~4 of Section~\ref{sec:proof-abstractBO} (trivial if $x=0$). For the consequence take
$x=\gamma\bar\rho_Q(\psi)$ and sum over $Q$; only cubes with $m_Q(\psi)>0$, which meet $U$, carry the error term.
By (ii) and the Cauchy--Schwarz inequality $\abs Q^{1/2}m_Q(\psi)^{1/2}\ge\int_Q\abs\psi$,
\begin{align*}
  \sum_Qh_Q\sqrt{\gamma\bar\rho_Q/2}\ &\ge\ (1-\eta)^{1/2}\Bigl(1-\frac{C\log(2+s)}s\Bigr)_+\frac\pi3\cdot\frac{e\sqrt\lambda}2\sum_Q\abs Q^{1/2}m_Q^{1/2}\\
  &\ge\ (1-\eta)^{1/2}\Bigl(1-\frac{C\log(2+s)}s\Bigr)_+\beta\norm\psi_1 ,
\end{align*}
since $\frac\pi3\cdot\frac{e\sqrt\lambda}{2}=\frac\pi6e\sqrt\lambda=\beta$.
(iv) Joint translations by $\tau$ are unitary on $L^2(\T_L;\C^3)$, commute with $\Pi'$, $\Pi_{\mathrm{cube}}$ and $\omega$, and map
$\sinc_a$, $Q$, $P'_Q$ and $\mathrm w_Q$ to $\sinc_{a+\tau}$, $Q+\tau$, $P'_{Q+\tau}$ and $\mathrm w_{Q+\tau}$.
(v) Since translations commute with $\Pi_{\mathrm{cube}}$ and $L$ is odd, the translated functions $\sinc_{a+\tau}\mathbf e_i$ lie in
the cubic band and are real, so the translated block coordinates are subspaces $S^\tau_Q$ of the same $\R^d$,
and the translated weight gives a matrix $\Theta^\tau$ on the same $\R^d$. For each $\tau$, Steps~1 and~2 of
Section~\ref{sec:proof-abstractBO} (which use only the field, Lemma~\ref{lem:vacuum} being valid for any orthogonal
decomposition) give
$\Hf+W_U\ge W_U+\mathcal A_\tau(q)-\mathcal V_0-\mathcal D_\theta$ as forms, with
$\mathcal A_\tau:=\sum_Q\bigl(\frac{h_Q^2}{8u'_{Q+\tau}}+c_\theta\mathcal E_{Q+\tau}\bigr)\ge0$. Since $B_q$ is a trigonometric polynomial, $u'_{Q+\tau}(q)$ and
$\mathcal E_{Q+\tau}(q)$ are continuous in $(\tau,q)$, so $\mathcal A_\tau(q)$ is jointly measurable with values in $[0,\infty]$, and
integrating the form inequalities against $\mu$ (Tonelli) gives $\Hf+W_U\ge W_U+\int\mathcal A_\tau\,d\mu-\mathcal V_0-\mathcal D_\theta=:\mathcal J_\mu$. For
fixed $\psi$ and $q$, the left side of Lemma~\ref{lem:phase}(b), applied with the blocks $Q+\tau$, does not depend
on $\tau$; averaging over $\tau$ gives
$W_U(q)\ge\inf_\psi\bigl[(1-\nu^{-2})T(\abs\psi)+\int\sum_Qm_{Q+\tau}(\psi)f_{Q+\tau}(q)\,d\mu(\tau)\bigr]$. Hence
$\mathcal J_\mu(q)\ge\inf_\psi\int[(1-\nu^{-2})T(\abs\psi)+\mathcal A_\tau(q)+\sum_Qm_{Q+\tau}f_{Q+\tau}(q)]\,d\mu(\tau)-\mathcal V_0-\mathcal D_\theta$, and the case
distinction of Step~3, applied for each $\tau$ and $Q$, bounds the integrand by $J^\tau_q(\psi)+\mathcal V_0+\mathcal D_\theta$. \qed

%% file: sec7_neumann.tex
\section{A lower bound for the magnetic Neumann Laplacian on a cube}\label{sec:neumann}

This section concerns the magnetic Neumann Laplacian on a single cube and does not involve the
electron model. The proof has two parts. The first, and the only delicate one, is a
non-concentration estimate (Section~\ref{sec:sup}): under the gradient condition the vorticity on
the cube is controlled by the gauge-projected field mass, with a loss polynomial in $s+\log\Upsilon$.
The second (Section~\ref{sec:floor}) is elementary: weak fields are handled by splitting a function
into its mean and its mean-zero part, large fields by a dual test field that removes the vortices
of the ground state, with the Poincar\'e, Sobolev--Poincar\'e and H\"older inequalities and Fourier
series on a reflection torus as the only tools.

\subsection{Statement}
Let $B$ be a real field on $\T_L$ built from the physical modes (a transverse
trigonometric polynomial with frequencies $\abs k<\pi$), let $Q$ be a cube in $\T_L$ with
side $s$, $1\le s\le L$, and centre $c$, and set
$u'_Q:=\norm{P'_Q\mathbf 1_QB}^2$, $a':=P'_Q\mathbf 1_QB$, $\varpi:=\curl B$ and
$\Upsilon_Q(B):=\int_{\T_L}\mathrm w_Q\abs{\nabla B}^2/u'_Q$.

\begin{theorem}[Neumann lower bound]\label{thm:neumann}
Let $C_1\ge2$ be the absolute constant of Theorem~\ref{thm:sup} below. There is an absolute constant $c'>0$ such that the following holds. Let
$\Upsilon>0$, $s\ge1$, $\nu\in\N$, $\eta\in(0,1]$ and $e>0$, and set
\[
  D_\Upsilon:=\bigl\lceil C_1(1+s+\log(1+\Upsilon))\bigr\rceil,\qquad K_V:=17D_\Upsilon^6,\qquad
  N_b:=c'\,\frac{\eta s}{K_V},\qquad u_b:=\frac{N_b}{\nu^2e^2}.
\]
Then every $B$ with $\Upsilon_Q(B)\le\Upsilon$ satisfies
$E_0^N(Q;\nu eB)\ge\nu^2(1-\eta)e^2\min\{u'_Q,u_b\}/\abs Q$. That is, the assumption \eqref{eq:F}
of the configuration-wise block Born--Oppenheimer bound (Proposition~\ref{prop:configBO}) holds with $(\nu,\eta,u_b,\Upsilon)$.
\end{theorem}

The statement is invariant under multiplication of the field by a constant: with
$N:=\nu^2e^2u'_Q$ and $\mathcal X:=\abs QE_0^N(Q;\nu eB)$ it reads
$\mathcal X\ge(1-\eta)\min(N,N_b)$. The loss $K_V$ is polynomial in $s+\log\Upsilon$; this is what
the proof of Theorem~\ref{thm:main}(b) can afford (Remark~\ref{rem:new}).

\subsection{Non-concentration of band-limited fields}\label{sec:sup}
\begin{lemma}[Reflections and the gauge projection]\label{lem:reflection}
Let $Q=\prod_i(0,s_i)$, $T=\prod_i\R/2s_i\Z$ and $\mathcal G_0=\{\mathrm{diag}(\pm1,\pm1,\pm1)\}$. For
vector fields set $(\mathcal Rf)(x):=Jf(Jx)$ and for pseudovector fields
$(\mathcal R^*f)(x):=\det(J)\,Jf(Jx)$, for $x\in JQ$, $J\in\mathcal G_0$. For $B\in C^1(\bar Q)$:
(a) $\mathcal Ra'\in L^2(T)$ is divergence free, has mean zero, and $\curl\mathcal Ra'=\mathcal R^*\varpi$
(without surface layers);
(b) $\mathcal Ra'=\curl(-\Delta_T)^{-1}\mathcal R^*\varpi$;
(c) $u'_Q=\sum_i\ip{\varpi_i}{L_i^{-1}\varpi_i}_{L^2(Q)}$, where $L_i$ is $-\Delta$ with Neumann conditions on
the faces $x_i\in\{0,s_i\}$ and Dirichlet conditions on the other four faces, and
$L_i\ge2\pi^2/s_{\max}^2$.
\end{lemma}

\begin{proof}
(a) For a test field $F$ let $F_J(y)=JF(Jy)$, so that $\curl F_J=\det(J)\,J(\curl F)(J\cdot)$. In the
integration by parts, the surface terms $\int_{\partial Q}n\cdot(F_J\times B)$ cancel: on the face
$y_1=0$ pair $J$ with $J\,J_1$, $J_1=\mathrm{diag}(-1,1,1)$; the determinants have opposite signs and
$\mathbf e_1\cdot(J_1F_J\times B)=\mathbf e_1\cdot(F_J\times B)$ (similarly for the other faces and the periodic
faces). For the gradient part, $\mathcal R\nabla\chi=\nabla E\chi$ with $E$ the even extension, which is
curl free. Divergence:
$\ip{\mathcal Ra'}{\nabla\psi}=\sum_J\int_Qa'\cdot\nabla(\psi(J\cdot))=0$ since $a'\perp\nabla H^1(Q)$; mean: $\sum_JJ=0$.
(b) A field with zero divergence, curl and mean vanishes (Fourier series).
(c) $\norm{\mathcal Ra'}^2=8u'_Q=\ip{\mathcal R^*\varpi}{(-\Delta_T)^{-1}\mathcal R^*\varpi}$, and on the subspace of the
appropriate parity of the $i$-th component $-\Delta_T$ is unitarily equivalent to $L_i$.
\end{proof}

\emph{Why non-concentration is needed, and why it is hard.} A lower bound of the form
$\mathcal X\ge(1-\eta)\min(N,N_b)$ is false for general band-limited fields: a field may be
superoscillating, with large amplitude outside $Q$ and vorticity concentrated in a small ball
inside $Q$, and then $E_0^N$ is smaller than $\hat\mu$ by an arbitrary factor
(Appendix~\ref{app:counterexample}). There the vorticity concentrates; conversely, a bound on the
vorticity on $Q$ by $u'_Q$ suffices (Proposition~\ref{prop:floor}). The weighted gradient gives one
immediately, $\norm\varpi^2_{L^2(Q)}\le2\Upsilon u'_Q$, but
in Section~\ref{sec:proof} the parameter $\Upsilon$ grows like a power of the cutoff, and a loss
polynomial in $\Upsilon$ is not affordable at fixed coupling. The point of the next theorem is that
the loss is only polynomial in $s+\log\Upsilon$. It uses analyticity: band limitation makes the
field entire, so that on $Q$ the vorticity is a polynomial of degree $D\approx s+\log\Upsilon$ up to
an error $2^{-D}$, and polynomials of fixed degree satisfy an inverse inequality for the operators
$L_i^{-1}$ of Lemma~\ref{lem:reflection}(c). The combination of analyticity of band-limited functions on a
complex neighbourhood with a polynomial-type bound is also the mechanism of Kovrijkine's proof of the
Logvinenko--Sereda theorem \cite{Kov01}; here it is applied to the negative order quadratic form
$\sum_i\ip{\varpi_i}{L_i^{-1}\varpi_i}$, under a weighted gradient condition instead of a global $L^2$ bound.

\begin{theorem}[Non-concentration]\label{thm:sup}
There is an absolute constant $C_1\ge2$ such that the following holds. Let $Q$ be a cube of side $s$ in
$\T_L$, $L\ge s\ge1$, let $B$ be as in Theorem~\ref{thm:neumann} with $\Upsilon_Q(B)\le\Upsilon$ and $u'_Q>0$, and
let $D\ge C_1(1+s+\log(1+\Upsilon))$ be an integer. Then
\[
  \norm\varpi_{L^2(Q)}\ \le\ \frac{34D^6}{s}\,(u'_Q)^{1/2}.
\]
\end{theorem}

\begin{proof}
(o) (Choice of $D$.) Let $A'_{\mathfrak c}$ be the absolute constant of Step~(ii) below,
$A'_0:=\frac{144\sqrt6}{\pi}A'_{\mathfrak c}$ and $\Gamma_D:=2(2\sqrt3D^3+8D^6)$, so that $\Gamma_D\le17D^6$ for $D\ge2$.
Since $D^6(D+2)^2\le C2^{D/2}$, there is an absolute constant $C_1\ge2$ such that every integer
$D\ge C_1(1+s+\log(1+\Upsilon))$ satisfies
\begin{equation}\label{eq:Dcond}
  2^D\ \ge\ A'_0\,D^6(D+2)^2\,s^{5/2}(1+\sqrt3s)^{10}\exp(2\pi\sqrt3s)\,\Upsilon^{1/2}.
\end{equation}
The factor $\exp(2\pi\sqrt3s)$ is the growth of the entire extension of the field on a polydisc of radius
$s$ (the smoothed band has radius $2\pi$), and $(1+\sqrt3s)^{10}$ comes from the polynomial weight; only
the logarithms of these factors enter $D$.

(i) (Sum over images.) Regard $B$ as a periodic function on $\R^3$ and let
$W(y):=(1+\abs{y-c})^{-8}$. For $y$ in the fundamental domain $c+[-L/2,L/2)^3$ the term $m=0$
satisfies $W\le\mathrm w_Q$, and for $m\ne0$, $\abs{y+mL-c}\ge\abs m_\infty L/2$, so
$\sum_{m\ne0}W(y+mL)\le C(1+L)^{-8}\le C\mathrm w_Q(y)$. Hence
$\int_{\R^3}W\abs{\nabla B}^2\le C_{\mathrm{im}}\int_{\T_L}\mathrm w_Q\abs{\nabla B}^2$ with an absolute constant
$C_{\mathrm{im}}$, independent of $L$ and $s$.

(ii) (Complex neighbourhood.) Let $\hat{\mathfrak c}\in C_c^\infty(\{\abs\xi\le2\pi\})$ be equal to $1$ on
$\{\abs\xi\le\frac32\pi\}$. Then $\partial_jB_k=\mathfrak c*\partial_jB_k$ (convolution on $\R^3$) extends to an
entire function, and integrating by parts with $(1-\Delta_\xi)^3$, applied to $\mathrm e^{i\xi\cdot v}$ only (the factor $\mathrm e^{-\xi\cdot\zeta}$ is kept with $\hat{\mathfrak c}$), gives
$\abs{\mathfrak c(v+i\zeta)}\le C'_{\mathfrak c}\exp(2\pi\abs\zeta)(1+\abs\zeta)^6(1+\abs v)^{-6}$. On the polydisc
$\mathbb D:=\{z\in\C^3:\abs{z_j-c_j}\le s\}$, the Cauchy--Schwarz inequality with weight $W$ and
$\int(1+\abs v)^{-4}dv=\frac{4\pi}{3}$ give
$M_{\mathrm{ext}}:=\max_i\sup_{\mathbb D}\abs{\varpi_i}\le A'_{\mathfrak c}\exp(2\pi\sqrt3s)(1+\sqrt3s)^{10}(\Upsilon u'_Q)^{1/2}$, where
$A'_{\mathfrak c}:=2C'_{\mathfrak c}(4\pi C_{\mathrm{im}}/3)^{1/2}$.

(iii) (Taylor approximation.) Let $P_i$ be the Taylor polynomial of $\varpi_i$ at $c$ of degree
$<D$. Cauchy's estimates give $\sup_Q\abs{\varpi-P}\le6\sqrt3(D+2)^22^{-D}M_{\mathrm{ext}}$, and
$\epsilon:=s^{3/2}\sup_Q\abs{\varpi-P}\ge\norm{\varpi-P}_{L^2(Q)}$.

(iv) (Three polynomial inequalities.) For polynomials of degree $\le n$ in each variable let
$d_n=(n+1)^3$ and $\tau_n=2\sqrt3n^2(n+1)$. Then
(N$_\infty$) $\norm P_\infty\le d_ns^{-3/2}\norm P_2$ (the diagonal of the reproducing kernel of
Legendre polynomials is $\le(n+1)^2/s$ in one variable, and $\abs{L_k}\le1$);
(M$_2$) $\norm{\nabla P}_2\le\tau_ns^{-1}\norm P_2$ (on each line parallel to a coordinate axis, from A.~Markov's inequality
$\norm{q'}_\infty\le2n^2s^{-1}\norm q_\infty$ on an interval of length $s$, see e.g.\ \cite{BE95}, and the
one-variable case $\norm q_\infty^2\le(n+1)^2s^{-1}\norm q_2^2$ of (N$_\infty$));
(H) $\ip{P}{L_i^{-1}P}\ge s^2\gamma_n^{-2}\norm P^2$ with $\gamma_n=2(\tau_n+8d_n^2)$ (take a cutoff $\chi$ of
width $s/(8d_n^2)$ at the four faces on which $L_i$ has Dirichlet conditions, so that $\ip P{\chi P}\ge\frac12\norm P^2$ and
$\norm{\nabla(\chi P)}\le(\tau_n+8d_n^2)\norm P/s$, and use the Cauchy--Schwarz inequality).

(v) (Assembly, $n=D-1$.) By Lemma~\ref{lem:reflection}(c) and $L_i\ge2\pi^2/s^2$,
$u'_Q\le\frac{s^2}{2\pi^2}\norm\varpi^2$, and \eqref{eq:Dcond} gives
$\epsilon\le\norm\varpi/(48D^6)$, where $48D^6\ge2\Gamma_D+2\ge2(1+\Gamma_D/(\sqrt2\pi))$. By (H), $\gamma_{D-1}\le\Gamma_D$ and the triangle inequality,
$u'^{1/2}_Q\ge\frac{s}{\Gamma_D}(\norm\varpi-\epsilon)-\frac{s\epsilon}{\sqrt2\pi}\ge\frac{s}{2\Gamma_D}\norm\varpi\ge\frac{s}{34D^6}\norm\varpi$, which is the claim.
\end{proof}

\subsection{A Neumann floor from the vorticity}\label{sec:floor}
In this subsection $Q$ is a cube of side $s$, $a\in C^1(\bar Q;\R^3)$, $a':=P'_Q\mathbf 1_Qa$,
$\varpi:=\curl a$, $N:=\norm{a'}^2_{L^2(Q)}$, $\hat\mu:=N/\abs Q$ and $\mathcal X:=\abs QE_0^N(Q;a)$. For a real field
$b\in L^3(Q;\R^3)$ let
\[
  E_0^N(Q;b):=\inf\Bigl\{\int_Q\abs{(\nabla+ib)v}^2:\ v\in H^1(Q),\ \norm v_{L^2(Q)}=1\Bigr\},
\]
which is finite since $\norm{bv}_2\le\norm b_3\norm v_6$; for bounded $b$ it is the bottom of the spectrum of
the magnetic Neumann Laplacian. It is gauge invariant: $a-a'=\nabla\chi$ with $\nabla\chi\in L^6(Q)$ ($a$ is bounded, and $a'\in L^6(Q)$ by
Lemma~\ref{lem:L6} below), and $\mathrm e^{i\chi}$ preserves $H^1(Q)$ since $\norm{v\nabla\chi}_2\le\norm v_3\norm{\nabla\chi}_6$, so $E_0^N(Q;a)=E_0^N(Q;a')$. Let
$S_*$ be the scale invariant Sobolev--Poincar\'e constant of cubes,
\begin{equation}\label{eq:Sstar}
  \norm{f-\bar f}_{L^6(Q)}\ \le\ S_*\norm{\nabla f}_{L^2(Q)}\qquad(f\in H^1(Q),\ \bar f:=\abs Q^{-1}\textstyle\int_Qf).
\end{equation}
It holds with the same constant for $\C^m$-valued $f$, since
$\norm f_6^2=\norm{\sum_j\abs{f_j}^2}_3\le\sum_j\norm{f_j}_6^2$ for the real components $f_j$; and
$S_*\ge1/\pi$ (take $f=\cos(\pi x_1/s)$ and use $\norm f_2\le s\norm f_6$).

\begin{lemma}[$L^6$ bound]\label{lem:L6}
$\norm{a'}_{L^6(Q)}\le2S_*\norm\varpi_{L^2(Q)}$ and $\norm{a'}_{L^3(Q)}^2\le\norm{a'}_{L^2(Q)}\norm{a'}_{L^6(Q)}$.
\end{lemma}

\begin{proof}
Let $u:=\mathcal Ra'$ on $T=(\R/2s\Z)^3$. By Lemma~\ref{lem:reflection}(a), $u$ has zero divergence and
mean, and $\curl u=\mathcal R^*\varpi$; for the Fourier coefficients $\abs k^2\abs{\hat u_k}^2=\abs{k\times\hat u_k}^2$,
so $\norm{\nabla u}^2_{L^2(T)}=\norm{\mathcal R^*\varpi}^2_{L^2(T)}=8\norm\varpi^2_{L^2(Q)}$. The restriction of $u$ to the
fundamental cube $(-s,s)^3$ lies in $H^1$ and has mean zero, so \eqref{eq:Sstar} on this cube gives
$\norm u_{L^6}\le\sqrt8S_*\norm\varpi_{L^2(Q)}$, and $\norm u^6_{L^6((-s,s)^3)}=8\norm{a'}^6_{L^6(Q)}$; hence
$\norm{a'}_6\le8^{-1/6}\sqrt8S_*\norm\varpi=2S_*\norm\varpi$. The second claim is H\"older's inequality.
\end{proof}

\begin{lemma}[Weak fields]\label{lem:weakfloor}
Let $b\in L^3(Q;\R^3)$ be real with $b\perp\nabla H^1(Q)$, $\hat\mu_b:=\norm b_2^2/\abs Q$ and $\eta\in(0,1]$. If
$S_*^2\norm b^2_{L^3(Q)}\le\eta/8$, then $E_0^N(Q;b)\ge(1-\eta)\hat\mu_b$.
\end{lemma}

\begin{proof}
Let $t:=S_*^2\norm b_3^2$ and $V:=\abs b^2$. Write $v\in H^1(Q)$ with $\norm v=1$ as $v=c+f$,
$c$ constant and $\int_Qf=0$, so that $\abs c^2\abs Q=1-\norm f^2$. Then
\[
  \int_Q\abs{(\nabla+ib)v}^2=\norm{\nabla f}^2+\int_QV\abs{c+f}^2+2\,\mathrm{Re}\int_Q\overline{\nabla f}\cdot ib(c+f).
\]
The cross term containing $c$ vanishes, since $b$ is real and orthogonal to $\nabla\mathrm{Re}f$ and
$\nabla\mathrm{Im}f$. By H\"older's inequality and \eqref{eq:Sstar}, the other cross term is at least
$-2\norm b_3\norm f_6\norm{\nabla f}\ge-2\sqrt t\norm{\nabla f}^2$, and $\int V\abs f^2\le\norm b_3^2\norm f_6^2\le t\norm{\nabla f}^2$.
Moreover $\abs{c+f}^2\ge(1-\eta)\abs c^2-(\eta^{-1}-1)\abs f^2$, and, by the Neumann gap
$\norm{\nabla f}^2\ge(\pi^2/s^2)\norm f^2$ and $\hat\mu_bs^2/\pi^2\le\norm b_3^2/\pi^2\le t$ (H\"older and $S_*\ge1/\pi$),
$\abs c^2\norm b_2^2=\hat\mu_b(1-\norm f^2)\ge\hat\mu_b-t\norm{\nabla f}^2$. Altogether
\[
  \int_Q\abs{(\nabla+ib)v}^2\ \ge\ (1-\eta)\hat\mu_b+\Bigl(1-2\sqrt t-\frac t\eta+\eta t\Bigr)\norm{\nabla f}^2,
\]
and for $t\le\eta/8$ the bracket is $\ge1-2/\sqrt8-1/8>0$.
\end{proof}

The proof uses only that the first order (paramagnetic) term vanishes for the constant
part of $v$, because $b\perp\nabla H^1(Q)$; no bound on the supremum of $b$ is needed.

\subsubsection*{Large fields: a dual test field}
Let $\tilde a:=\mathcal Ra'$ on $T=(\R/2s\Z)^3$ (the $j$-th component is odd in $x_j$ and even in the other
variables), $\tilde{\mathbf g}:=\curl(-\Delta_T)^{-1}\tilde a$, $\mathbf g_0:=\tilde{\mathbf g}\restriction_Q$ and $\Phi^*:=\sup_Q\abs{\mathbf g_0}$.

\begin{lemma}[Agmon's inequality on the torus, cf.\ \cite{Agm65}]\label{lem:agmon}
Let $u\in H^2(T;\C^3)$ have mean zero. Then the Fourier series of $u$ converges absolutely, $u$ is
continuous, and
\[
  \norm u^2_{L^\infty(T)}\ \le\ C_A\norm{\nabla u}_{L^2(T)}\norm{\Delta u}_{L^2(T)},
\]
where
\[
  C_A:=\frac{c_3}{4\pi^3},\qquad c_3:=\sup_{X\ge1}\sum_{n\in\Z^3\setminus\{0\}}\frac{X}{\abs n^2(X^2+\abs n^2)}<\infty
\]
do not depend on $s$.
\end{lemma}

\begin{proof}
Write $u=\sum_k\hat u_k\mathrm e^{ik\cdot x}$, $k\in(\pi/s)\Z^3$, $\hat u_0=0$; then $\norm{\nabla u}^2=\abs T\sum\abs k^2\abs{\hat u_k}^2$ and
$\norm{\Delta u}^2=\abs T\sum\abs k^4\abs{\hat u_k}^2$, $\abs T=8s^3$, and $\abs{u(x)}\le\sum_k\abs{\hat u_k}$. For $\mu>0$ the
Cauchy--Schwarz inequality with the weights $w_k:=\mu\abs k^2+\abs k^4$ gives
\[
  \Bigl(\sum_k\abs{\hat u_k}\Bigr)^2\le\frac{\mu\norm{\nabla u}^2+\norm{\Delta u}^2}{8s^3}\cdot\frac{s^4}{\pi^4}\sum_{n\ne0}\frac{1}{\abs n^2(X^2+\abs n^2)},\qquad
  X:=\frac{\sqrt\mu\,s}{\pi}.
\]
Choose $\sqrt\mu:=\norm{\Delta u}/\norm{\nabla u}$ (if $\nabla u=0$ then $u=0$). Since $\abs k\ge\pi/s$ for $k\ne0$,
$\norm{\Delta u}\ge(\pi/s)\norm{\nabla u}$, i.e.\ $X\ge1$, so the last sum is $\le c_3/X$ and the right side is
$\le\frac{c_3}{4\pi^3}\norm{\nabla u}\norm{\Delta u}$. Finally $c_3<\infty$: for $X\ge1$ the terms with $\abs n\le X$ contribute at most
$X^{-1}\sum_{0<\abs n\le X}\abs n^{-2}\le C$, and those with $\abs n>X$ at most $X\sum_{\abs n>X}\abs n^{-4}\le C$.
\end{proof}

\begin{lemma}\label{lem:dual}
(a) $\tilde{\mathbf g}\in H^2(T)\cap C(T)$, $\curl\tilde{\mathbf g}=\tilde a$, $\norm{\nabla\tilde{\mathbf g}}_{L^2(T)}=\sqrt{8N}$ and
$\norm{\Delta\tilde{\mathbf g}}_{L^2(T)}=\sqrt8\norm\varpi_{L^2(Q)}$; hence $\Phi^{*2}\le8C_A\sqrt N\norm\varpi_{L^2(Q)}$.
(b) On each face of $Q$, $n\times\mathbf g_0=0$.
\end{lemma}

\begin{proof}
(a) $\tilde a$ has zero mean and divergence (Lemma~\ref{lem:reflection}(a)), so
$\hat{\tilde{\mathbf g}}_k=ik\times\hat{\tilde a}_k/\abs k^2$ and $\abs k\abs{\hat{\tilde{\mathbf g}}_k}=\abs{\hat{\tilde a}_k}$; thus
$\curl\tilde{\mathbf g}=-\Delta(-\Delta_T)^{-1}\tilde a=\tilde a$, $\norm{\nabla\tilde{\mathbf g}}=\norm{\tilde a}=\sqrt{8N}$, and
$\norm{\Delta\tilde{\mathbf g}}=\norm{\curl\tilde a}=\norm{\mathcal R^*\varpi}=\sqrt8\norm\varpi$ by Lemma~\ref{lem:reflection}(a). Continuity and
the bound on $\Phi^*$ follow from Lemma~\ref{lem:agmon}, since $\tilde{\mathbf g}$ has mean zero.
(b) From $\tilde a(Jx)=J\tilde a(x)$ we get $\tilde{\mathbf g}(Jx)=\det(J)J\tilde{\mathbf g}(x)$; points of the face $x_1=0$ are fixed
by $J_1$, so the tangential components vanish.
\end{proof}

\begin{lemma}[A phase-free current identity]\label{lem:current}
Let $Q$ be a cube, $\psi\in H^1(Q)$, $\rho:=\abs\psi$, $c>0$ and $\mathcal K_\rho:=\{\rho\ge c\}$. Let $\mathbf G\in H^1(Q;\R^3)\cap L^\infty(Q;\R^3)$
satisfy $n\times\mathbf G=0$ on the faces of $Q$ in the trace sense, and let $\mathbf G$ and $\curl\mathbf G$ vanish a.e.\ outside $\mathcal K_\rho$.
Then $\int_{\mathcal K_\rho}\rho^{-2}\,\mathrm{Im}(\bar\psi\nabla\psi)\cdot\curl\mathbf G=0$.
\end{lemma}

\begin{proof}
The integrand is integrable, as $\abs{\mathrm{Im}(\bar\psi\nabla\psi)}/\rho^2\le\abs{\nabla\psi}/c$ on $\mathcal K_\rho$. Let $\phi\in C^\infty([0,\infty))$
with $\phi=0$ on $[0,c^2/4]$ and $\phi=1$ on $[c^2,\infty)$, let $h(t):=\phi(t)/t$, and
$W:=h(\rho^2)\,\mathrm{Im}(\bar\psi\nabla\psi)$. Since $\sup_t\sqrt t\,h(t)<\infty$, $W\in L^2(Q)$, and $W=\mathrm{Im}(\bar\psi\nabla\psi)/\rho^2$ on
$\mathcal K_\rho$; as $\curl\mathbf G=0$ a.e.\ outside $\mathcal K_\rho$, it suffices to show $\int_QW\cdot\curl\mathbf G=0$. Write $\psi=\psi_1+i\psi_2$
with real $\psi_j$. For smooth $\psi$ one computes
$\curl W=2h'(\rho^2)\rho^2\nabla\psi_1\times\nabla\psi_2+2h(\rho^2)\nabla\psi_1\times\nabla\psi_2=2\phi'(\rho^2)\nabla\psi_1\times\nabla\psi_2$. Approximate
$\psi$ in $H^1(Q)$ and almost everywhere by $\psi^{(n)}\in C^\infty(\bar Q)$ and integrate by parts; the boundary term
$\int_{\partial Q}(W^{(n)}\times\mathbf G)\cdot n=-\int_{\partial Q}W^{(n)}\cdot(n\times\mathbf G)$ vanishes. As $n\to\infty$, $W^{(n)}\to W$ in $L^2$
(bounded, a.e.\ convergent coefficients times $L^2$-convergent gradients) and $\curl W^{(n)}\to2\phi'(\rho^2)\nabla\psi_1\times\nabla\psi_2$ in $L^1$,
while $\mathbf G\in L^\infty$. Hence $\int_QW\cdot\curl\mathbf G=\int_Q2\phi'(\rho^2)\,\mathbf G\cdot(\nabla\psi_1\times\nabla\psi_2)=0$, because
$\phi'(\rho^2)=0$ on $\mathcal K_\rho$ and $\mathbf G=0$ a.e.\ outside $\mathcal K_\rho$.
\end{proof}

Neither a continuous phase of $\psi$ nor its continuity is used; in the application below the vortices of
$\psi$ lie in the region where $\mathbf G$ vanishes.

\begin{proposition}[Neumann floor from the vorticity]\label{prop:floor}
There is an absolute constant $c_\star>0$ such that the following holds. If $\kappa\ge1$ and
$\norm\varpi_{L^2(Q)}\le\kappa\sqrt N/s$, then for every $\eta\in(0,1]$
\[
  \mathcal X\ \ge\ (1-\eta)\min\{N,\ c_\star\eta s/\kappa\}.
\]
\end{proposition}

\begin{proof}
We show this with $c_\star:=\min\{9/(12800S_*^3),\,1/(36864C_A)\}$.

\emph{Weak fields.} If $N\le N_1:=\eta s/(16S_*^3\kappa)$, Lemma~\ref{lem:L6} gives
$\norm{a'}_3^2\le\sqrt N\cdot2S_*\kappa\sqrt N/s$, so $S_*^2\norm{a'}_3^2\le2S_*^3\kappa N/s\le\eta/8$, and
Lemma~\ref{lem:weakfloor} with $b=a'$ and gauge invariance give $\mathcal X\ge(1-\eta)N$.

\emph{Large fields.} Let $N>N_1$ and suppose $\mathcal X<c_\star\eta s/\kappa$. As $\eta\le1\le\kappa$, $9/12800<1/128$ and
$S_*\ge1/\pi$, this gives $\mathcal X<9N_1/800<9N/800$, $\mathcal X<s/(128S_*^3\kappa)$, $\mathcal X<s/(36864C_A\kappa)$ and
$\mathcal X<9\pi^3s/12800<\pi^2s/2$. By gauge invariance $E:=E_0^N(Q;a')=\mathcal X/\abs Q$; let $\psi$ be
a normalized ground state for $a'$ (it exists by the direct method: since $\norm{a'v}_2\le\norm{a'}_6\norm v_3$ with $a'\in L^6(Q)$ by Lemma~\ref{lem:L6}, minimizing sequences are bounded in $H^1(Q)$,
which embeds compactly in $L^2(Q)$ and $L^3(Q)$), and let $\rho=\abs\psi\in H^1(Q)$.
(1) The diamagnetic inequality $\abs{\nabla\rho}\le\abs{(\nabla+ia')\psi}$ a.e.\ \cite[Theorem~7.21]{LLAna}
and the Poincar\'e inequality $\norm{\rho-\bar m}^2\le\frac{s^2}{\pi^2}\norm{\nabla\rho}^2$ give $\abs Q\bar m^2>\frac12$, where $\bar m$
is the mean of $\rho$.
(2) $\abs{\{\rho<\bar m/2\}}<512S_*^6\mathcal X^3$: on this set $\abs{\rho-\bar m}>\bar m/2$, so
Chebyshev's inequality, \eqref{eq:Sstar}, $\norm{\nabla\rho}^2\le\mathcal X/\abs Q$ (by (1)) and $\bar m^{-6}<8\abs Q^3$ give the bound
$2^6S_*^6\norm{\nabla\rho}^6\bar m^{-6}<512S_*^6\mathcal X^3$.
(3) The cutoff $\zeta:=\min\{1,(6(\frac12-\rho/\bar m))_+\}$ satisfies $\norm{\nabla\zeta}^2<72\mathcal X$:
$\nabla\zeta=-6\bar m^{-1}\nabla\rho$ a.e.\ on $\{\bar m/3<\rho<\bar m/2\}$ and $\nabla\zeta=0$ a.e.\ elsewhere, so
$\norm{\nabla\zeta}^2\le36\bar m^{-2}\norm{\nabla\rho}^2<72\abs Q\cdot\mathcal X/\abs Q=72\mathcal X$ by (1).
(4) On $\mathcal U:=\{\rho>\bar m/4\}$ define the current $v:=\mathrm{Im}(\bar\psi(\nabla+ia')\psi)/\rho^2$; then
$E\ge\int_{\mathcal U}\rho^2\abs v^2>\norm v^2_{L^2(\mathcal U)}/(32\abs Q)$.
(5) Let $\mathbf G:=(1-\zeta)\mathbf g_0\in H^1(Q)\cap L^\infty(Q)$ and $\mathcal K_\rho:=\{\rho\ge\bar m/3\}\subset\mathcal U$. Then $\mathbf G$ and
$\curl\mathbf G$ vanish a.e.\ outside $\mathcal K_\rho$ (as $\zeta=1$ there), and $n\times\mathbf G=0$ on the faces in the trace sense
(Lemma~\ref{lem:dual}(b)). Since $a'$ is real, $v-a'=\mathrm{Im}(\bar\psi\nabla\psi)/\rho^2$ on $\mathcal U$, so Lemma~\ref{lem:current} with
$c=\bar m/3$ gives $\ip v{\curl\mathbf G}=\ip{a'}{\curl\mathbf G}$.
(6) On $Q$, $\curl\mathbf g_0=a'$ (Lemma~\ref{lem:dual}(a)), so $\curl\mathbf G=a'-\mathbf r$ with
$\mathbf r:=\zeta a'+\nabla\zeta\times\mathbf g_0$. Since $0\le\zeta\le1$ vanishes outside $A:=\{\rho<\bar m/2\}$, H\"older's
inequality, Lemma~\ref{lem:L6} and (2) give
$\norm{\zeta a'}^2\le\norm{a'}_6^2\abs A^{2/3}<\frac{4S_*^2\kappa^2N}{s^2}\cdot64S_*^4\mathcal X^2<\frac N{64}$, as $\mathcal X<s/(128S_*^3\kappa)$.
Lemma~\ref{lem:dual}(a) gives $\Phi^{*2}\le8C_A\kappa N/s$, so (3) and $\mathcal X<s/(36864C_A\kappa)$ give
$\norm{\nabla\zeta\times\mathbf g_0}^2\le\Phi^{*2}\norm{\nabla\zeta}^2<72\Phi^{*2}\mathcal X<N/64$. Hence
$\norm{\mathbf r}<\sqrt N/4$, so $\ip{a'}{\curl\mathbf G}\ge N-\sqrt N\norm{\mathbf r}>\frac34N$ and
$\norm{\curl\mathbf G}<\frac54\sqrt N$. As $\curl\mathbf G$ is supported in $\mathcal K_\rho\subset\mathcal U$, (5) and the
Cauchy--Schwarz inequality give
$\norm v_{L^2(\mathcal U)}\ge\ip v{\curl\mathbf G}/\norm{\curl\mathbf G}>\frac35\sqrt N$, and (4) gives
$\mathcal X>\norm v^2_{L^2(\mathcal U)}/32>9N/800$, a contradiction. Hence $\mathcal X\ge c_\star\eta s/\kappa$.
\end{proof}

\subsection{Proof of Theorem~\ref{thm:neumann}}\label{sec:proof-neumann}
We take $c':=c_\star/2$ with $c_\star$ as in
Proposition~\ref{prop:floor}. Fix $B$ with $\Upsilon_Q(B)\le\Upsilon$ and $u'_Q>0$ (if $u'_Q=0$ both sides vanish).
Theorem~\ref{thm:sup} with $D=D_\Upsilon$ gives $\norm{\curl B}_{L^2(Q)}\le34D_\Upsilon^6s^{-1}(u'_Q)^{1/2}=2K_Vs^{-1}(u'_Q)^{1/2}$. For
$a=\nu eB$ the projected field is $\nu ea'$, $N=\nu^2e^2u'_Q$, and $\norm{\curl a}_{L^2(Q)}\le2K_V\sqrt N/s$.
Proposition~\ref{prop:floor} with $\kappa=2K_V\ge1$ gives
$\mathcal X\ge(1-\eta)\min(N,c_\star\eta s/(2K_V))=(1-\eta)\min(N,N_b)$. \qed

%% file: sec8_proof.tex
\section{Proof of the lower bound}\label{sec:proof}

\subsection{Localization}
Let $n_e>0$, $\mathcal C_e=(-\frac{n_e}2,\frac{n_e}2)^3$, and let $h^\infty_{n_e}$ be the operator of
Section~\ref{sec:periodization} with the electron confined to $\mathcal C_e$ by Dirichlet conditions,
$h_{\mathrm{el}}[a]=\frac\lambda2(p+ea)^2$, $\mathfrak k=L^2(\mathcal C_e)$. The first lemma relates it to the form $\mathfrak q$ of
Remark~\ref{rem:def} for $H^{(\lambda)}$ (band $\abs k<\pi$, kinetic coefficient $\lambda$).

\begin{lemma}[The band model in the Schr\"odinger representation]\label{lem:Qform}
For all $\Psi\in C_c^\infty(\mathcal C_e)\otimes_{\mathrm{alg}}\mathcal F_{\mathrm{fin}}(C_c^\infty)$, extended by zero in the electron variable,
\[
  \mathfrak q(\Psi)\ \ge\ \inf\sigma(h^\infty_{n_e})\,\norm\Psi^2 .
\]
\end{lemma}

\begin{proof}
(i) Let $\mathfrak h_{<\pi}:=L^2(\{\abs k<\pi\}\times\{1,2\})$ and $\mathfrak h_{\ge\pi}$ its orthogonal complement, so that
$\mathcal F=\mathcal F_{\mathrm b}(\mathfrak h_{<\pi})\otimes\mathcal F_{\mathrm b}(\mathfrak h_{\ge\pi})$ and
$\Hf=d\Gamma(\omega_{<\pi})\otimes1+1\otimes d\Gamma(\omega_{\ge\pi})$; the field $A_\pi$ acts on the first factor only.
(ii) The map $(\iota f)(k):=\sum_r\eps_r(k)f(k,r)$ is an isometry from $\mathfrak h_{<\pi}$ onto the transverse fields in the
space $\mathfrak h=L^2(\mathcal K;\C^3)$ of Section~\ref{sec:periodization}, and $\Gamma(\iota)$ maps $\mathcal F_{\mathrm b}(\mathfrak h_{<\pi})$ isometrically
onto the transverse Fock space tensored with the longitudinal vacuum. Since $\iota^*\iota=1$, $\omega$ commutes with $P_\perp$
and $P_\perp=\sum_r\eps_r\eps_r^{\mathsf T}$, $\Gamma(\iota)$ intertwines $d\Gamma(\omega)$ and intertwines $A_{\pi,i}(x)$ with $A_i(x)=\Phi(f^i_x)$ on
finite particle vectors. (iii) Since $Jf^i_x=f^i_x$, under the Wiener--It\^o--Segal isomorphism $U$ of the commuting family
$\{\Phi(h):Jh=h\}$ (see \cite[Theorems~5.53 and 5.54, Section~5.27]{Ara18}) the operator $A_i(x)$ becomes multiplication by the Gaussian field $A_{i,q}(x)$, which has
a $C^1$ version on $\bar{\mathcal C}_e$ (Remark~\ref{rem:smooth}). (iv) Every $\Psi\in C_c^\infty(\mathcal C_e)\otimes_{\mathrm{alg}}\mathcal F_{\mathrm{fin}}(C_c^\infty)$ is a
finite sum $\sum_k\varphi_k\otimes F_k\otimes G_k$ with $\varphi_k\in C_c^\infty(\mathcal C_e)$ and finite particle vectors $F_k$, $G_k$ of the two
factors. Let $\Psi_q:=\sum_k\varphi_k\,(U\Gamma(\iota)F_k)(q)\,G_k$. By (ii), (iii) and Fubini's theorem,
$\frac\lambda2\norm{(p+eA_\pi)\Psi}^2=\int d\mu(q)\,\frac\lambda2\norm{(p+eA_q)\Psi_q}^2$, and for each $q$ the vector $\Psi_q$ lies in
$H^1_0(\mathcal C_e)\otimes\mathcal F_{\mathrm b}(\mathfrak h_{\ge\pi})$. Hence $\mathfrak q(\Psi)$ is the value at the image of $\Psi$ of the closed form of
$h^\infty_{n_e}\otimes1+1\otimes d\Gamma(\omega_{\ge\pi})$, which is at least $\inf\sigma(h^\infty_{n_e})\norm\Psi^2$.
\end{proof}

\begin{lemma}[IMS localization]\label{lem:IMS}
In the setting above,
\[
  E_0(H^{(\lambda)})\ \ge\ \inf\sigma(h^\infty_{n_e})-\frac{3\pi^2\lambda}{n_e^2}.
\]
\end{lemma}

This is the IMS localization formula (Ismagilov, Morgan, Simon, Sigal) \cite[Section~3.1, Theorem~3.2]{CFKS87}.

\begin{proof}
Fix a real $j\in C_c^\infty((-\frac12,\frac12))$ with $\int j^2=1$ and $\norm{j'}^2\le2\pi^2$ (a smooth approximation of $\sqrt2\cos(\pi t)$, for which $\norm{j'}^2=\pi^2$), and let $J(x):=n_e^{-3/2}\prod_ij(x_i/n_e)$. Then
$J\in C_c^\infty(\mathcal C_e)$, $\int J^2=1$ and $\norm{\nabla J}^2\le6\pi^2/n_e^2$. With $J_y:=J(\cdot-y)$,
$\int J_y(x)^2dy=1$ and $\int\abs{\nabla J_y(x)}^2dy=\norm{\nabla J}^2$ for every $x$. For $\Psi\in\mathcal D_0$,
$J_y\Psi\in\mathcal D_0$, $J_y$ commutes with the field, and
$(p+eA)(J_y\Psi)=J_y(p+eA)\Psi-i(\nabla J_y)\Psi$. Expanding pointwise, the cross term
$2\,\mathrm{Re}\ip{J_y(p+eA)\Psi}{-i(\nabla J_y)\Psi}$ integrates to zero in $y$, since
$\int J_y\nabla J_y\,dy=\frac12\nabla_x\int J_y^2\,dy=0$ (Fubini). Hence, with
$\ip{J_y\Psi}{HJ_y\Psi}$ meaning the form value,
\[
  \ip\Psi{H\Psi}=\int\ip{J_y\Psi}{HJ_y\Psi}\,dy-\frac\lambda2\norm{\nabla J}^2\norm\Psi^2 .
\]
The electron support of $J_y\Psi$ is a compact subset of $y+\mathcal C_e$. The joint
translation $\mathrm e^{iy\cdot(p+P_{\mathrm f})}$ of electron and field commutes with $H$ and moves it into
$C_c^\infty(\mathcal C_e)\otimes_{\mathrm{alg}}\mathcal F_{\mathrm{fin}}(C_c^\infty)$ (as $\mathrm e^{iy\cdot P_{\mathrm f}}$ multiplies each photon wave function by $\mathrm e^{iy\cdot k}$); hence, by Lemma~\ref{lem:Qform},
$\ip{J_y\Psi}{HJ_y\Psi}\ge\inf\sigma(h^\infty_{n_e})\norm{J_y\Psi}^2$, and $\int\norm{J_y\Psi}^2dy=\norm\Psi^2$. Since
$\mathcal D_0$ is a form core (Remark~\ref{rem:def}), the lemma follows.
\end{proof}

\subsection{Proof of Theorem~\ref{thm:main}(b)}
\begin{proposition}[Relative error]\label{prop:relerr}
There are absolute constants $s_0$ and $C$ such that the following holds. Let $\lambda=e^2\ell^7$, $\Xi=e^2\ell^2$,
let $s\ge s_0$ be an integer, $\nu\in\N$, $\eta\in(0,1)$ and $M\ge1$ with $s\le M\ell$, let
$n_e:=2s\lceil M\ell/(2s)\rceil\in[M\ell,3M\ell]$, and let $K_V$ be as in Proposition~\ref{prop:configBO}, so that
$K_V\le C(1+s+\log(1+\lambda e^2))^6$. If $\nu^{-2}+\eta+C\log(2+s)/s\le1$, then
\[
  1-\frac{E_0(H^{(\lambda)})}{F_\lambda[\beta]}\ \le\ \frac1{\nu^2}+\eta+\frac{C\log(2+s)}{s}+C\Bigl(\frac{M^3}{\Xi}+\frac1{M^2}+\frac{\nu^2M^3s^2K_V}{\eta\,\ell^2}\Bigr).
\]
\end{proposition}

\begin{proof}
Step 1 (localization and periodization). Lemma~\ref{lem:IMS} and Theorem~\ref{thm:periodization}(a) (example (i),
$C_W=\lambda e^2/2$) give $E_0(H^{(\lambda)})\ge\limsup_L\inf\sigma(h^L)-\frac{3\pi^2\lambda}{n_e^2}$ along odd $L$, since the
quantitative form of Theorem~\ref{thm:periodization}(a) holds for every $L\ge L_0$.

Step 2 (block Born--Oppenheimer). Since $n_e\in2s\N$, the box $\mathcal C_e$ has corners in $\Z^3$ and side in $s\N$.
On $\T_L$ with $L$ odd and $L\ge\max(n_e+8s,s^2,L_{\mathrm G}(s))$, Remark~\ref{rem:Qrep} gives
$h^L=\Hf+\int^\oplus h_{\mathrm{el}}[B_q]\ge H_{\mathrm{BO}}$, and Corollary~\ref{thm:blockBO} gives
\[
  \inf\sigma(h^L)\ \ge\ \min\Bigl(1-\frac1{\nu^2},\ (1-\eta)^{1/2}\Bigl(1-\frac{C\log(2+s)}{s}\Bigr)_+\Bigr)F_\lambda[\beta]
  -\Bigl(\frac{3\sqrt3\pi}{2}+C\Bigr)n_e^3-\frac{\pi^2}{2}\,\frac{\nu^2e^2n_e^3s^3}{N_b},
\]
which does not depend on $L$.

Step 3 (relative size). With $F_\lambda[\beta]=c_Fe^2\ell^5$ and $N_b=c'\eta s/K_V$,
\[
  \frac{3\pi^2\lambda}{n_e^2\,c_Fe^2\ell^5}\le\frac{3\pi^2}{c_FM^2},\qquad \frac{n_e^3}{c_Fe^2\ell^5}\le\frac{27M^3}{c_F\Xi},\qquad
  \frac{\nu^2e^2n_e^3s^3}{N_b\,c_Fe^2\ell^5}\le\frac{27\nu^2M^3s^3}{c_FN_b\ell^2}\le\frac{27\nu^2M^3s^2K_V}{c_Fc'\eta\ell^2};
\]
finally $(1-\eta)^{1/2}\ge1-\eta$ and $(1-x)(1-y)\ge1-x-y$.
\end{proof}

Every term on the right of Proposition~\ref{prop:relerr} is small in a suitable order of limits: first
the parameters $\nu,\eta,s,M$ are fixed so that the first four terms are small, then $\Xi$ and $\ell$ are taken
large. The only term in which the coupling enters beyond $\Xi$ is the last one, through
$\log(1+\lambda e^2)=\log(1+e^4\ell^7)$, and this is the origin of the condition on $\log(2+e^2)$ in
Theorem~\ref{thm:main}(b).
The order of the choices and what each quantity may depend on are listed in Table~\ref{tab:choices}.
\begin{table}[htbp]
\caption{Order of the choices in the proof of Theorem~\ref{thm:main}(b). Each quantity may depend only on those in earlier rows.}\label{tab:choices}
\centering
\renewcommand{\arraystretch}{1.2}
\small\begin{tabular}{p{0.4cm}p{4.6cm}p{4.0cm}p{4.4cm}}
\hline
 & quantity & depends on & where\\
\hline
0 & $C_1$, $c_\star$, $c'$, $C_w$, $\theta_0=\frac12$, $s_0$, $C$ & nothing (absolute) & Sections~\ref{sec:blockBO}--\ref{sec:neumann}, Proposition~\ref{prop:relerr}\\
1 & $\eps\in(0,1)$ & given & Theorem~\ref{thm:main}(b)\\
2 & $s\ge s_0$, $\nu$, $\eta$, $M$ & $\eps$ & proof of Theorem~\ref{thm:main}(b)\\
3 & $\Xi_0$, $\kappa_0$, $\ell_0\ (\ge s/M)$ & $\eps$, through step 2 & proof of Theorem~\ref{thm:main}(b)\\
4 & $\alpha,\Lambda$ (hence $e,\lambda,\ell,\Xi$) & given; $\ell\ge\ell_0$, $\Xi\ge\Xi_0$, $\log(2+e^2)\le\kappa_0\ell^{1/3}$ & Theorem~\ref{thm:main}(b)\\
5 & $n_e$, $\Upsilon$, $K_V$, $N_b$, $u_b$ & determined by steps 2 and 4, not chosen & Propositions~\ref{prop:configBO}, \ref{prop:relerr}, Theorem~\ref{thm:neumann}\\
6 & odd $L\to\infty$, $\delta\downarrow0$ & all of the above & Step~1 of the proof of Proposition~\ref{prop:relerr}; not in the bound\\
\hline
\end{tabular}
\end{table}
The gradient condition of Theorems~\ref{thm:neumann} and~\ref{thm:sup} is paid as follows. Blocks with
$\int\mathrm w_Q\abs{\nabla B_q}^2\le\Upsilon u'_Q(q)$ use Theorem~\ref{thm:neumann}; its loss $K_V$ enters, through $N_b$, the
last term of Proposition~\ref{prop:relerr}. For the other blocks the weight term $c_\theta\mathcal E_Q\ge c_\theta\Upsilon u'_Q=\gamma u'_Q/\abs Q$
replaces the electron bound (Step~3 of Section~\ref{sec:proof-abstractBO}); its cost $\mathcal D_\theta\le Cn_e^3$ (Lemma~\ref{lem:weight},
Proposition~\ref{prop:configBO}(ii)) is part of the term $(\frac{3\sqrt3\pi}2+C)n_e^3$ in Step~2 of the proof of Proposition~\ref{prop:relerr}, that is, of $CM^3/\Xi$. For fixed $\alpha$
the conditions of step 4 hold for all large $\Lambda$ (Corollary~\ref{cor:limit}(i)).

\begin{proof}[Proof of Theorem~\ref{thm:main}(b)]
Let $\eps\in(0,1)$. Fix an integer $s\ge s_0$, $\nu\in\N$, $\eta\in(0,1)$ and $M\ge1$ such that
$\nu^{-2}+\eta+C\log(2+s)/s+CM^{-2}<\eps/2$; these depend only on $\eps$. Let $\Xi_0:=8CM^3/\eps$, so that
$CM^3/\Xi<\eps/4$ for $\Xi\ge\Xi_0$. Since $1+e^4\ell^7\le(2+e^2)^2(1+\ell)^7$,
\[
  K_V\ \le\ C\bigl(1+s+7\log(1+\ell)+2\log(2+e^2)\bigr)^6\ \le\ C_\eps\bigl((\log(2+\ell))^6+(\log(2+e^2))^6\bigr),
\]
so if $\log(2+e^2)\le\kappa_0\ell^{1/3}$, the last term of Proposition~\ref{prop:relerr} is at most
$C'_\eps((\log(2+\ell))^6\ell^{-2}+\kappa_0^6)$, which is $<\eps/4$ for $\kappa_0=\kappa_0(\eps)$ small and $\ell\ge\ell_0(\eps)$ large (with also
$s\le M\ell$). Then $E_0(H^{(\lambda)})\ge(1-\eps)F_\lambda[\beta]$. Return to physical units by Lemma~\ref{lem:scaling} and
\eqref{eq:coeff}; finally $\ELL\ge E_0$.
\end{proof}

\subsection{Partition of the energy}
\begin{corollary}[Virial-type partition]\label{cor:virial}
Let $\mathcal T:=\frac12(p+eA_\Lambda)^2$. In the limit of Corollary~\ref{cor:limit}, let $\Psi$ be normalized states in the
form domain of $H$ with $\ip\Psi{H\Psi}\le(1+\eps')E_0$ and $\eps'\to0$. Then
$\ip\Psi{\mathcal T\Psi}/E_0\to\frac57$ and $\ip\Psi{\Hf\Psi}/E_0\to\frac27$.
\end{corollary}

\begin{proof}
For $c>0$ let $H_c:=c\mathcal T+\Hf$; for $c\in[\frac12,\frac32]$ its form domain is that of $H$. By
Lemma~\ref{lem:scaling} with $a=c$, $\inf\sigma(H_c)=c^{-1}E_0(\alpha,c\Lambda)$. The hypotheses of Corollary~\ref{cor:limit}
hold for $(\alpha,c\Lambda)$ whenever they hold for $(\alpha,\Lambda)$ ($\ell$ and $\Xi$ change by the factors $c^{1/7}$ and $c^{2/7}$), so
$\inf\sigma(H_c)=(1+o(1))c^{5/7}E_0(\alpha,\Lambda)$ for each fixed $c\in[\frac12,\frac32]$. Fix $\delta\in(0,\frac12]$. Since
$\inf\sigma(H_c)\le\ip\Psi{H_c\Psi}\le(1+\eps')E_0+(c-1)\langle\mathcal T\rangle$, the choices $c=1\pm\delta$ give
$\pm\delta\langle\mathcal T\rangle\ge[(1\pm\delta)^{5/7}-1-o(1)-\eps']E_0$, with $E_0>0$ by \eqref{eq:LL-bounds}. Hence
$\limsup\abs{\langle\mathcal T\rangle/E_0-\frac57}\le\delta^{-1}\max_\pm\abs{(1\pm\delta)^{5/7}-1\mp\frac57\delta}\le C\delta$, and $\delta\downarrow0$ gives the claim
for $\mathcal T$. Finally $\langle\Hf\rangle=\langle H\rangle-\langle\mathcal T\rangle$.
\end{proof}

This is the standard scaling argument; the existence of a ground state of the fiber at zero
momentum \cite{HaSi23} is not needed.

\begin{remark}[Where the new input lies]\label{rem:new}
Lieb and Loss \cite{LL00} left open the order of the lower bound. In our proof the new input
for the lower bound lies in two places: the combination of phase multiplication and the
Hardy inequality in field space (Proposition~\ref{prop:configBO}), and the Neumann lower bound for
band-limited fields under a gradient condition (Theorem~\ref{thm:neumann}). The price of the
latter is the polylogarithmic factor $K_V$, which appears as the last term in Proposition~\ref{prop:relerr}.
It must grow at most like a power of $\log(1+\lambda e^2)$: a factor $K_V\gtrsim(\lambda e^2)^{c}$ would make the last term
diverge already for fixed $\alpha$ if $c>\frac27$.
\end{remark}

%% file: sec9_remarks.tex
\section{Remarks and extensions}\label{sec:remarks}

Remarks~\ref{rem:fiberP} and~\ref{rem:BHF} are consequences of Theorem~\ref{thm:main}; the
discussion in the other subsections is heuristic.

\subsection{Dimension}
For the same model on $\R^d$, the local density approximation at scale $R$ balances the
kinetic energy $R^{-2}$ against $\beta_d\norm\phi_1\propto e\Lambda^dR^{d/2}$, which gives
\[
  E_{\mathrm{LDA}}\asymp e^{\frac{4}{d+4}}\Lambda^{\frac{4d}{d+4}},\qquad
  E_{\mathrm{vac}}=\frac{e^2}{2}\langle\Omega,A(0)^2\Omega\rangle\asymp e^2\Lambda^{d-1},\qquad
  \frac{E_{\mathrm{vac}}}{E_{\mathrm{LDA}}}\asymp e^{\frac{2(d+2)}{d+4}}\Lambda^{\frac{d^2-d-4}{d+4}} .
\]
The exponent $\frac{d^2-d-4}{d+4}$ equals $-\frac13$ for $d=2$, $+\frac27$ for $d=3$ and $+1$ for $d=4$.
For $d=2$ the vacuum trial state beats the local density approximation as $\Lambda\to\infty$,
so the exponent $\frac{12}{7}$ of Lieb and Loss is a phenomenon of dimension $d\ge3$. On the
side of the proof, the infrared bound $\int_{\abs k<r}\abs k^{-2}d^dk$ in
Theorem~\ref{thm:periodization} is finite only for $d\ge3$. For $d\ge4$, Sections~\ref{sec:upper}--\ref{sec:blockBO}
should carry over, but not verbatim: the weight $(1+\abs\xi^2)^2$ in the proof of Proposition~\ref{prop:coherent} and the exponent $4$ of the weights $\mathrm w_Q$ of Section~\ref{sec:finite-model} are chosen for $d=3$ and must be raised (above $d/2$ and above $d$, respectively), at the price of more regularity of $\varphi$ and of different constants in Section~\ref{sec:neumann}; moreover the curl and reflection
constructions of Section~\ref{sec:neumann} (Lemmas~\ref{lem:reflection} and~\ref{lem:dual}) would have to
be rewritten in terms of differential forms; we do not pursue this.

\subsection{Smooth cutoffs}
For a cutoff $\chi(\abs k/\Lambda)$ the exact computation for a uniform density suggests the leading
coefficient $\beta\langle\chi\rangle$ with $\langle\chi\rangle:=3\int_0^\infty\chi(t)t^2dt$. The lower bound of
Section~\ref{sec:blockBO} uses that $\Pi'$ is a projection: with $\chi(\abs k/\Lambda)\Pi'$ in place of $\Pi'$ the trace of the
Gram matrix of a block is about $\langle\chi^2\rangle\rhom\abs Q$, which is smaller than $\langle\chi\rangle\rhom\abs Q$ unless $\chi$ is an
indicator function, so a smooth cutoff would need a different argument; our main theorem is restricted to the sharp cutoff.

\subsection{Non-zero total momentum}
\begin{remark}\label{rem:fiberP}
Let $H\cong\int^\oplus H(P)\,dP$ be the fiber decomposition \cite{LLP53} (see \cite[Theorem~2.3]{LMS07} for the Pauli--Fierz model), $E(P):=\inf\sigma(H(P))$, so
that $E(0)=E_0$ \cite{LMS07}. Then $E(P)\ge\inf\sigma(H)=E(0)$ by the direct integral and the
continuity of $E(P)$ \cite[Theorem~4.1(i)]{LMS07}, and $E(P)\le E(0)+\frac{P^2}{2}$ by
\cite[Theorem~4.1(ii)]{LMS07}. Together with Theorem~\ref{thm:main} this gives $1\le E(P)/E_0\le1+P^2/(2E_0)$, hence $E(P)/E_0\to1$
uniformly for $\abs P=o(\alpha^{1/7}\Lambda^{6/7})$; we do not pursue fibers with larger total momentum here.
\end{remark}

\subsection{The Bogoliubov--Hartree--Fock energy}
\begin{remark}\label{rem:BHF}
The BHF functionals of \cite{BBM26,BHe26} are variational over quasi-free states of the
fiber at a fixed total momentum $P$ ($P=0$ in \cite{BBM26}), so $E_{\mathrm{BHF}}\ge E(P)\ge E(0)=E_0$ by their definition, Remark~\ref{rem:fiberP} and
\cite{LMS07}. Theorem~\ref{thm:main} therefore gives
$E_{\mathrm{BHF}}\ge(1-o(1))\,c_\alpha\alpha^{2/7}\Lambda^{12/7}$.
A matching upper bound $E_{\mathrm{BHF}}\lesssim\Lambda^{12/7}$ is not obtained here. Compared with the minimum of order $\Lambda^{3/2}$ of the convex
functional $\mathcal G$ obtained by removing a non-convex term \cite[Theorem~1.9]{BBM26}, the removed non-convex part
must contribute at least $c\,\alpha^{2/7}\Lambda^{12/7}$ at the minimizer of $\mathcal G$, for large $\Lambda$.

Normalization: in \cite[(1.30)--(1.34)]{BBM26} the Hamiltonian is
$\frac12(i\nabla_x+g_{\mathrm{BBM}}A(x))^2+H_{\mathrm{ph}}$ with $A(x)=a^*(G_x)+a(G_x)$,
$G(k)=\abs k^{-1/2}\eps(k,\pm)$. Comparing with \eqref{eq:A}, $e=\sqrt2(2\pi)^{3/2}g_{\mathrm{BBM}}$, that is
$e^2=16\pi^3g_{\mathrm{BBM}}^2$ and $\alpha=4\pi^2g_{\mathrm{BBM}}^2$ (the sign of $A$ is changed by $\Gamma(-1)$).
\cite{BBM26} fixes an infrared cutoff $\sigma>0$ independently of $\Lambda$. The passage to the
model $H^\sigma$ with photon momenta $\sigma\le\abs k\le\Lambda$ needs one trial state: for
$\Psi\otimes\Omega_{<\sigma}$, with $\Psi$ in a core of $H^\sigma$ and $\Omega_{<\sigma}$ the vacuum of the low modes (the
modes with $\abs k>\Lambda$ are also in the vacuum and do not contribute; the product lies in the
form domain $D(\abs p)\cap D(\Hf^{1/2})$), the cross terms vanish since
$\ip\Omega{A_{<\sigma}(x)\Omega}=0$, and
$E_0(\alpha,\Lambda)\le\inf\sigma(H^\sigma)+\frac{e^2}{2}\ip\Omega{A_{<\sigma}(0)^2\Omega}=\inf\sigma(H^\sigma)+\frac{e^2\sigma^2}{8\pi^2}$.
The fiber energy $E^\sigma(P):=\inf\sigma(H^\sigma(P))$ is an infimum of quadratic functions of $P$,
hence upper semicontinuous, so $\inf\sigma(H^\sigma(0))\ge\mathop{\mathrm{ess\,inf}}_PE^\sigma(P)=\inf\sigma(H^\sigma)$.
The fiber Hamiltonians \cite[(1.40), (1.54)]{BBM26} lack the factor $\frac12$ of (1.30), but
$(P_{\mathrm f}+g_{\mathrm{BBM}}A(0))^2+H_{\mathrm{ph}}\ge\frac12(P_{\mathrm f}+g_{\mathrm{BBM}}A(0))^2+H_{\mathrm{ph}}$, so with either
reading $E_{\mathrm{BHF}}\ge\inf\sigma(H^\sigma(0))$. Consequently, for fixed $\sigma$ and $\Lambda\to\infty$
(Corollary~\ref{cor:limit}(i)),
$E_{\mathrm{BHF}}\ge(1-o(1))\,c_\alpha\alpha^{2/7}\Lambda^{12/7}-\frac{e^2\sigma^2}{8\pi^2}$ for the model of \cite{BBM26} itself.
\end{remark}

\subsection{Spin}
For spin $\frac12$, $\frac12(\sigma\cdot(p+eA))^2=\frac12(p+eA)^2+\frac e2\sigma\cdot B$. Lieb and Loss state
\cite[Theorem~1.7]{LL00} that the ground state energy with this kinetic energy is at most
$C\alpha^{1/2}\Lambda^{3/2}$ (the proof is referred there to a forthcoming paper). If so, the Zeeman
term lowers the energy below the order $\Lambda^{12/7}$ of the spinless model at fixed $\alpha$, and
the leading term of Theorem~\ref{thm:main} does not describe electrons with spin. In the proof,
condition \eqref{eq:W} and the electron part of Proposition~\ref{prop:configBO} (Lemma~\ref{lem:phase}
concerns the magnetic form only) would change. We make no claim about this case.

%% file: appA_BH.tex
\section{Comparison with the work of Bach and Hach}\label{app:BH}

We compare with \cite{BH22}; all equation numbers refer to the published version. The
arXiv version \cite{BH20} has the same equations, numbered by section in Roman numerals
(for example (6.5), (7.10), (1.10) correspond to (VI.5), (VII.10), (I.10)), except that the
main statement (I.10) and its restatement (VIII.3) have the denominator $F_1\alpha^{2/7}\Lambda^{12/7}$ without the factor
$(4/(9\pi))^{2/7}$ of (1.10), while the proof, (VIII.11) and (VIII.16), uses the denominator $F[\beta(\alpha,\Lambda)]=(4/(9\pi))^{2/7}F_1\alpha^{2/7}\Lambda^{12/7}$. The
coefficient $\beta(\alpha,\Lambda)$ and the normalization of the field are the same in both versions. The abstract of the published version \cite{BH22} states the result in the form of (I.10), $\abs{\ELL/(F_1\alpha^{2/7}\Lambda^{12/7})-1}\le C\alpha^{4/105}\Lambda^{-4/105}$, without the factor $(4/(9\pi))^{2/7}$ of (1.10) and (8.3). In both versions Theorems~6.1 and~7.2 are stated for all $\alpha>0$, while the proof of the upper bound in (1.10) assumes $\alpha\in(0,1]$ (after (8.11)).
Our arguments in the main text use only the definition of the functional $\mathcal F_\beta$ from
\cite{BH22}, not its theorems.

\subsection{Normalization}
The vector potential of \cite[(2.3)]{BH22} is
\[
  A_{\mathrm{BH}}=(2\pi)^{-3/2}\int\eps\abs k^{-1/2}(a^*\mathrm e^{-ikx}+a\mathrm e^{ikx})\,dk=\sqrt2A
\]
with $A$ as in \eqref{eq:A},
and the coupling is $\sqrt{\alpha_{\mathrm{BH}}}A_{\mathrm{BH}}$. Hence the Hamiltonian of \cite{BH22} with
coupling $\alpha_{\mathrm{BH}}$ is our $H$ with $e^2=2\alpha_{\mathrm{BH}}$, and
\[
  E^{\mathrm{BH}}_{\mathrm{LL}}(\alpha_{\mathrm{BH}},\Lambda)=\ELL\big|_{e^2=2\alpha_{\mathrm{BH}}},\qquad
  E^{\mathrm{BH}}_0(\alpha_{\mathrm{BH}},\Lambda)=E_0\big|_{e^2=2\alpha_{\mathrm{BH}}}
\]
(the sign of $A$ is changed by $\Gamma(-1)$). In this normalization our leading coefficient is
$\beta_*=\frac{\sqrt{2\alpha_{\mathrm{BH}}}\Lambda^3}{6\pi^2}=0.0239\sqrt{\alpha_{\mathrm{BH}}}\Lambda^3$.
The conventions are compared in Table~\ref{tab:BH-conventions}.
\begin{table}[htbp]
\caption{Conventions of this paper and of \cite{BH22} (equation numbers of \cite{BH22}).}\label{tab:BH-conventions}
\centering
\renewcommand{\arraystretch}{1.2}
\small\begin{tabular}{p{2.5cm}p{5.0cm}p{5.8cm}}
\hline
 & this paper & \cite{BH22}\\
\hline
Fourier transform & $\hat f(k)=\int\mathrm e^{-ik\cdot x}f(x)\,dx$ (Section~\ref{sec:model}) & $(2\pi)^{-3/2}\int\mathrm e^{-ik\cdot x}f(x)\,dx$, (3.30)\\
CCR & $[a(k,r),a^*(k',r')]=\delta_{rr'}\delta(k-k')$ & the same, (2.4)--(2.5)\\
vector potential & $A_\Lambda$, \eqref{eq:A}, kernel $(2\pi)^{-3/2}(2\abs k)^{-1/2}$ & $A_{\mathrm{BH}}=\sqrt2A_\Lambda$, (2.3), kernel $(2\pi)^{-3/2}\abs k^{-1/2}$\\
Hamiltonian & $\frac12(p+eA_\Lambda)^2+\Hf$ & $\frac12(i\nabla+\sqrt{\alpha_{\mathrm{BH}}}A_{\mathrm{BH}})^2+\Hf$, (2.1); thus $e^2=2\alpha_{\mathrm{BH}}$\\
$\Theta$ & --- & $\alpha_{\mathrm{BH}}P_C\chi_{\sigma,\Lambda}M_{\phi^2}\chi_{\sigma,\Lambda}P_C$ ($M$: multiplication in position space) by (1.17) and (3.22)--(3.30)\\
$X$ & --- & $X(A)=\Tr(\sqrt{\abs k^2+A}-\abs k)$, (1.19)\\
trace density & $\Tr(PM_\phi P)=\frac{\Lambda^3}{3\pi^2}\norm\phi_1$, $P=\chi_{0,\Lambda}P_C$ & printed as $\frac{8\pi}{3}\Lambda^3\norm\phi_1$ in (6.6)\\
coefficient of $\norm\phi_1$ & $\beta_*=\frac{\sqrt{2\alpha_{\mathrm{BH}}}}{6\pi^2}\Lambda^3$ & $\beta(\alpha,\Lambda)=\sqrt{4\alpha_{\mathrm{BH}}/(9\pi)}\,\Lambda^3=(2\pi)^{3/2}\beta_*$, (5.6)\\
\hline
\end{tabular}
\end{table}
The energies of \cite{BH22} are infima over (product) states in the form domain $H^1(\R^3)\otimes D(N_{\mathrm{ph}}^{1/2})$ of \cite[(2.1)]{BH22}, ours over the form domain of Remark~\ref{rem:def}. The infima coincide: $\mathcal D_0\subset H^1(\R^3)\otimes D(N_{\mathrm{ph}}^{1/2})$, this domain lies in our form domain since $\Hf\le\Lambda N_{\mathrm{ph}}$ on the band $\abs k<\Lambda$, and by (T2) of Section~\ref{sec:periodization} the form, also restricted to product states $\phi\otimes\Psi$, is continuous in the form norm of $\frac12p^2+\Hf+1$, for which $C_c^\infty(\R^3)$ and $\mathcal F_{\mathrm{fin}}(C_c^\infty)$ are cores; hence both infima equal the infimum over $\mathcal D_0$ (for $\ELL$, over products of vectors from these cores).

\subsection{The leading coefficient}
When quoting \cite{BH22} we write $\alpha$ for its coupling constant $\alpha_{\mathrm{BH}}$ (in this subsection and in Proposition~\ref{prop:A2}); this $\alpha$ is not the fine structure constant of \eqref{eq:H-intro}. \cite{BH22} sets $\beta(\alpha,\Lambda):=\sqrt{4\alpha/(9\pi)}\,\Lambda^3=0.376\sqrt\alpha\Lambda^3$ (after (1.6), and (5.6)) and
states the main result, Theorem~1.1, (1.10), in terms of $F_{\beta(\alpha,\Lambda)}$. This is
$(2\pi)^{3/2}=15.75$ times the coefficient of the previous subsection. Our coefficient is
confirmed independently by the exact computation for a uniform density: in a box of
volume $V$ with $\rho\equiv1/V$, each mode has frequency $\sqrt{\omega_j^2+e^2/V}$; for
$e^2/V\gg\Lambda^2$ the field energy is $\approx\frac12N\sqrt{e^2/V}$ with
$N=\rhom V$ modes and $\rhom=2\cdot\frac{4\pi}{3}\Lambda^3(2\pi)^{-3}=\Lambda^3/(3\pi^2)$, in agreement with the
coefficient of Proposition~\ref{prop:coherent} and of the Gram trace in Corollary~\ref{cor:hardy}.

\emph{Our reading of the discrepancy} (see Section~\ref{sec:intro}, Scope). The identity
$(\hat\phi*)^*(\hat\phi*)=\mathcal F\phi^2\mathcal F^*$ before \cite[(6.5)]{BH22} acquires a factor $(2\pi)^3$ on the right
for the unitary Fourier transform, and the trace
``$2(\mathrm{Vol}[B(0,\Lambda)]-\mathrm{Vol}[B(0,\sigma)])\norm\phi_1$'' in \cite[(6.6)]{BH22} needs the phase space factor
$(2\pi)^{-3}$. Together these give $(2\pi)^{3/2}$, and the final step (7.10) has the same
feature. The definition of the operator $\Theta_{\phi,\alpha}$ in \cite[(1.17), (3.22)--(3.29)]{BH22}, with the
unitary Fourier convention (3.30), is $\alpha P_C\chi_{\sigma,\Lambda}M_{\phi^2}\chi_{\sigma,\Lambda}P_C$ ($M$ denoting multiplication in
position space) and is correct: the kernel $(2\pi)^{-3/2}\widehat{\phi^2}(k-k')$ in (3.26) is exactly the
kernel of $M_{\phi^2}$, and this definition gives $\Tr(PM_\phi P)=\frac{\Lambda^3}{3\pi^2}\norm\phi_1$ and the
coefficient $\frac{\sqrt{2\alpha}}{6\pi^2}\Lambda^3$. The expression
$\frac{\alpha}{(2\pi)^3}P_C\chi_{\sigma,\Lambda}\phi(x)^2\chi_{\sigma,\Lambda}P_C$ in (1.20), with $\phi(x):=\mathcal F\phi\mathcal F^*$ as in Section~6
of \cite{BH22}, has an extra factor $(2\pi)^{-3}$, inherited from the identity before (6.5).
Inequality (7.9) is written with (1.20); since the true $\Theta$ is $(2\pi)^3$ times larger and $X$ is
monotone, the lower bound obtained from (7.9) is correct but weaker. What makes (7.8)
fail is the bound on the trace $\Tr[P_C\chi\phi(x)\chi P_C]$ in (7.10), with the same missing
$(2\pi)^{-3}$ as in (6.6). These places are collected in Table~\ref{tab:BH-places}.
\begin{table}[htbp]
\caption{Places in \cite{BH22} affected by the factor $(2\pi)^{3/2}$, as printed and with the factors of the normalization above (published numbering, $\alpha=\alpha_{\mathrm{BH}}$).}\label{tab:BH-places}
\centering
\renewcommand{\arraystretch}{1.2}
\small\begin{tabular}{p{2.3cm}p{5.4cm}p{5.6cm}}
\hline
place in \cite{BH22} & as printed & with the factors of the Normalization\\
\hline
before (6.5) & $(\hat\phi*)^*(\hat\phi*)=\mathcal F\phi^2\mathcal F^*$ & $(\hat\phi*)^*(\hat\phi*)=(2\pi)^3\mathcal F\phi^2\mathcal F^*$\\
(1.20), (7.9) & $\Theta_{\phi,\alpha}=\frac{\alpha}{(2\pi)^3}P_C\chi_{\sigma,\Lambda}\phi(x)^2\chi_{\sigma,\Lambda}P_C$ & $\Theta_{\phi,\alpha}=\alpha P_C\chi_{\sigma,\Lambda}\phi(x)^2\chi_{\sigma,\Lambda}P_C$, as in (1.17), (3.22)--(3.30)\\
(6.6), (7.10) & $\Tr[P_C\chi_{\sigma,\Lambda}\phi(x)\chi_{\sigma,\Lambda}P_C]=\frac{8\pi}{3}(\Lambda^3-\sigma^3)\norm\phi_1$ & $\Tr[P_C\chi_{\sigma,\Lambda}\phi(x)\chi_{\sigma,\Lambda}P_C]=\frac{\Lambda^3-\sigma^3}{3\pi^2}\norm\phi_1$\\
main term of (6.3), (7.8) & $\sqrt{\frac{4\alpha}{9\pi}}(\Lambda^3-\sigma^3)\norm\phi_1$ & $\frac{\sqrt{2\alpha}}{6\pi^2}(\Lambda^3-\sigma^3)\norm\phi_1$ (for (7.8) Proposition~\ref{prop:A2}; for (6.3) the trace $\Tr(PM_\phi P)$ above; the error terms of (6.3) and (7.8) are not discussed)\\
(1.10), (8.3) & denominator $(\frac{4}{9\pi})^{2/7}F_1\alpha^{2/7}\Lambda^{12/7}$ & leading term $(\frac{\sqrt2}{6\pi^2})^{4/7}F_1\alpha^{2/7}\Lambda^{12/7}$ for fixed $\alpha$ (Corollary~\ref{cor:limit} with $e^2=2\alpha$)\\
\hline
\end{tabular}
\end{table}

\subsection{Two statements}
\begin{proposition}\label{prop:A1}
For the model of \cite{BH22} without infrared cutoff ($\sigma=0$) and fixed $\alpha_{\mathrm{BH}}>0$,
\[
  \limsup_{\Lambda\to\infty}\frac{E^{\mathrm{BH}}_{\mathrm{LL}}(\alpha_{\mathrm{BH}},\Lambda)}{F_1\,\alpha_{\mathrm{BH}}^{2/7}\Lambda^{12/7}}\ \le\
  \Bigl(\frac{\sqrt2}{6\pi^2}\Bigr)^{4/7}=0.118353\ldots\ <\ \Bigl(\frac{4}{9\pi}\Bigr)^{2/7}=0.571918\ldots
\]
Hence the lower inequality in \cite[Theorem~1.1, (1.10)]{BH22} does not hold as stated for large
$\Lambda$: the upper limit is at most $0.20694$ times the leading term of (1.10).
\end{proposition}

\begin{proof}
By the normalization above and Theorem~\ref{thm:main}(a),
\[
  \ELL|_{e^2=2\alpha_{\mathrm{BH}}}\le(1+o(1))(6\pi^2)^{-4/7}F[1](2\alpha_{\mathrm{BH}})^{2/7}\Lambda^{12/7},
\]
and
$(6\pi^2)^{-4/7}2^{2/7}=(\sqrt2/(6\pi^2))^{4/7}$; moreover $F_1=F[1]$ (same definition). Only the upper
bound is used.
\end{proof}

\begin{proposition}\label{prop:A2}
The inequality \cite[Theorem~7.2, (7.8)]{BH22},
\[
  \tfrac12X(2\Theta_{\phi,\alpha})\ \ge\ \sqrt{\tfrac{4\alpha}{9\pi}}\bigl(\Lambda^3-\sigma^3\bigr)\norm\phi_1-\frac{C\alpha^{1/4}\Lambda^{1/2}(L\Lambda+1)^3}{L^{3/2}}\sqrt{\norm\phi_1},
\]
fails for $\sigma=0$, any fixed $0\le\phi\in C^1_c$ with $\phi\not\equiv0$ supported in $B(0,L)$, any fixed
sufficiently large $\Lambda$, and $\alpha\to\infty$. On the other hand, with $P$ as below,
\[
  \tfrac12X(2\Theta_{\phi,\alpha})\ \ge\ \frac{\sqrt{2\alpha}}{6\pi^2}\bigl(\Lambda^3-\sigma^3\bigr)\norm\phi_1-\Lambda^{1/2}(2\alpha)^{1/4}\Tr(PM_\phi P)^{1/2};
\]
we do not discuss whether the last term can be bounded by the error term of (7.8).
\end{proposition}

\begin{proof}
Let $P:=\chi_{0,\Lambda}P_C$ be the projection onto transverse fields with momenta in the ball of
radius $\Lambda$, so that $\Theta=\alpha PM_{\phi^2}P$ (previous subsection).

(1) $X(A)\le\Tr A^{1/2}$. Indeed, by \cite[Lemma~5.2, (5.21)]{BH22},
$X(A)=\Tr[A^{1/2}((k^2+A)^{1/2}+\abs k)^{-1}A^{1/2}]$, and $B:=(k^2+A)^{1/2}+\abs k\ge A^{1/2}$ by operator
monotonicity of the square root. For $\eps>0$, $(B+\eps)^{-1}\le(A^{1/2}+\eps)^{-1}$, hence
$A^{1/2}(B+\eps)^{-1}A^{1/2}\le A^{1/2}$; let $\eps\downarrow0$ (monotone limit) and take traces.

(2) $\Tr(2\Theta)^{1/2}=\sqrt{2\alpha}\Tr(PM_{\phi^2}P)^{1/2}$. First, if $0\le\rho\le\rho'$ are bounded multiplication
operators and $P\rho'P$ is compact, then $0\le P\rho P\le P\rho'P$, so by the min-max principle the eigenvalues
satisfy $\mu_j(P\rho P)\le\mu_j(P\rho'P)$ and $\Tr(P\rho P)^{1/2}\le\Tr(P\rho'P)^{1/2}$. Given $\eps>0$, choose
$0\le\phi'\in C_c^\infty$ with $\phi'\ge\phi$ and $\norm{\phi'}_1\le\norm\phi_1+\eps$ (for instance
$\phi'=(\phi*j_\delta+\omega_\phi(\delta))\chi$ with a mollifier $j_\delta$, the modulus of continuity $\omega_\phi$ of $\phi$, a cutoff
$0\le\chi\le1$ equal to $1$ near $\supp\phi$, and $\delta$ small). Scale $x\mapsto(\Lambda/\pi)x$ to the band $\pi$: with
$R:=\Lambda/\pi$ and $\varphi_R:=\phi'(\cdot/R)$ (Proposition~\ref{prop:coherent} is homogeneous of degree one in $\varphi$),
Proposition~\ref{prop:coherent}(a) with a fixed width $t$ such that $1+C/t\le\frac32$ gives
$\Tr(PM_{\phi^2}P)^{1/2}\le\Tr(\Pi'\varphi_R^2\Pi')^{1/2}\le\frac32\cdot\frac\pi3R^3\norm{\phi'}_1+O(R^2)$, since the error term is
$O((R+t)^3R^{-1})$ for fixed $\phi'$ and $t$. As $\frac\pi3R^3=\frac{\Lambda^3}{3\pi^2}$, for small $\eps$ and large $\Lambda$,
$\Tr(PM_{\phi^2}P)^{1/2}\le2\frac{\Lambda^3}{3\pi^2}\norm\phi_1$, and for all
$\alpha>0$, $\frac12X(2\Theta)\le\frac{2\sqrt2}{6\pi^2}\sqrt\alpha\Lambda^3\norm\phi_1\approx0.0478\sqrt\alpha\Lambda^3\norm\phi_1$. The leading
term on the right of (7.8) is $0.376\sqrt\alpha\Lambda^3\norm\phi_1$, and the error term is of order
$\alpha^{1/4}$; this is a contradiction as $\alpha\to\infty$.

(3) Since $PM_{\phi^2}P-(PM_\phi P)^2=PM_\phi(1-P)M_\phi P\ge0$ and $X$ is monotone, \cite[Lemma~7.1,
(7.1)]{BH22} with $p=\frac12$ gives
$\frac12X(2\Theta)\ge\frac{\sqrt{2\alpha}}{2}\Tr(PM_\phi P)-\Lambda^{1/2}(2\alpha)^{1/4}\Tr(PM_\phi P)^{1/2}$, and
\[
  \Tr(PM_\phi P)=2\cdot\frac{4\pi}{3}\Lambda^3(2\pi)^{-3}\norm\phi_1=\frac{\Lambda^3}{3\pi^2}\norm\phi_1
\]
(with $\Lambda^3-\sigma^3$ in place of $\Lambda^3$ if $\sigma>0$).
\end{proof}

\begin{remark}
Proposition~\ref{prop:A1} compares Theorem~1.1 of \cite{BH22} with our upper bound only;
Proposition~\ref{prop:A2} examines Theorem~7.2 of \cite{BH22} directly; it uses
Proposition~\ref{prop:coherent} but not Theorem~\ref{thm:main}. The asymptotics of \cite{BH22}
have been quoted in later work in the form $\ELL\sim F_1\alpha^{2/7}\Lambda^{12/7}$ of (I.10) of
\cite{BH20} and of the abstract of \cite{BH22}, that is, without the factor $(4/(9\pi))^{2/7}$, for example in \cite[(1.7)]{BBM26}
and \cite[(I.12)]{BHe26} (where the constant is denoted by $F$).
\end{remark}

\subsection{Relation of the results}
\cite[Theorem~1.1]{BH22} states the asymptotics of the product state energy $\ELL$, with the
power rate $\Lambda^{-4/105}$ and the coefficient $\beta(\alpha,\Lambda)$ of (5.6) there; by
Proposition~\ref{prop:A1} its lower inequality cannot hold with this coefficient, while its
upper inequality is consistent with Theorem~\ref{thm:main}(a). Corollary~\ref{cor:limit} gives the asymptotics of both $\ELL$ and the true ground
state energy $E_0$ (with the coefficient of the first subsection); the exponent $\frac{12}{7}$ and
the form $F[\beta]$ of the leading term agree with \cite{BH22}. What is new is the lower bound for
$E_0$, and hence for $\ELL$, with the coefficient $\beta_*$: the leading asymptotics of both
energies, $E_0/\ELL\to1$, and the uniformity in the window $\alpha^{5/2}\Lambda\to\infty$.

%% file: appB_constants.tex
\section{The leading constant}\label{app:constants}

\subsection{The minimizer of the comparison functional}
\begin{proposition}[Minimizer and constant]\label{prop:F1}
Let $x_1\in(\pi,\frac32\pi)$ be the first positive root of $\tan x=x$.
\begin{enumerate}[label=(\alph*)]
\item $F[1]$ is attained, by the case of equality in Nash's inequality \cite{CL93} (see the
paragraph after the proposition); moreover, by \cite[Theorems~7--9]{Hac19}, the symmetric decreasing
rearrangement of a minimizer is unique. Up to the scaling $\phi\mapsto t^{3/2}\phi(t\,\cdot)$ and
normalization the minimizer obtained from Nash's inequality is
\[
  u(x)=\Bigl(1-\frac{\sin\abs x}{\abs x\cos x_1}\Bigr)\mathbf 1_{\abs x\le x_1}.
\]
The function $u$ is non-negative, $C^1$, supported in $\bar B(0,x_1)$, and is the radial
solution of $-\Delta u+1=u$ in $B(0,x_1)$ with $u=u'=0$ on the
boundary, which forces $\tan x_1=x_1$; the normalized minimizer $u_*:=t_*^{3/2}u(t_*\cdot)/\norm u_2$,
with $t_*$ as in the proof of (b), satisfies the Euler--Lagrange equation $-\Delta u_*+1=t_*^2u_*$ on its
support $\bar B(0,x_1/t_*)$. (Whether the minimizer itself is unique up to
translations is not used.)
\item One has
\[
  F[1]=\frac73\Bigl(\frac3{10}\Bigr)^{3/7}\Bigl(\frac{3\pi}{10}\Bigr)^{2/7}x_1^{6/7}
  =\frac74\Bigl(\frac23\Bigr)^{3/7}C_{\mathrm N}^{-3/7},\qquad
  C_{\mathrm N}=\frac{5^{5/3}}{2^{2/3}\cdot3\,x_1^2\,(4\pi/3)^{2/3}} ,
\]
where $C_{\mathrm N}$ is the sharp constant in the three-dimensional Nash inequality
$\norm\phi_2^{10/3}\le C_{\mathrm N}\norm{\nabla\phi}_2^2\norm\phi_1^{4/3}$ \cite[Theorem, (1)--(2)]{CL93}.
\end{enumerate}
\end{proposition}

The second expression for $F[1]$ is \cite[(1.29)]{BH22}, where it is derived from the
sharp Nash inequality of Carlen and Loss \cite{CL93}: Nash's inequality gives
$F[1]\ge\min_y(\frac{y^{-4/3}}{2C_{\mathrm N}}+y)$, with equality because the Nash
optimizer is $u$. We verify the first expression directly from the minimizer.

\begin{proof}[Proof of the first identity in (b)]
Write $f(r)=\sin r/r$, so that $(r^2f')'=-r^2f$ and $f'(x_1)=0$, and use
$\cos x_1=\sin x_1/x_1$. Then
$\int_0^{x_1}u\,r^2dr=\frac{x_1^3}{3}-\frac{\sin x_1-x_1\cos x_1}{\cos x_1}=\frac{x_1^3}{3}$,
$\int_0^{x_1}\abs{u'}^2r^2dr=\sec^2x_1\int_0^{x_1}\sin^2r\,dr=\frac{x_1}{2}\tan^2x_1=\frac{x_1^3}{2}$
(integrate $\int r^2f'^2$ by parts), and $\int_0^{x_1}u^2r^2dr=\frac{x_1^3}{3}+\frac{x_1^3}{2}=\frac56x_1^3$.
Hence $a:=\frac12\norm{\nabla u}^2/\norm u_2^2=\frac3{10}$ and
$b:=\norm u_1/\norm u_2=\frac{4\pi}{3}\sqrt{\frac{3}{10\pi}}\,x_1^{3/2}$. For the normalized
dilate $u_t:=t^{3/2}u(t\cdot)/\norm u_2$ one has $\mathcal F_1(u_t)=at^2+bt^{-3/2}$, minimized at $t=t_*$, $t_*^{7/2}=\frac{3b}{4a}=\norm u_2$,
with value $\frac73a^{3/7}(\frac34b)^{4/7}$.
\end{proof}

The function $u_*$ is $C^{1,1}=W^{2,\infty}$ ($u$ is analytic inside its support, $u=u'=0$ on the boundary,
and $u''$ jumps there), which is what Remark~\ref{rem:upper-rate} uses.

\subsection{Numerical values}
The following enclosures were obtained with interval arithmetic (mpmath, 200 bits); the
script is the ancillary file \texttt{constants\_interval.py} (also available at
\url{https://github.com/tabayashi0117/pauli-fierz-uv-limit-code}) and runs in less than a second.
\begin{align*}
  x_1\ (\text{first positive root of }\tan x=x)&\in[4.49340945790906417530788092727,\\
  &\qquad\ 4.49340945790906417530788092729],\\
  F[1]\ (\text{both expressions in Prop.~\ref{prop:F1}(b)})&\in[4.964615752497504363152959214,\\
  &\qquad\ 4.964615752497504363152959215],\\
  c_e:=(6\pi^2)^{-4/7}F[1]&\in[0.48201050176893950665949395426,\\
  &\qquad\ 0.48201050176893950665949395427],\\
  c_\alpha=(4\pi)^{2/7}c_e&\in[0.9933824009729449788048457056,\\
  &\qquad\ 0.9933824009729449788048457057].
\end{align*}
Here $c_e$ is the coefficient of $e^{4/7}\Lambda^{12/7}$ and $c_\alpha$ that of $\alpha^{2/7}\Lambda^{12/7}$; in the
working convention $F_\lambda[\beta]=c_Fe^2\ell^5$ with $c_F=(\pi/6)^{4/7}F[1]$. The vacuum trial state gives
$\frac{e^2}{2}\langle A_\Lambda(0)^2\rangle_\Omega=\frac{e^2\Lambda^2}{8\pi^2}=\frac{\alpha\Lambda^2}{2\pi}$, and the crossover ratio is
$\frac{\alpha\Lambda^2/(2\pi)}{c_\alpha\alpha^{2/7}\Lambda^{12/7}}=\frac{(\alpha^{5/2}\Lambda)^{2/7}}{2\pi c_\alpha}$ with $2\pi c_\alpha=6.2416\ldots$; in terms of $e$ the same ratio is $(e^5\Lambda)^{2/7}/(8\pi^2c_e)$ with $8\pi^2c_e=38.058\ldots$

\subsection{Absolute constants of the lower bound}
All constants of the lower bound are absolute: they do not depend on the coupling or on the side
$L$ of the periodic box (beyond the stated lower bounds on $L$), and the block side $s$ enters only as
stated. They are not needed numerically, since the main theorem has no explicit rate; each is defined where it
first appears, and those of the Neumann lower bound in Section~\ref{sec:neumann}.

\subsection{Conventions}
The three normalizations used in this paper are compared in Table~\ref{tab:conventions}.
\begin{table}[htbp]
\caption{Physical and working conventions, and the convention of \cite{BH22}.}\label{tab:conventions}
\centering
\renewcommand{\arraystretch}{1.2}
\small\begin{tabular}{p{2.3cm}p{3.0cm}p{1.6cm}p{2.1cm}p{4.4cm}}
\hline
convention & kinetic energy & cutoff & $\ell$ & leading term\\
\hline
physical (main theorem) & $\frac12(p+eA)^2$ & $\abs k\le\Lambda$ & $(\Lambda/(\pi e^2))^{1/7}$ & $c_\alpha\alpha^{2/7}\Lambda^{12/7}$\\
working (Sections~\ref{sec:upper}--\ref{sec:proof}) & $\frac\lambda2(p+eA)^2$, $\lambda=\Lambda/\pi$ & $\abs k<\pi$ & $(\lambda/e^2)^{1/7}$ & $c_Fe^2\ell^5$ (times $\Lambda/\pi$)\\
\cite{BH22} & $\frac12(i\nabla+\sqrt{\alpha_{\mathrm{BH}}}A_{\mathrm{BH}})^2$, $A_{\mathrm{BH}}=\sqrt2A$ & $\sigma\le\abs k<\Lambda$ & --- & $(\sqrt2/(6\pi^2))^{4/7}F[1]\alpha_{\mathrm{BH}}^{2/7}\Lambda^{12/7}$\\
\hline
\end{tabular}
\end{table}

%% file: appD_counterexample.tex
\section{The gradient condition cannot be dropped}\label{app:counterexample}

Theorem~\ref{thm:neumann} is stated only for fields with $\Upsilon_Q(B)\le\Upsilon$. We show that this
condition cannot be removed: the lower bound \eqref{eq:F} fails for suitable band-limited
fields. The construction concentrates vorticity in a small ball inside $Q$ and uses a
trial function that avoids the ball.

Let $Q$ be a cube of side $s\ge4$, $P'_Q$ the projection in $L^2(Q;\R^3)$ onto the orthogonal
complement of $\nabla H^1(Q)$, and $\hat\mu(b):=\norm{P'_Q(\mathbf 1_Qb)}^2/\abs Q$. A band-limited field is a real
field of the form $\Pi'f$ with $f\in L^2(\R^3;\R^3)$.

\begin{lemma}[Density]\label{lem:density}
The set $\{P'_Q(\mathbf 1_Q\Pi'f):f\in L^2(\R^3;\R^3)\}$ is dense in $\Ran P'_Q$. This remains true if $f$ is
restricted to real Schwartz fields whose Fourier transforms lie in
$C_c^\infty(\{0<\abs k<\pi\})$.
\end{lemma}

\begin{proof}
Let $g\in\Ran P'_Q$ be orthogonal to all $P'_Q\mathbf 1_Q\Pi'f$. Since $P'_Qg=g$, we get
$\ip{\mathbf 1_Qg}{\Pi'f}=0$ for all $f$, i.e.\ $\Pi'(\mathbf 1_Qg)=0$. As $\mathbf 1_Qg$ has compact support,
$\widehat{\mathbf 1_Qg}$ is entire, and its transverse component vanishes on $0<\abs k<\pi$:
$k\times\widehat{\mathbf 1_Qg}(k)=0$ there. The left side is entire and vanishes on an open set, hence
identically, so $\curl(\mathbf 1_Qg)=0$ as a distribution on $\R^3$. A compactly supported curl free
$L^2$ field is a gradient $\nabla\varphi$, so on $Q$, $g\in\nabla H^1(Q)$; since $g\perp\nabla H^1(Q)$, $g=0$. For the
second statement note that this family of $f$ is dense in $\Ran\Pi'$ (the sphere and the
origin have measure zero) and $P'_Q\mathbf 1_Q$ is bounded. Real fields suffice because $\Pi'$ and
$P'_Q$ commute with complex conjugation.
\end{proof}

\begin{proposition}\label{prop:counterexample}
Fix $\nu\in\N$, $\eta\in[0,1)$ and $e,u_b>0$.
(a) There is a band-limited real Schwartz field $B$ with
\[
  E_0^N\bigl(Q;\nu eB\bigr)\ <\ \nu^2(1-\eta)\,\frac{e^2\min\{u'_Q(B),u_b\}}{\abs Q}.
\]
(b) There is $L_0$ such that for every integer $L\ge L_0$ and every position of $Q$ in $\T_L$
there is a field built from the physical modes ($0<\abs k<\pi$, transverse) satisfying the same
inequality.
Moreover, the ratio $E_0^N/\hat\mu$ can be made arbitrarily small.
\end{proposition}

\begin{proof}
Let $y_b:=\nu^2e^2u_b/\abs Q$ and $0<Y\le y_b$.

(1) Target field. Let $c$ be the centre of $Q$, $0<r\le\frac12$, $\phi\in C_c^\infty(B_r(c))$ real and
non-zero, and $v_*:=\curl(\phi\,\mathbf e_3)$, normalized so that $\norm{v_*}^2=Y\abs Q$. Since $v_*$ is divergence
free and supported in the interior of $Q$, $\int_Qv_*\cdot\nabla\chi=0$ for all $\chi\in H^1(Q)$; so
$v_*\in\Ran P'_Q$ and $\hat\mu(v_*)=Y$.

(2) Trial function. Let $\zeta\in C_c^\infty(B_{2r}(c))$ with $0\le\zeta\le1$, $\zeta=1$ on $B_r(c)$ and
$\abs{\nabla\zeta}\le2/r$, and $\psi:=1-\zeta$ (real). Then $\int_Q\abs{\nabla\psi}^2\le\frac4{r^2}\abs{B_{2r}}=\frac{128\pi r}{3}$ and
$\int_Q\psi^2\ge\abs Q-\abs{B_{2r}}\ge\abs Q/2$ (as $r\le\frac12$, $s\ge4$).

(3) Upper bound. Let $b$ be a bounded field on $Q$ and $b':=P'_Q(\mathbf 1_Qb)$, so $b-b'=\nabla\chi$; by
gauge invariance $E_0^N(Q;b)=E_0^N(Q;b')$ (Section~\ref{sec:floor}; for smooth $B$, $\nabla\chi\in L^6$ by
Lemma~\ref{lem:L6}). Since $\psi$ is real, the cross term vanishes, and as $v_*\psi=0$,
\[
  E_0^N(Q;b')\ \le\ \frac{\int\abs{\nabla\psi}^2+\int\abs{b'}^2\psi^2}{\int\psi^2}\ \le\ \frac2{\abs Q}\Bigl(\frac{128\pi r}{3}+\norm{b'-v_*}^2\Bigr).
\]

(4) Choice. Take $r:=\min(\frac12,\frac{3(1-\eta)Y\abs Q}{1024\pi})$, so that the first term is
$\le\frac{(1-\eta)Y}{4}$. By Lemma~\ref{lem:density} choose a band-limited Schwartz field $B$ with
$\norm{P'_Q(\mathbf 1_Q\nu eB)-v_*}^2\le\epsilon\abs Q$. Then $E_0^N(Q;\nu eB)\le\frac{(1-\eta)Y}{2}+2\epsilon$, while
$\hat\mu(\nu eB)\ge(\sqrt Y-\sqrt\epsilon)^2$, so the right side of the claim is at least
$(1-\eta)\min(\hat\mu,y_b)\ge(1-\eta)(\sqrt Y-\sqrt\epsilon)^2$. For small $\epsilon$ this proves (a). Letting both $r$
and $\epsilon$ tend to zero, $E_0^N(Q;\nu eB)\le\frac{256\pi r}{3\abs Q}+2\epsilon\to0$ while $\hat\mu\to Y$, so the ratio
$E_0^N/\hat\mu$ becomes arbitrarily small.

(5) Periodic box. Periodize the field of (4), $B_L(x):=\sum_{m\in\Z^3}B(x+mL)$. By Poisson summation
the Fourier coefficients of $B_L$ are the values of $\hat B$ on $\frac{2\pi}{L}\Z^3$, so only transverse
modes with $0<\abs k<\pi$ occur. Since $B$ is a Schwartz field, $B_L\to B$ uniformly on $Q$ as
$L\to\infty$, independently of the position of $Q$ (translate $B$). The bound in (3) and $\hat\mu$
are continuous in $L^2(Q)$, so the strict inequality of (4) persists for large $L$.
\end{proof}

\begin{remark}[Structure of the counterexample]
The field is superoscillating: it has large amplitude outside $Q$, while inside $Q$ its
vorticity is concentrated in a small ball. The gradient condition
$\int_{\T_L}\mathrm w_Q\abs{\nabla B}^2\le\Upsilon u'_Q$ penalizes the amplitude outside $Q$ through the weight
$\mathrm w_Q$, which has a positive polynomial tail on the whole space; by
Theorem~\ref{thm:sup} it gives $\norm{\curl B}_{L^2(Q)}\le34D^6(u'_Q)^{1/2}/s$, so that the vorticity cannot
concentrate inside $Q$ on scales much smaller than $s/D^6$, whereas here
$\norm{\curl v_*}=\norm{\nabla v_*}\ge(\pi/r)\norm{v_*}$ (the target field is divergence free and supported in $B_r(c)$).
Configurations violating the condition are handled by the weight $\theta$ in Step~3 of
Section~\ref{sec:proof-abstractBO}; this is why Theorem~\ref{thm:abstractBO} and Corollary~\ref{thm:blockBO} need the
lower bound only under the gradient condition.
\end{remark}

%% file: backmatter.tex
\section*{Statements and Declarations}

\subsection*{Tool and computational resource disclosure}
This work was carried out with substantial assistance from Claude models
(Anthropic), accessed through the Claude Code agent in 2026. The AI system was
used to explore strategies of proof, to write out detailed arguments and
computations, to check them, and to prepare the code in the ancillary files.
The author has read and checked the entire paper and takes full
responsibility for its content. No AI system is an author.

\subsection*{Funding}
This research received no external funding.

\subsection*{Competing interests}
The author has no competing interests to declare that are relevant to the content of this article.

\subsection*{Data and code availability}
No data were generated or analysed in this study.

The code used to certify the numerical constants (Appendix~\ref{app:constants}) is available
as the ancillary file \texttt{anc/constants\_interval.py} of the arXiv submission and at
\url{https://github.com/tabayashi0117/pauli-fierz-uv-limit-code}.